\documentclass[journal]{IEEEtran}

\usepackage{amsmath,amssymb,amsfonts,amsthm,mathtools}
\usepackage{bm}
\usepackage{mathrsfs}
\usepackage{graphicx}
\usepackage{booktabs}
\usepackage{cite}
\usepackage{xcolor}
\usepackage{url}
\usepackage[
    colorlinks=true,
    linkcolor=blue,
    citecolor=blue,
    urlcolor=blue
]{hyperref}

\usepackage{setspace}

\usepackage{stmaryrd}

\usepackage{float}
\usepackage[caption=false,font=footnotesize]{subfig}

\graphicspath{{Ketquapre/}}

\allowdisplaybreaks[3]

\usepackage{float}

\usepackage{titlesec}

\titlespacing*{\section}
{0pt}
{0.75ex plus 0.15ex minus 0.10ex}
{0.35ex plus 0.10ex minus 0.05ex}

\titlespacing*{\subsection}
{0pt}
{0.60ex plus 0.15ex minus 0.10ex}
{0.25ex plus 0.08ex minus 0.05ex}

\titlespacing*{\subsubsection}
{0pt}
{0.45ex plus 0.10ex minus 0.08ex}
{0.20ex plus 0.05ex minus 0.03ex}

\graphicspath{{./KETQUA/}}

\allowdisplaybreaks
\newtheorem{definition}{Definition}
\newtheorem{assumption}{Assumption}
\newtheorem{lemma}{Lemma}
\newtheorem{proposition}{Proposition}
\newtheorem{theorem}{Theorem}

\newtheorem{remark}{Remark}

\DeclareMathOperator{\diag}{diag}
\DeclareMathOperator{\col}{col}
\DeclareMathOperator{\rank}{rank}

\DeclareMathOperator{\sgn}{sgn}
\DeclareMathOperator{\atanh}{atanh}

\newcommand{\R}{\mathbb R}
\newcommand{\T}{^{\top}}
\newcommand{\norm}[1]{\left\lVert #1\right\rVert}
\newcommand{\abs}[1]{\left\lvert #1\right\rvert}
\newcommand{\spow}[2]{\left\lceil #1\right\rfloor^{#2}}

\newcommand{\bx}{\bm x}
\newcommand{\bu}{\bm u}
\newcommand{\ba}{\bm a}
\newcommand{\bd}{\bm d}
\newcommand{\bW}{\bm W}
\newcommand{\bphi}{\bm\phi}

\newcommand{\bomega}{\bm\omega}
\newcommand{\bOmega}{\bm\Omega}
\newcommand{\bGamma}{\bm\Gamma}

\newcommand{\bR}{\mathbf R}
\newcommand{\bS}{\mathbf S}
\newcommand{\bTmat}{\mathbf T}
\newcommand{\bUbar}{\overline{\mathbf U}}

\usepackage{orcidlink}

\usepackage{xcolor}
\usepackage{tikz}
\usepackage{scalerel}
\usepackage{hyperref}

\newcommand{\orcidicon}[1]{%
\href{https://orcid.org/#1}{%
\mbox{%
\scalerel*{%
\begin{tikzpicture}[yscale=-1,transform shape]
\pic{orcidlogo};
\end{tikzpicture}%
}{|}%
}%
}%
}

\begin{document}

\title{Predefined-Time Resilient Integral Reinforcement Learning for
Input-Constrained Unknown Nonlinear Systems Under FDI Attacks and Disturbances:
A Fully Data-Driven Approach}
\author{
Tien Dat Vu~\orcidicon{0009-0004-0792-000X}\,\IEEEmembership{Student Member~IEEE,}
Minh Doan~\orcidicon{0000-0002-9460-4256}\,\IEEEmembership{Member~IEEE}
\thanks{T. D. Vu and M. Doan are with the Faculty of Mechanical Engineering,
Ho Chi Minh City University of Technology (HCMUT), Vietnam National University
Ho Chi Minh City (VNU-HCM), Ho Chi Minh City, Vietnam
(e-mail: dat.vuv@hcmut.edu.vn; minh.doan@hcmut.edu.vn).}

\thanks{Corresponding author: M. Doan (e-mail: minh.doan@hcmut.edu.vn).}
}
\markboth{
Preprint 
}{}

\maketitle

\begin{abstract}
This paper investigates optimal control for nonlinear systems with unknown dynamics, input constraints, disturbances, and adversarial signals. The objective is to develop a learning-based control method that allows the designer to prescribe the desired convergence time in advance. An integral reinforcement-learning framework is proposed to avoid requiring exact knowledge of the system dynamics while ensuring that the control input always satisfies the actuator constraints. Current and recorded data are combined to train the critic without requiring persistent excitation. The learning gain is selected directly from the prescribed convergence deadline. Lyapunov analysis is then used to establish practical predefined-time convergence of the coupled state–critic system in the presence of disturbances and adversarial channels. The effectiveness of the proposed method is validated through the stabilization control of a two-link robot manipulator
\end{abstract}


\section{Introduction}

\IEEEPARstart{O}{ptimal} control of nonlinear systems is challenging
when unknown drift dynamics, actuator limits, false-data-injection
(FDI) attacks, and external disturbances coexist. The associated
Hamilton--Jacobi--Isaacs (HJI) equation generally depends on the
unknown drift, quadratic input penalties do not enforce hard actuator
bounds, and conventional reinforcement learning (RL) mainly provides
asymptotic convergence without an assignable learning or closed-loop
settling time. These issues are especially relevant to cyber--physical
systems (CPSs), whose networked sensing and control channels are exposed
to FDI attacks.

Continuous-time RL and adaptive dynamic programming approximate HJB
solutions online \cite{Vamvoudakis2010,Vrabie2009}, whereas integral
reinforcement learning (IRL) reduces model dependence through
trajectory-integral identities. In particular, \cite{Modares2014}
eliminates the unknown drift from the learning equation and handles
input constraints through a nonquadratic cost, yielding bounded
hyperbolic-tangent policies \cite{AbuKhalaf2005}, while zero-sum games
provide a natural framework for disturbances and adversarial inputs
\cite{Basar1995,Pasqualetti2013,Fawzi2014}. These methods mainly provide
asymptotic or ultimate convergence. Meanwhile, concurrent and composite
learning replace sustained persistent excitation (PE) by current and
recorded or filtered information, enabling parameter convergence from
finite informative data
\cite{ParikhKamalapurkarDixon2019,PanYu2016,PanYu2018}; this is
particularly attractive for RL because continuous probing is
undesirable after satisfactory regulation, yet its integration with
predefined-time finite-window Bellman--Isaacs learning remains largely
unexplored.

Finite- and fixed-time stability strengthen transient guarantees
\cite{Bhat2000,Polyakov2012}. Our recent fixed-time IRL results for
unknown nonlinear systems and multi-agent extensions
\cite{Vu2026FixedTimeIRL,VuEtAl2026FixedTimeMASIRL} provide
initial-condition-independent learning and regulation bounds without
requiring the drift, but these bounds remain gain-induced.
Predefined-time stability instead permits the settling-time bound to be
assigned a priori \cite{SanchezTorres2018,AldanaLopez2019,Jimenez2020}.
The critic-only method in \cite{Kokolakis2025Automatica} explicitly
relates the learning gain to a prescribed deadline, but its pointwise
Bellman error contains the system vector field and its stabilization
conditions prescribe, rather than construct, the required
value-function and running-cost decay structure. Related predefined-time
zero-sum RL \cite{Su2026PredefinedZeroSum} and data-driven value
iteration after dynamics learning \cite{Cai2026DataDrivenPredefined}
likewise do not provide a finite-window Bellman--Isaacs formulation that
directly removes the unknown drift while jointly addressing actuator
constraints, disturbances, FDI attacks, finite-data excitation, and an
assignable closed-loop deadline.

Finite-window IRL also creates a distinct timing issue: after critic
convergence at \(t_c\), the integral residual still contains
pre-convergence data until one complete reinforcement window has
elapsed. Hence, data informativity, critic learning, reinforcement-window
flushing, and state convergence must all enter the same deadline.
Accordingly, this paper develops a unified end-to-end predefined-time
resilient IRL framework for input-constrained nonlinear systems with
unknown drift, FDI attacks, and external disturbances, connecting cost
construction, finite-data learning, critic convergence, and closed-loop
stabilization in one analytical chain from the constrained minimax cost
to the final coupled state--critic deadline.

The contributions are fivefold. First, a constrained minimax cost is
constructed whose state penalty generates the required low- and
high-order predefined-time dissipation, while its nonquadratic input
term preserves actuator bounds by construction. Second, a finite-window
Bellman--Isaacs equation is derived without the unknown drift, together
with a normalized replay-based two-power critic law whose gain is
computed explicitly from an assigned critic deadline. Third, current
integral information and a finite stack of recorded Bellman--Isaacs
data are combined, replacing classical PE by a finite
replay-informativeness condition and removing the need for persistent
excitation after sufficient data have been collected. Fourth, exact
and practical predefined-time critic convergence are established under
this finite-data condition, with FDI attacks and disturbances handled
through the resilient zero-sum formulation. Finally, a mixed
pointwise--integral Lyapunov analysis combines the data-informativity
time, critic-learning time, reinforcement-window length, and
state-convergence time into a complete predefined-time bound for the
coupled state--critic system. Validation is provided on a two-link robot
manipulator under actuator constraints, disturbances, and FDI attacks.


\textbf{Notation:} Throughout this paper, \(\mathbb R\), \(\mathbb R_{\geq0}\), and \(\mathbb R_{>0}\) denote the real, nonnegative real, and positive real numbers, respectively. For \(x\in\mathbb R^n\), \(\|x\|\) denotes the Euclidean norm. For a matrix \(\mathbf A\), \(\mathbf A^\top\), \(\|\mathbf A\|\), \(\lambda_{\min}(\mathbf A)\), and \(\lambda_{\max}(\mathbf A)\) denote its transpose, induced \(2\)-norm, and minimum and maximum eigenvalues, respectively; for symmetric \(\mathbf A\), \(\mathbf A>0\) and \(\mathbf A<0\) denote positive and negative definiteness. The symbols \(\mathbf I_n\), \(\operatorname{diag}\{\cdot\}\), and \(\operatorname{col}\{\cdot\}\) denote the identity matrix, diagonal concatenation, and column stacking. For \(V\in C^1\), \(\nabla V:=\partial V/\partial x\) denotes the column gradient. For \(s\in\mathbb R\) and \(r>0\), define \(\spow{s}{r}:=|s|^r\sgn(s)\), with \(\spow{0}{r}:=0\), so that \(\frac{d}{ds}\bigl(|s|^{r+1}/(r+1)\bigr)=\spow{s}{r}\); for vectors, \(\spow{\cdot}{r}\), \(\tanh(\cdot)\), \(\sgn(\cdot)\), and scalar powers act componentwise. For \(z>0\), \(\Gamma_{\mathrm E}(z):=\int_{0}^{\infty}t^{z-1}e^{-t}\,dt\), while for \(a,b>0\), \(\mathrm B(a,b):=\int_{0}^{1}t^{a-1}(1-t)^{b-1}\,dt=\int_{0}^{\infty}t^{a-1}(1+t)^{-(a+b)}\,dt=\Gamma_{\mathrm E}(a)\Gamma_{\mathrm E}(b)/\Gamma_{\mathrm E}(a+b)\).

\section{Preliminaries and Problem Formulation }
\label{sec:problem}
\begingroup
\setlength{\jot}{1.5pt}
\setlength{\abovedisplayskip}{3pt}
\setlength{\belowdisplayskip}{3pt}


\subsection{Preliminaries}

\begin{theorem}[Normalized predefined-time comparison]
\label{thm:normalized_predefined_time}
Let \(V:\mathbb R^n\to\mathbb R_{\geq0}\) be positive definite, radially unbounded, and locally absolutely continuous along the trajectories of a forward-complete system. Let \(\alpha,\beta,p,q,r,T_s>0\) satisfy
\begin{equation}
0<p<q,\qquad pr<1<qr,
\label{eq:normalized_exponent_conditions}
\end{equation}
and define
\begin{equation}
\begin{aligned}
\gamma_{p,q,r}
:={}&
\frac{
\Gamma_{\!E}\!\left(\dfrac{1-pr}{q-p}\right)
\Gamma_{\!E}\!\left(\dfrac{qr-1}{q-p}\right)
}{
\alpha^r\Gamma_{\!E}(r)(q-p)
}
\left(\frac{\alpha}{\beta}\right)^{
\frac{1-pr}{q-p}}.
\end{aligned}
\label{eq:gamma_normalization}
\end{equation}
Here, \(\Gamma_{\!E}(z):=\int_0^\infty t^{z-1}e^{-t}dt\), \(z>0\), is the Euler gamma function. If, for almost all \(t\) with \(V(\bx(t))>0\),
\begin{equation}
\dot V(\bx(t))
\leq
-\frac{\gamma_{p,q,r}}{T_s}
\left[
\alpha V^p(\bx(t))
+
\beta V^q(\bx(t))
\right]^r,
\label{eq:normalized_predefined_dissipation}
\end{equation}
then the origin is globally predefined-time stable and
\begin{equation}
T(\bx_0)\leq T_s,
\qquad
\forall\,\bx_0\in\mathbb R^n.
\label{eq:normalized_settling_conclusion}
\end{equation}
\end{theorem}

\begin{proof}
    See Appendix~\ref{App1}
\end{proof}

\begin{remark}
\label{rem:normalized_comparison_interpretation}
Theorem~\ref{thm:normalized_predefined_time} offers a clear insight: to design for a *predefined* time, one must first achieve *fixed-time* convergence and then adjust the scaling parameters to attain the specific predefined time—essentially allowing us to determine the desired convergence time in advance. It separates the prescribed
time scale from the shape of the two-power dissipation. In particular,
the factor \(1/T_s\) makes the guaranteed settling-time bound scale
linearly with the assigned deadline; replacing \(T_s\) by
\(\eta T_s\), \(\eta>0\), multiplies the required dissipation rate by
\(1/\eta\). The two terms balance at
\(V_\times=(\alpha/\beta)^{1/(q-p)}\), so \(\alpha/\beta\) determines
the transition between the low- and high-power regimes. Moreover, a
common rescaling \((\alpha,\beta)\mapsto(c\alpha,c\beta)\), \(c>0\),
gives \(\gamma_{p,q,r}\mapsto c^{-r}\gamma_{p,q,r}\), while
\([\alpha V^p+\beta V^q]^r\mapsto
c^r[\alpha V^p+\beta V^q]^r\); hence the normalized dissipation
\(\gamma_{p,q,r}[\alpha V^p+\beta V^q]^r/T_s\) is invariant to such a
common scaling. Thus \(T_s\) fixes the convergence deadline, whereas
the ratio \(\alpha/\beta\) shapes how the dissipation is distributed
between the local and large-\(V\) regimes.
\end{remark}

\begin{definition}[\(\mathcal L_2\)-gain]
\label{def:l2_gain}
Consider a nonlinear system driven by an exogenous input \(w\in\mathbb{R}^{m_w}\), with performance output \(z\in\mathbb{R}^{m_z}\). The system is said to have a finite \(\mathcal L_2\)-gain not exceeding \(\gamma>0\) from \(w\) to \(z\) if there exists a nonnegative function \(\beta(\cdot)\), with \(\beta(0)=0\), such that every admissible trajectory satisfies \(\int_{0}^{T}\|z(t)\|^{2}dt\leq\gamma^{2}\int_{0}^{T}\|w(t)\|^{2}dt+\beta(x(0))\) for all \(T\geq0\). In particular, for \(x(0)=0\), this condition reduces to \(\|z\|_{2}\leq\gamma\|w\|_{2}\), so that \(\gamma\) quantifies the worst-case energy amplification from the exogenous input to the regulated output. Equivalently, a storage function \(V(x)\geq0\) certifies this attenuation level whenever its derivative along the system trajectories satisfies \(\dot V+\|z\|^{2}-\gamma^{2}\|w\|^{2}\leq0\). This dissipation inequality provides the standard bridge from \(\mathcal L_2\)-attenuation to a Hamilton--Jacobi inequality and, under an optimal control formulation, to the corresponding Hamilton--Jacobi--Bellman or Hamilton--Jacobi--Isaacs equation.
\end{definition}

Theorem~\ref{rem:normalized_comparison_interpretation} serves as a crucial premise for developing the theory of fixed-time convergence in Sections~\ref{sec:critic} and ~\ref{sec:closed_loop}, whereas Definition~\ref{def:l2_gain} is significant because it is used to formulate the HJI equation in Eqs~\eqref{eq:performance},~\eqref{eq:running_cost}.

\subsection{Secure Saturation-Constrained Differential Game}
\label{subsec:secure_game}

Consider the nonlinear control-affine system
\begin{equation}
\dot{\bx}
=
\bm f(\bx)
+
\bm g(\bx)(\bu+\ba)
+
\bm k(\bx)\bd,
\qquad
\bx(0)=\bx_0,
\label{eq:plant}
\end{equation}
where
\(
\bx\in\mathbb R^n
\),
\(
\bu,\ba\in\mathbb R^m
\),
and
\(
\bd\in\mathbb R^s
\).
The vector \(\bu\) is the minimizing secure input,
\(\ba\) is a matched same-channel FDI attack, and \(\bd\) is an
exogenous disturbance.

\begin{remark}[FDI location and actuator constraint]
The constraint \(|u_i|<\bar u_i\) limits only the secure control
component generated by the controller. Two actuator-side FDI
architectures must be distinguished. If \(\ba\) is a matched malicious
input injected after the constrained action \(\bu\) has been generated,
then the plant receives \(\bu+\ba\) and the model above applies
directly. In this case, \(|a_i|\leq\bar a_i\) implies only
\(|u_i+a_i|\leq\bar u_i+\bar a_i\), not
\(|u_i+a_i|<\bar u_i\). Conversely, if \(\ba\) corrupts the command
before physical actuator saturation, the plant-input channel would
instead contain the componentwise saturated signal
\(\operatorname{sat}(\bu+\ba)\), and the differential game and saddle
policies developed below would generally require reformulation. The
present work adopts the former architecture.
\end{remark}


\begin{assumption}[Matched additive actuator-channel attack]
The FDI signal \(\ba\) acts additively through the matched plant-input
channel after the controller has generated the constrained secure action
\(\bu\). Accordingly, the model considered throughout is
\(\dot{\bx}
=\bm f(\bx)+\bm g(\bx)(\bu+\ba)+\bm k(\bx)\bd\).
The hard actuator constraint is imposed on the defender-generated
control action \(\bu\) only, and no constraint of the form
\(|u_i+a_i|<\bar u_i\) is assumed.
\end{assumption}

\begin{assumption}[Bounded attack and disturbance]
\label{ass:bounded_exogenous}
The attack and disturbance satisfy
\(\ba\in\mathcal L_2([0,\infty))\cap\mathcal L_\infty([0,\infty))\) and
\(\bd\in\mathcal L_2([0,\infty))\cap\mathcal L_\infty([0,\infty))\),
with \(\|\ba(t)\|\leq\bar a\) and \(\|\bd(t)\|\leq\bar d\) for all
\(t\geq0\), where \(\bar a,\bar d\geq0\) are known finite constants.
\end{assumption}

\begin{definition}[Admissible secure policy]
\label{def:admissible_secure_policy}
A continuous policy \(\bu_c:\Omega\rightarrow\mathcal U\) is admissible
on \(\Omega\), denoted by \(\bu_c\in\Psi(\Omega)\), if
\(\bu_c(\bm0)=\bm0\), the resulting closed-loop system is
asymptotically stable on \(\Omega\) in the absence of exogenous signals,
and the associated infinite-horizon cost is finite. A data-driven construction of an admissible initial policy is provided in
Appendix~\ref{app:koopman_warm_start}.
\end{definition}


\begin{assumption}[Plant regularity and operating domain]
\label{ass:plant_regularity}
Let \(\Omega\subset\mathbb R^n\) be a prescribed compact set containing
the origin in its interior. The unknown drift
\(f:\mathbb R^n\to\mathbb R^n\) is locally Lipschitz on an open
neighborhood of \(\Omega\) and satisfies \(f(0)=0\). The known input
maps \(g:\mathbb R^n\to\mathbb R^{n\times m}\) and
\(k:\mathbb R^n\to\mathbb R^{n\times s}\) are locally Lipschitz on the
same neighborhood. For every admissible policy considered below and
every admissible exogenous signal, the corresponding closed-loop
system admits a unique forward solution for every initial condition
\(x_0\in\Omega\), at least up to the first exit time from \(\Omega\).
\end{assumption}

The actuator constraint is
\begin{equation}
\mathcal U
:=
\left\{
\bu\in\mathbb R^m:
|u_\ell|<\bar u_\ell,\ 
\ell=1,\ldots,m
\right\},
\label{eq:input_set}
\end{equation}
where
\begin{equation}
\bUbar
:=
\diag(\bar u_1,\ldots,\bar u_m),
\qquad
\bR
:=
\diag(r_1,\ldots,r_m)
\succ0 .
\label{eq:input_matrices}
\end{equation}
The constraint is embedded into the control penalty
\begin{equation}
U(\bu)
=
2\sum_{\ell=1}^{m}
\bar u_\ell r_\ell
\int_{0}^{u_\ell}
\atanh\left(
\frac{\sigma}{\bar u_\ell}
\right)d\sigma .
\label{eq:nonquadratic_cost}
\end{equation}
The barrier-type penalty \(U\) is positive definite and strictly
convex on \(\mathcal U\), while its gradient becomes unbounded at the
constraint boundary. Hence, the stationary minimizing policy remains
strictly inside \(\mathcal U\) without a posteriori clipping.

For admissible feedback policies, consider
\begin{equation}
J(\bx_0,\bu,\ba,\bd)
=
\int_{0}^{\infty}
\ell(\bx,\bu,\ba,\bd)\,dt,
\label{eq:performance}
\end{equation}
with
\begin{equation}
\begin{aligned}
\ell(\bx,\bu,\ba,\bd)
:={}&
Q(\bx)
+U(\bu)\\
&-
\gamma_a^2\ba^\top\bTmat\ba
-
\gamma_d^2\bd^\top\bS\bd,
\end{aligned}
\label{eq:running_cost}
\end{equation}
where
\[
\bTmat=\bTmat^\top>0,
\qquad
\bS=\bS^\top>0,
\qquad
\gamma_a,\gamma_d>0.
\]
At this stage, \(Q:\mathbb R^n\rightarrow\mathbb R_{\geq0}\) is not
prescribed. It will be recovered from an inverse-optimal construction
so that the HJI identity itself generates the desired predefined-time
dissipation.


The saddle value is
\begin{equation}
V^*(\bx_0)
=
\min_{\bu\in\mathcal U}
\max_{\ba,\bd}
J(\bx_0,\bu,\ba,\bd).
\label{eq:value}
\end{equation}

\begin{assumption}[Ideal value-function regularity]
\label{ass:value_regularity}
The saddle value \(V^*:\Omega\rightarrow\mathbb R_{\geq0}\) exists,
is positive definite, satisfies \(V^*(0)=0\), and
\(V^*\in C(\Omega)\cap C^1(\Omega\setminus\{0\})\).
The HJI identity and the associated saddle policies are understood
for all \(x\in\Omega\setminus\{0\}\).
\end{assumption}

For \(V\in C^1\), define
\begin{equation}
\begin{aligned}
\mathcal H(\bx,\nabla V,\bu,\ba,\bd)
:={}&
Q(\bx)+U(\bu)
-\gamma_a^2\ba^\top\bTmat\ba\\
&-
\gamma_d^2\bd^\top\bS\bd
+\nabla V^\top
\big[
\bm f+\bm g(\bu+\ba)+\bm k\bd
\big],
\end{aligned}
\label{eq:Hamiltonian}
\end{equation}
where the state-dependent matrices are evaluated at \(\bx\).
The HJI equation is
\begin{equation}
0
=
\min_{\bu\in\mathcal U}
\max_{\ba,\bd}
\mathcal H(
\bx,\nabla V^*,\bu,\ba,\bd
).
\label{eq:HJI}
\end{equation}
The pointwise saddle policies are unique:
\begin{subequations}
\label{eq:optimal_saddle_policies}
\begin{align}
\bu^*(\bx)
={}&
-\bUbar
\tanh\left[
\frac{1}{2}
\bR^{-1}\bUbar^{-1}
\bm g^\top(\bx)\nabla V^*(\bx)
\right],
\label{eq:u_star}\\
\ba^*(\bx)
={}&
\frac{1}{2\gamma_a^2}
\bTmat^{-1}
\bm g^\top(\bx)\nabla V^*(\bx),
\label{eq:a_star}\\
\bd^*(\bx)
={}&
\frac{1}{2\gamma_d^2}
\bS^{-1}
\bm k^\top(\bx)\nabla V^*(\bx).
\label{eq:d_star}
\end{align}
\end{subequations}

\begin{remark}
    The term \(Q(\bx)\) in Equation \eqref{eq:Hamiltonian} of this section is defined in a general manner; however, in the subsequent part, it will be formulated specifically to ensure fixed-time stability of the closed-loop system.
\end{remark}

\subsection{Normalized Predefined-Time Comparison}
\label{subsec:normalized_comparison}

This section provides the reasoning tools and the methodology for designing the cost function to achieve fixed-time/Predefined-Time stability.

The origin is retained as the nominal regulation target and is not
removed from the HJI domain \(\Omega\). For the fixed-time comparison,
choose \(r_->0\) such that \(B_{r_-}(0)\subset\Omega\).
If \(x=0\), the nominal coordination objective is already attained;
if \(0<\|x\|<r_-\), the trajectory already belongs to the
prescribed terminal neighborhood \(B_{r_-}(0)\). Hence, only the
nontrivial region \(\|x\|\ge r_-\) requires the fixed-time
comparison. Define
\(\Omega^r:=\{x\in\Omega:\|x\|\ge r_-\}\).
Since \(\Omega\) is compact and
\(\{x:\|x\|\ge r_-\}\) is closed,
\(\Omega^r\) is compact and \(0\notin\Omega^r\).
Thus \(V^*(x)/\|x\|^2\) and
\(\|\nabla_xV^*(x)\|/\|x\|\) are well defined
on \(\Omega^r\). If a disturbance or attack subsequently drives the
trajectory outside \(B_{r_-}(0)\), the same comparison applies
again. The exclusion of the origin from \(\Omega^r\) is purely
analytical and does not alter the HJI operating domain.

\begin{lemma}[Local bounds of the ideal HJI value function]
\label{lem:value_quadratic_gradient}
Suppose that \(V^*\) is positive definite on \(\Omega\),
satisfies \(V^*(0)=0\), and
\(V^*\in C^1(\Omega\setminus\{0\})\). Then there exist
\(\underline{\alpha}\in\mathcal K_\infty\) and constants
\(\bar c>0\), \(c_{\nabla}>0\) such that
\(\underline{\alpha}(\|x\|)\le V^*(x)
\le\bar c\|x\|^2\) and
\(\|\nabla_xV^*(x)\|
\le c_{\nabla}\|x\|\) for all
\(x\in\Omega^r\).
\end{lemma}

\begin{proof}
Since \(\Omega^r\) is compact and \(0\notin\Omega^r\),
the maps \(V^*(x)/\|x\|^2\) and
\(\|\nabla_xV^*(x)\|/\|x\|\) are continuous on
\(\Omega^r\). Hence
\(\bar c:=\max_{x\in\Omega^r}
V^*(x)/\|x\|^2<\infty\) and
\(c_{\nabla}:=\max_{x\in\Omega^r}
\|\nabla_xV^*(x)\|/\|x\|<\infty\),
which give \(V^*(x)\le\bar c\|x\|^2\) and
\(\|\nabla_xV^*(x)\|
\le c_{\nabla}\|x\|\). Moreover, positive definiteness of
\(V^*\) and compactness of \(\Omega^r\) imply
\(\underline c:=\min_{x\in\Omega^r}
V^*(x)/\|x\|^2>0\). Choosing
\(\underline{\alpha}(s):=\underline c s^2\in\mathcal K_\infty\)
yields \(\underline{\alpha}(\|x\|)
\le V^*(x)\), completing the proof.
\end{proof}

\begin{remark}[Cost-induced fixed-time structure]
\label{rem:cost_induced_fixed_time}
For the fixed-time design developed later, the general penalty \(Q(x)\) is specialized as \(Q(x)=x^\top\mathbf Q_x x+\lambda_x\kappa_1\|x\|^{2\gamma_1}+\lambda_x\kappa_2\|x\|^{2\nu}\), where \(\mathbf Q_x=\mathbf Q_x^\top>0\), \(\kappa_1,\kappa_2,\lambda_x>0\), and \(0<\gamma_1<1<\nu\). This construction depends only on the measurable state and does not introduce the unknown value function into the running cost. From Lemma~\ref{lem:value_quadratic_gradient}, \(\|x\|^{2\gamma_1}\geq\bar c^{-\gamma_1}(V^*(x))^{\gamma_1}\) and \(\|x\|^{2\nu}\geq\bar c^{-\nu}(V^*(x))^\nu\). Along the ideal saddle-point trajectory, \(\dot V^*(x)=-Q(x)-U(u^*)+\gamma_a^2(a^*(x))^\top\mathbf T a^*(x)+\gamma_d^2(d^*(x))^\top\mathbf S d^*(x)\). Since \(\Omega^r\) is compact and \(V^*\) is continuously differentiable, the last two terms are bounded on \(\Omega^r\); hence, there exists \(\Delta_H\geq0\) such that \(\gamma_a^2(a^*(x))^\top\mathbf T a^*(x)+\gamma_d^2(d^*(x))^\top\mathbf S d^*(x)\leq\Delta_H\) for all \(x\in\Omega^r\). Using \(U(u^*)\geq0\) together with the above state-to-value power bounds gives
\begin{equation}
    \dot V^*(x)
    \leq
    -\lambda_x c_{x1}\big(V^*(x)\big)^{\gamma_1}
    -\lambda_x c_{x2}\big(V^*(x)\big)^\nu
    +\Delta_H,
\label{eq:ideal_fixed_time_value_inequality}
\end{equation}
where \(c_{x1}=\kappa_1\bar c^{-\gamma_1}\) and \(c_{x2}=\kappa_2\bar c^{-\nu}\). Thus, the two fixed-time powers arise directly from the HJI-shaped state penalty rather than from an external assumption on the unknown value function.
\end{remark}

\begin{remark}[Computation of \(\bar c\)]
\label{rem:computable_cbar}
The constant \(\bar c\) in Lemma~\ref{lem:value_quadratic_gradient} need not be evaluated from the unknown \(V^*\). All signals entering the relevant cost and finite-window quantities admit uniform bounds on \(\Omega^r\). For the primitive formulation in~\eqref{eq:plant}, Assumption~\ref{ass:bounded_exogenous} gives \(\|a(t)\|\leq\bar a\) and \(\|d(t)\|\leq\bar d\). If the signals are generated by the ideal saddle policies, Lemma~\ref{lem:value_quadratic_gradient}, Assumption~\ref{ass:plant_regularity}, and compactness of \(\Omega^r\) give \(\bar g:=\max_{\bx\in\Omega^r}\|g(\bx)\|<\infty\), \(\bar k:=\max_{\bx\in\Omega^r}\|k(\bx)\|<\infty\), and \(\bar x_r:=\max_{\bx\in\Omega^r}\|\bx\|<\infty\), so that, from (14b)--(14c), \(\|a^*(\bx)\|\leq\|T^{-1}\|\bar g c_{\nabla}\bar x_r/(2\gamma_a^2)=:\bar a^*<\infty\) and \(\|d^*(\bx)\|\leq\|S^{-1}\|\bar k c_{\nabla}\bar x_r/(2\gamma_d^2)=:\bar d^*<\infty\). Likewise, since \(\|\nabla\phi(\bx)\|\leq\bar\phi_g\) on \(\Omega^r\) and Lemma~\ref{lem:window_bounded_learned_signals} gives \(\|\hat{\bW}(t)\|\leq\bar W_{\hat{}}\), the learned saddle signals satisfy \(\|\hat a(t)\|\leq\|T^{-1}\|\bar g\bar\phi_g\bar W_{\hat{}}/(2\gamma_a^2)=:\bar{\hat a}<\infty\) and \(\|\hat d(t)\|\leq\|S^{-1}\|\bar k\bar\phi_g\bar W_{\hat{}}/(2\gamma_d^2)=:\bar{\hat d}<\infty\), while the saturation parameterization gives \(|\hat u_i(t)|<\bar u_i\).

Accordingly, let \(u_b\in\Psi(\Omega)\) be any certified admissible secure policy and let \(\bar J_b<\infty\) be a computable worst-case performance certificate satisfying \(\bar J_b\geq\sup_{\bx_0\in\Omega^r}\sup_{a,d}J(\bx_0;u_b,a,d)\), where the admissible \(a,d\) obey Assumptions~\ref{ass:bounded_exogenous} and~\ref{ass:plant_regularity}. Since \(V^*(\bx)=\min_u\max_{a,d}J(\bx;u,a,d)\), one has \(V^*(\bx)\leq\sup_{a,d}J(\bx;u_b,a,d)\leq\bar J_b\) for every \(\bx\in\Omega^r\). Moreover, \(\|\bx\|\geq r_-\) on \(\Omega^r\);.The operating domain \(\Omega\), and hence \(\Omega^r=\{\bx\in\Omega:\|\bx\|\geq r_-\}\), is prescribed a priori from the design specifications, physical operating limits, safety requirements, or the intended data-collection region, rather than inferred from the unknown drift dynamics. hence
\[
\bar c
=
\max_{\bx\in\Omega^r}
\frac{V^*(\bx)}{\|\bx\|^2}
\leq
\frac{\bar J_b}{r_-^2}
=:
\bar c_b .
\]
Thus, the explicit certified number \(\bar c_b=\bar J_b/r_-^2\) may be used everywhere in place of the unknown exact \(\bar c\), and consequently \(\|\bx\|^{2\gamma_1}\geq\bar c_b^{-\gamma_1}(V^*(\bx))^{\gamma_1}\) and \(\|\bx\|^{2\nu}\geq\bar c_b^{-\nu}(V^*(\bx))^\nu\). The price is conservatism: a small prescribed radius \(r_-\) or a worst-case \(\bar J_b\) computed over a large operating domain can make \(\bar c_b\) large, thereby increasing the penalty or learning gains required by the predefined-time design and potentially producing more aggressive near-saturation transients, greater numerical stiffness and dynamic range, and tighter sampling and computational requirements in digital hardware. This is the deliberate tradeoff for obtaining an a-priori computable convergence certificate without knowing \(V^*\).
\end{remark}

\section{Integral Bellman--Isaacs Learning}
\label{sec:IRL}

Let \(\bphi:\Omega^r\rightarrow\R^N\) be a continuously
differentiable basis vector satisfying \(\bphi(\bm0)=\bm0\).
The ideal and learned value functions are represented as
\begin{equation}
V^*(\bx)
=
\bW^{*\top}\bphi(\bx)+\varepsilon(\bx),
\qquad
\hat V(\bx)
=
\hat{\bW}^{\top}\bphi(\bx),
\label{eq:value_approx}
\end{equation}
where \(\bW^*\in\R^N\) is an unknown constant ideal weight,
\(\hat{\bW}\in\R^N\) is its online estimate, and
\begin{equation}
\widetilde{\bW}
:=
\bW^*-\hat{\bW}.
\label{eq:weight_error}
\end{equation}
On the compact operating domain \(\Omega^r\), assume that
\begin{equation}
\begin{aligned}
\norm{\bW^*}&\le W_M,
&\qquad
\norm{\bphi}&\le\bar\phi,
&\qquad
\norm{\nabla\bphi}&\le\bar\phi_g,
\\[-1mm]
|\varepsilon|&\le\bar\varepsilon,
&\qquad
\norm{\nabla\varepsilon}&\le\bar\varepsilon_g.
\end{aligned}
\label{eq:approximation_bounds}
\end{equation}


Consequently,
\begin{equation}
\begin{aligned}
\nabla V^*
&=\nabla\bphi\T\bW^*+\nabla\varepsilon,
&\qquad
\nabla\hat V
&=\nabla\bphi\T\hat{\bW},
\\[-1mm]
\nabla V^*-\nabla\hat V
&=\nabla\bphi\T\widetilde{\bW}
+\nabla\varepsilon .
\end{aligned}
\label{eq:value_gradient_decomposition}
\end{equation}
Since \(\Omega^r\) is compact and \(\bphi\in C^1(\Omega^r)\), define
\(\bar\phi:=\max_{x\in\Omega^r}\|\bphi(x)\|<\infty\) and
\(\bar\phi_g:=\max_{x\in\Omega^r}\|\nabla\bphi(x)\|<\infty\); hence
\(\|\bphi(x)\|\leq\bar\phi\) and
\(\|\nabla\bphi(x)\|\leq\bar\phi_g\) for all \(x\in\Omega^r\).

Replacing \(\nabla V^*\) in
\eqref{eq:u_star}--\eqref{eq:d_star} by
\(\nabla\hat V\) gives the learned saddle policies
\begin{align}
\hat{\bu}
&=
-\bUbar\tanh\left(
\frac{1}{2}
\bR^{-1}\bUbar^{-1}
\bm g\T\nabla\bphi\T\hat{\bW}
\right),
\label{eq:u_hat}\\
\hat{\ba}
&=
\frac{1}{2\gamma_a^2}
\bTmat^{-1}\bm g\T
\nabla\bphi\T\hat{\bW},
\label{eq:a_hat}\\
\hat{\bd}
&=
\frac{1}{2\gamma_d^2}
\bS^{-1}\bm k\T
\nabla\bphi\T\hat{\bW}.
\label{eq:d_hat}
\end{align}
For compactness, define
\begin{equation}
\begin{aligned}
\ell_p^*(t)
&:=
\ell_p\!\left(
\bx(t),\bu^*(t),\ba^*(t),\bd^*(t)
\right),
\\[-1mm]
\hat\ell_p(t)
&:=
\ell_p\!\left(
\bx(t),\hat{\bu}(t),\hat{\ba}(t),\hat{\bd}(t)
\right).
\end{aligned}
\label{eq:ideal_learned_running_costs}
\end{equation}

\begin{lemma}[Window-wise bounded learned signals]
\label{lem:window_bounded_learned_signals}
Suppose that \(\bx(t)\in\Omega^r\), where \(\Omega^r\) is compact, and that
the critic weight \(\hat{\bW}\) evolves according to
\eqref{eq:critic_update}. Then, for every fixed \(t\geq T\), there
exist finite constants \(\bar W_t,\bar u,\bar a_t,\bar d_t,\bar\ell_t>0\)
such that
\(\|\hat{\bW}(\tau)\|\leq\bar W_t\),
\(\|\hat{\bu}(\tau)\|\leq\bar u\),
\(\|\hat{\ba}(\tau)\|\leq\bar a_t\),
\(\|\hat{\bd}(\tau)\|\leq\bar d_t\), and
\(|\hat\ell_p(\tau)|\leq\bar\ell_t\) for all
\(\tau\in[t-T,t]\). The constants may be enlarged, if necessary, to
cover the finitely many recorded replay windows.
\end{lemma}

\begin{proof}
Denote the right-hand side of \eqref{eq:critic_update} by
\(\mathcal G_W(t)\), so that
\(\dot{\hat{\bW}}=\mathcal G_W(t)\). On every compact interval
\([t-T,t]\), continuity of the basis functions on \(\Omega^r\), bounded
normalized regressors, and the signed-power residual mappings imply
\(\mathcal G_W\in L_{\mathrm{loc}}^1\). Hence
\(\hat{\bW}(\tau)=\hat{\bW}(t-T)+
\int_{t-T}^{\tau}\mathcal G_W(s)\,ds\), so
\(\hat{\bW}\in AC_{\mathrm{loc}}\subset C^0\) and therefore
\(\bar W_t:=\max_{\tau\in[t-T,t]}\|\hat{\bW}(\tau)\|<\infty\).
The saturation-aware control law gives
\(|\hat u_j(\tau)|<\bar u_j\), hence
\(\|\hat{\bu}(\tau)\|\leq
\bar u:=\|\operatorname{col}\{\bar u_1,\ldots,\bar u_m\}\|\).
Since \(\hat{\ba}\) and \(\hat{\bd}\) depend continuously on
\(\bx\), \(\nabla\bm\phi(\bx)\), and \(\hat{\bW}\), compactness of
\(\Omega^r\) and boundedness of \(\hat{\bW}\) imply
\(\bar a_t:=\max_{\tau\in[t-T,t]}\|\hat{\ba}(\tau)\|<\infty\) and
\(\bar d_t:=\max_{\tau\in[t-T,t]}\|\hat{\bd}(\tau)\|<\infty\).
Consequently,
\(\hat\ell_p(\tau)=
\ell_p(\bx(\tau),\hat{\bu}(\tau),\hat{\ba}(\tau),\hat{\bd}(\tau))\)
is continuous on the compact interval \([t-T,t]\), and thus
\(\bar\ell_t:=\max_{\tau\in[t-T,t]}|\hat\ell_p(\tau)|<\infty\).
Since the replay stack contains only finitely many stored windows,
these bounds can be enlarged to dominate the corresponding recorded
quantities as well.
\end{proof}

Let \(T>0\) be the reinforcement-window length. Along the ideal
saddle vector field, the HJI identity gives
\(\dot V^*(\bx(t))=-\ell_p^*(t)\). Integration over
\([t-T,t]\) therefore yields the finite-window
Bellman--Isaacs identity
\begin{equation}
V^*(\bx(t))
-
V^*(\bx(t-T))
+
\int_{t-T}^{t}\ell_p^*(\tau)d\tau
=
0.
\label{eq:integral_HJI}
\end{equation}


\begin{remark}
Equation~\eqref{eq:integral_HJI} is exact along the ideal saddle vector field. If, during
learning, the measured transition
\(\bx(t-T)\mapsto\bx(t)\) is instead generated by
\(\dot{\bx}
=f(\bx)+g(\bx)(\hat{\bu}+\ba)+k(\bx)\bd\),
define the finite-window trajectory mismatch
\(\eta_{\rm tr}(t):=
\int_{t-T}^{t}
\nabla V^{*\top}(\bx(\tau))g(\bx(\tau))
\big[(\hat{\bu}(\tau)+\ba(\tau))
-(\bu^*(\tau)+\ba^*(\tau))\big]\,d\tau
+\int_{t-T}^{t}
\nabla V^{*\top}(\bx(\tau))k(\bx(\tau))
\big[\bd(\tau)-\bd^*(\tau)\big]\,d\tau\).
Here, \(\eta_{\rm tr}(t)\in\mathbb{R}\) denotes the accumulated
value-function-rate mismatch over \([t-T,t]\) induced by the difference
between the measured closed-loop vector field
\(f(\bx)+g(\bx)(\hat{\bu}+\ba)+k(\bx)\bd\)
and the ideal saddle vector field
\(f(\bx)+g(\bx)(\bu^*+\ba^*)+k(\bx)\bd^*\).
Accordingly, the measured-trajectory counterpart of \eqref{eq:integral_HJI} is
\(V^*(\bx(t))-V^*(\bx(t-T))
+\int_{t-T}^{t}\ell_p^*(\tau)\,d\tau-\eta_{\rm tr}(t)=0\).
Hence, the integral-data definitions in \eqref{eq:online_delta}--\eqref{eq:normalized_errors} remain unchanged,
while \(\eta_{\rm tr}(t)\) enters the subsequent critic regression only
through the aggregate integral residual in \eqref{eq:epsilon_H_definition}. In particular,
\(\eta_{\rm tr}(t)=0\) whenever
\(\hat{\bu}(\tau)+\ba(\tau)
=\bu^*(\tau)+\ba^*(\tau)\) and
\(\bd(\tau)=\bd^*(\tau)\) for all
\(\tau\in[t-T,t]\), in which case (30) is recovered exactly.
\end{remark}


Unlike a pointwise HJI equation, the finite-window formulation contains
the unknown drift only through the measured state transition
\(x(t-T)\mapsto x(t)\); no direct evaluation of \(f(x)\) is required in
the critic regression~\cite{Vrabie2009,Modares2014}. Any departure of
the measured trajectory from the ideal saddle vector field is captured
by the finite-window mismatch \(\eta_{\rm tr}(t)\) introduced above.

Define the online integral data
\begin{align}
\Delta\bphi(t)
&:=
\bphi(\bx(t))-\bphi(\bx(t-T)),
\label{eq:Dphi}\\
\rho(t)
&:=
\int_{t-T}^{t}\hat\ell_p(\tau)d\tau,
\label{eq:rho}\\
\delta(t)
&:=
\hat{\bW}^{\top}(t)\Delta\bphi(t)+\rho(t).
\label{eq:online_delta}
\end{align}
The difference between \(\hat\ell_p\) and \(\ell_p^*\), together with
the finite-window trajectory mismatch \(\eta_{\rm tr}(t)\) introduced
above, is not discarded; both appear explicitly as bounded residual
components below. At sampling instants \(t_i\), store the tuples
\begin{equation}
\mathcal D
=
\left\{
\Delta\bphi_i,\rho_i
\right\}_{i=1}^{k},
\qquad
\Delta\bphi_i:=\Delta\bphi(t_i),
\qquad
\rho_i:=\rho(t_i),
\label{eq:memory}
\end{equation}
and recompute their residuals using the current weight:
\begin{equation}
\delta_i(t)
=
\hat{\bW}^{\top}(t)\Delta\bphi_i+\rho_i,
\qquad
i=1,\ldots,k.
\label{eq:recorded_delta}
\end{equation}

To prevent large state transitions from dominating the update, define
\begin{align}
m(t)
&:=
1+\Delta\bphi\T(t)\Delta\bphi(t),
&
\bomega(t)
&:=
\frac{\Delta\bphi(t)}{m(t)},
\label{eq:online_normalization}\\
m_i
&:=
1+\Delta\bphi_i\T\Delta\bphi_i,
&
\bomega_i
&:=
\frac{\Delta\bphi_i}{m_i},
\label{eq:recorded_normalization}\\
e(t)
&:=
\frac{\delta(t)}{m(t)},
&
e_i(t)
&:=
\frac{\delta_i(t)}{m_i}.
\label{eq:normalized_errors}
\end{align}
For every vector \(\bm z\),
\[
\frac{\norm{\bm z}}{1+\norm{\bm z}^2}\le\frac12;
\]
hence the normalized regressors satisfy the uniform bounds
\begin{equation}
\norm{\bomega(t)}\le\frac12,
\qquad
\norm{\bomega_i}\le\frac12.
\label{eq:normalized_regressor_bound}
\end{equation}

The trajectory mismatch is also uniformly bounded on the compact
operating domain. Indeed, define
\(c_{Vg}:=\max_{\bx\in\Omega^r}
\norm{\bm g^\top(\bx)\nabla V^*(\bx)}\),
\(c_{Vk}:=\max_{\bx\in\Omega^r}
\norm{\bm k^\top(\bx)\nabla V^*(\bx)}\), and
\(u_m:=(\sum_{j=1}^{m}\bar u_j^2)^{1/2}\).
Compactness of \(\Omega^r\), \(V^*\in C^1(\Omega^r)\), and continuity
of \(\bm g\) and \(\bm k\) imply \(c_{Vg},c_{Vk}<\infty\), whereas
the saturated policies give
\(\norm{\hat{\bu}}\le u_m\) and \(\norm{\bu^*}\le u_m\).
Moreover, the saddle policies imply
\(\norm{\ba^*}\le\bar a^*:=
\norm{\mathbf T^{-1}}c_{Vg}/(2\gamma_a^2)\) and
\(\norm{\bd^*}\le\bar d^*:=
\norm{\mathbf S^{-1}}c_{Vk}/(2\gamma_d^2)\).
Together with \(\norm{\ba}\le\bar a\) and
\(\norm{\bd}\le\bar d\), the definition of
\(\eta_{\rm tr}(t)\) yields
\(\abs{\eta_{\rm tr}(t)}
\le
T[c_{Vg}(2u_m+\bar a+\bar a^*)
+c_{Vk}(\bar d+\bar d^*)]
=:\bar\eta_{\rm tr}<\infty\).

\begin{proposition}[Integral residual decomposition]
\label{prop:residual}
Suppose that the running-cost mismatch is uniformly bounded on
\(\Omega^r\), i.e.,
\begin{equation}
\left|
\hat\ell_p(t)-\ell_p^*(t)
\right|
\le
\bar\ell_\Delta
\label{eq:running_cost_mismatch_bound}
\end{equation}
for a finite \(\bar\ell_\Delta\ge0\). Then the normalized online and
recorded residuals admit the decompositions
\begin{align}
e(t)
&=
-\widetilde{\bW}\T\bomega(t)
+
\bar\epsilon_H(t),
\label{eq:online_decomposition}\\
e_i(t)
&=
-\widetilde{\bW}\T\bomega_i
+
\bar\epsilon_{Hi},
\qquad i=1,\ldots,k,
\label{eq:recorded_decomposition}
\end{align}
where
\begin{equation}
\bar\epsilon_H(t)
=
\frac{
-\Delta\varepsilon(t)
+\int_{t-T}^{t}
\bigl[\hat\ell_p(\tau)-\ell_p^*(\tau)\bigr]\,d\tau
+\eta_{\rm tr}(t)
}{
m(t)
}.
\label{eq:epsilon_H_definition}
\end{equation}
\end{proposition}

\begin{proof}
Substitution of \eqref{eq:value_approx} into the measured-trajectory
counterpart of \eqref{eq:integral_HJI} introduced above gives
\begin{equation}
\bW^{*\top}\Delta\bphi(t)
+
\int_{t-T}^{t}\ell_p^*(\tau)d\tau
+
\Delta\varepsilon(t)
-
\eta_{\rm tr}(t)
=
0,
\label{eq:ideal_weight_integral}
\end{equation}
where
\begin{equation}
\Delta\varepsilon(t)
:=
\varepsilon(\bx(t))-\varepsilon(\bx(t-T)).
\label{eq:Delta_epsilon}
\end{equation}
Since
\(\hat{\bW}=\bW^*-\widetilde{\bW}\), adding
\(\rho(t)=\int_{t-T}^{t}\hat\ell_p(\tau)d\tau\) to
\(\hat{\bW}^{\top}\Delta\bphi(t)\) and using
\eqref{eq:ideal_weight_integral} yield
\begin{equation}
\begin{aligned}
\delta(t)
={}&
-\widetilde{\bW}\T\Delta\bphi(t)
-\Delta\varepsilon(t)
\\[-1mm]
&+
\int_{t-T}^{t}
\left[
\hat\ell_p(\tau)-\ell_p^*(\tau)
\right]d\tau
+\eta_{\rm tr}(t).
\end{aligned}
\label{eq:delta_decomposition_raw}
\end{equation}
Dividing by \(m(t)\) gives
\eqref{eq:online_decomposition}, with
\begin{equation}
\bar\epsilon_H(t)
=
\frac{
-\Delta\varepsilon(t)
+
\displaystyle\int_{t-T}^{t}
[\hat\ell_p(\tau)-\ell_p^*(\tau)]d\tau
+
\eta_{\rm tr}(t)
}{
m(t)
}.
\label{eq:epsilon_H_definition_proof}
\end{equation}
Because \(m(t)\ge1\),
\(\abs{\Delta\varepsilon(t)}\le2\bar\varepsilon\),
\eqref{eq:running_cost_mismatch_bound} holds, and
\(\abs{\eta_{\rm tr}(t)}\le\bar\eta_{\rm tr}\),
\begin{equation}
|\bar\epsilon_H(t)|
\le
2\bar\varepsilon
+
T\bar\ell_\Delta
+
\bar\eta_{\rm tr}
=
\epsilon_{Hm}.
\label{eq:online_residual_bound_proof}
\end{equation}
Applying the same calculation to each stored window
\([t_i-T,t_i]\) proves
\eqref{eq:recorded_decomposition} and the corresponding bound on
\(\bar\epsilon_{Hi}\).
\end{proof}

The constant \(\bar\ell_\Delta\) is finite because
\(\Omega^r\) is compact, \(\nabla\bphi\) is bounded, and the policy maps
\eqref{eq:u_hat}--\eqref{eq:d_hat} are continuous on the relevant
state--weight domain. 


Define the normalized replay matrix
\begin{equation}
\bar\bOmega
:=
\left[
\bomega_1\ \cdots\ \bomega_k
\right]
\in\R^{N\times k}.
\label{eq:Omega_bar}
\end{equation}

\begin{definition}[Sufficiently rich integral data]
\label{def:rich}
A replay stack \(\mathcal D\) with \(k\ge N\) is sufficiently rich if
\begin{equation}
\rank(\bar\bOmega)=N,
\qquad
\sigma_{\min}(\bar\bOmega)
\ge
\underline\sigma_\Omega
>
0.
\label{eq:richness}
\end{equation}
\end{definition}

Condition \eqref{eq:richness} means that the recorded integral
transitions span every critic-weight direction. A bounded probing
signal may be used during an initial data-collection interval; once
\eqref{eq:richness} is attained, the probing signal may be removed
because the stored regressors retain the informative directions
\cite{Modares2014}. Let
\begin{equation}
T_E
:=
\inf\left\{
t\ge0:
\sigma_{\min}(\bar\bOmega(t))
\ge\underline\sigma_\Omega
\right\}
\label{eq:data_informativity_time}
\end{equation}
denote the first data-informativity time.

\begin{remark}[Reinforcement-window trade-off]
\label{rem:window_tradeoff}
From \eqref{eq:online_residual_bound_proof},
\(\epsilon_{Hm}=2\bar\epsilon+T\bar\ell_\Delta\), so
\(T\downarrow\Rightarrow\epsilon_{Hm}\downarrow\) through the accumulated
policy-mismatch term. However,
\(\Delta\bphi_i=\bphi(x(t_i))-\bphi(x(t_i-T))
=T\dot{\bphi}(x(t_i))+o(T)\) as \(T\to0\), and hence
\(\bomega_i=\Delta\bphi_i/(1+\|\Delta\bphi_i\|^2)=O(T)\).
Therefore, an excessively small \(T\) may yield
\(\sigma_{\min}(\bar\bOmega)\downarrow0\), even when
\(\operatorname{rank}(\bar\bOmega)=N\), which weakens numerical
informativity and increases the critic gain required for a prescribed
learning horizon. Thus \(T\) trades residual accuracy against replay-stack
conditioning.
\end{remark}

For fixed trajectory data, the online and replay residuals depend
affinely on the current weight:
\begin{equation}
\begin{aligned}
e(\hat{\bW},t)
&=
\hat{\bW}^{\top}\bomega(t)
+
\frac{\rho(t)}{m(t)},
&
\nabla_{\hat{\bW}}e
&=
\bomega(t),
\\
e_i(\hat{\bW},t)
&=
\hat{\bW}^{\top}\bomega_i
+
\frac{\rho_i}{m_i},
&
\nabla_{\hat{\bW}}e_i
&=
\bomega_i.
\end{aligned}
\label{eq:residual_weight_dependence}
\end{equation}

\begin{remark}[Intrepretation for~\eqref{eq:residual_weight_dependence} ]
For completeness, the affine dependence of the integral residual on
the current critic weight follows directly from the saddle-point
stationarity conditions. Let the approximate Hamiltonian associated
with the learned policies be
\(\hat{\mathcal H}
=\hat r+
\hat{\bW}^{\top}\nabla\bm\phi(\bx)
[\bm f(\bx)+\bm g(\bx)(\hat{\bu}+\hat{\ba})
+\bm k(\bx)\hat{\bd}]\).
Although
\(\hat{\bu}=\hat{\bu}(\hat{\bW})\),
\(\hat{\ba}=\hat{\ba}(\hat{\bW})\), and
\(\hat{\bd}=\hat{\bd}(\hat{\bW})\), differentiation with respect to
\(\hat{\bW}\) gives the direct term
\(\nabla\bm\phi(\bx)
[\bm f+\bm g(\hat{\bu}+\hat{\ba})+\bm k\hat{\bd}]\)
plus terms proportional to
\(\partial\hat{\bu}/\partial\hat{\bW}\),
\(\partial\hat{\ba}/\partial\hat{\bW}\), and
\(\partial\hat{\bd}/\partial\hat{\bW}\).
The latter vanish because
\(\partial U(\hat{\bu})/\partial\hat{\bu}
+\bm g^{\top}\nabla\bm\phi^{\top}\hat{\bW}=\bm0\),
\(-2\gamma_a^2\bTmat\hat{\ba}
+\bm g^{\top}\nabla\bm\phi^{\top}\hat{\bW}=\bm0\), and
\(-2\gamma_d^2\bS\hat{\bd}
+\bm k^{\top}\nabla\bm\phi^{\top}\hat{\bW}=\bm0\).
Hence
\(\partial\hat{\mathcal H}/\partial\hat{\bW}
=\nabla\bm\phi(\bx)
[\bm f+\bm g(\hat{\bu}+\hat{\ba})+\bm k\hat{\bd}]
=\nabla\bm\phi(\bx)\dot{\bx}\).
Over the fixed integration window,
\(\int_{t-T}^{t}
\nabla\bm\phi(\bx(\tau))\dot{\bx}(\tau)\,d\tau
=\bm\phi(\bx(t))-\bm\phi(\bx(t-T))\).
Therefore, with the normalization already introduced in the integral
residual, the current regressor is fixed during instantaneous
weight differentiation and
\(\nabla_{\hat{\bW}}e(\hat{\bW},t)=\bomega(t)\).
Likewise, each stored tuple
\((\bomega_i,\rho_i,m_i)\) is frozen after insertion into the replay
stack, so
\(\nabla_{\hat{\bW}}e_i(\hat{\bW},t)=\bomega_i\).
  
\end{remark}

A conventional quadratic Bellman-error loss generates a residual
term proportional only to \(e\). To obtain the two-power structure
required for fixed- and predefined-time convergence, the critic loss
must instead generate one residual power below one and another above
one \cite{Polyakov2012,Kokkolakis2025}. Accordingly, select

\begin{equation}
\begin{aligned}
\mathcal L_c(\hat{\bW},t)
:={}&
\sum_{q\in\{\gamma_1,\gamma_2\}}
\frac{1}{q+1}
\Bigg[
|e(\hat{\bW},t)|^{q+1}
\\[-1mm]
&\qquad\qquad+
\sum_{i=1}^{k}
|e_i(\hat{\bW},t)|^{q+1}
\Bigg],
\end{aligned}
\label{eq:critic_objective}
\end{equation}
where \(0<\gamma_1<1<\gamma_2\), so that
\(\gamma_1+1\in(1,2)\) and \(\gamma_2+1>2\).

the first term dominates the local error regime, whereas the second
prevents the settling-time bound from increasing with a large initial
weight error.

Using
\eqref{eq:residual_weight_dependence}, and the chain rule gives
\begin{equation}
\begin{aligned}
\nabla_{\hat{\bW}}\mathcal L_c
={}&
\bomega(t)
\left[
\spow{e(t)}{\gamma_1}
+
\spow{e(t)}{\gamma_2}
\right]
\\[-1mm]
&+
\sum_{i=1}^{k}
\bomega_i
\left[
\spow{e_i(t)}{\gamma_1}
+
\spow{e_i(t)}{\gamma_2}
\right].
\end{aligned}
\label{eq:critic_objective_gradient}
\end{equation}
The critic is updated through the continuous-time negative-gradient
flow
\begin{equation}
\dot{\hat{\bW}}
=
-\alpha_c
\nabla_{\hat{\bW}}\mathcal L_c,
\qquad
\alpha_c>0,
\label{eq:critic_gradient_flow}
\end{equation}
or, equivalently,
\begin{equation}
\begin{aligned}
\dot{\hat{\bW}}
={}&
-\alpha_c\bomega(t)
\left[
\spow{e(t)}{\gamma_1}
+
\spow{e(t)}{\gamma_2}
\right]
\\[-1mm]
&-
\alpha_c
\sum_{i=1}^{k}
\bomega_i
\left[
\spow{e_i(t)}{\gamma_1}
+
\spow{e_i(t)}{\gamma_2}
\right].
\end{aligned}
\label{eq:critic_update}
\end{equation}

The current residual supplies new trajectory information, while the
replay terms preserve previously excited directions. Under
\eqref{eq:richness}, the latter convert the projected errors
\(\widetilde{\bW}^{\top}\bomega_i\) into coercive bounds on the full
weight error. In view of
Theorem~\ref{thm:normalized_predefined_time}, once this two-power
dissipation structure is established, the resulting critic
convergence-time scale is governed by the learning gain
\(\alpha_c\). The next section therefore derives an explicit
selection of \(\alpha_c\) from the assigned critic-learning horizon
\(T_W\).

\section{Predefined-Time Integral Critic}
\label{sec:critic}

The following exact two-power integral provides the basis for the
deadline-dependent critic-gain design. In contrast with conventional
fixed-time estimates obtained by splitting the Lyapunov domain into
low- and high-energy regions, the result below evaluates the complete
settling-time integral and therefore yields a less conservative
deadline parameterization.

\begin{lemma}[Exact two-power settling-time integral]
\label{lem:Gamma_comparison}
Let \(Y:[t_0,\infty)\rightarrow\mathbb R_{\geq0}\) be absolutely
continuous and satisfy, almost everywhere,
\begin{equation}
\dot Y
\le
-aY^{\frac{\gamma_1+1}{2}}
-bY^{\frac{\gamma_2+1}{2}},
\qquad
a,b>0,\quad
0<\gamma_1<1<\gamma_2.
\label{eq:two_power_exact}
\end{equation}
Then \(Y\) reaches zero in a time satisfying
\begin{equation}
\begin{aligned}
T(Y(t_0))
&\le
\int_{0}^{Y(t_0)}
\frac{ds}{
a s^{\frac{\gamma_1+1}{2}}
+
b s^{\frac{\gamma_2+1}{2}}
}
\\
&<
\frac{
2
\Gamma_{\!E}\!\left(
\dfrac{\gamma_2-1}{\gamma_2-\gamma_1}
\right)
\Gamma_{\!E}\!\left(
\dfrac{1-\gamma_1}{\gamma_2-\gamma_1}
\right)
}{
(\gamma_2-\gamma_1)
a^{\frac{\gamma_2-1}{\gamma_2-\gamma_1}}
b^{\frac{1-\gamma_1}{\gamma_2-\gamma_1}}
},
\end{aligned}
\label{eq:Gamma_time_bound}
\end{equation}
where the last bound is independent of \(Y(t_0)\). Moreover,
\eqref{eq:Gamma_time_bound} is the \(r=1\) two-power specialization of
the normalized predefined-time comparison in
Theorem~\ref{thm:normalized_predefined_time}.
\end{lemma}

\begin{proof}
Whenever \(Y>0\), \eqref{eq:two_power_exact} gives
\(dt\leq-dY/
\bigl(aY^{\frac{\gamma_1+1}{2}}
+bY^{\frac{\gamma_2+1}{2}}\bigr)\), which yields the first bound in
\eqref{eq:Gamma_time_bound}. For the uniform bound, use the normalized
integral identity established in
Theorem~\ref{thm:normalized_predefined_time} with
\(r=1\),
\(p=(\gamma_1+1)/2\),
\(q=(\gamma_2+1)/2\),
\(\alpha=a\), and \(\beta=b\).
Since
\(q-p=(\gamma_2-\gamma_1)/2\),
\((1-p)/(q-p)=(1-\gamma_1)/(\gamma_2-\gamma_1)\), and
\((q-1)/(q-p)=(\gamma_2-1)/(\gamma_2-\gamma_1)\),
\eqref{eq:gamma_normalization} reduces exactly to the last term in
\eqref{eq:Gamma_time_bound}. Hence the bound is finite and independent
of \(Y(t_0)\).
\end{proof}

\begin{remark}
\label{rem:two_power_quantitative} 
Lemma~\ref{lem:Gamma_comparison} addresses the converse of the question posed in Theorem~\ref{thm:normalized_predefined_time} specifically, determining the maximum convergence time for the dynamics described in~\eqref{lem:Gamma_comparison} and subsequently links this to Theorem~\ref{thm:normalized_predefined_time} to establish the relationship between this desired time and the coefficients of the differential equation; this is referred to as the inverse design criterion.
Let
\(A_\gamma:=(\gamma_2-1)/(\gamma_2-\gamma_1)\) and
\(B_\gamma:=(1-\gamma_1)/(\gamma_2-\gamma_1)\), so that
\(A_\gamma,B_\gamma\in(0,1)\) and \(A_\gamma+B_\gamma=1\).
Lemma~\ref{lem:Gamma_comparison} can then be written as
\(T(Y(t_0))<C_\gamma a^{-A_\gamma}b^{-B_\gamma}\), where
\(C_\gamma:=
2\Gamma_{\!E}(A_\gamma)\Gamma_{\!E}(B_\gamma)/
(\gamma_2-\gamma_1)\).
Hence \(a\) and \(b\) quantify, respectively, the low- and high-power
contributions to the uniform settling-time bound, while their balance
occurs at
\(Y_\times=(a/b)^{2/(\gamma_2-\gamma_1)}\).
In particular, the common scaling
\((a,b)\mapsto(ca,cb)\) gives \(T\mapsto T/c\).
For the critic dynamics below, with
\(\mu=(\gamma_1+1)/2\), \(\nu=(\gamma_2+1)/2\),
\(a=a_0(2\alpha_c)^\mu\), and
\(b=b_0(2\alpha_c)^\nu\), one has
\(\mu A_\gamma+\nu B_\gamma=1\); consequently,
\(T_c\) scales as \(1/\alpha_c\). This identity is the quantitative
basis for selecting \(\alpha_c\) directly from the assigned critic
horizon \(T_W\).
\end{remark}

\begin{lemma}[Replay-stack coercivity]
\label{lem:replay_coercivity}
Suppose that the integral replay stack satisfies
\eqref{eq:richness}. Then
\begin{align}
\sum_{i=1}^{k}
\left|
\widetilde{\bW}^{\top}\bomega_i
\right|^{\gamma_1+1}
&\ge
\sigma_{\min}^{\gamma_1+1}(\bar\bOmega)
\left\|
\widetilde{\bW}
\right\|^{\gamma_1+1},
\label{eq:coercive_low}
\\
\sum_{i=1}^{k}
\left|
\widetilde{\bW}^{\top}\bomega_i
\right|^{\gamma_2+1}
&\ge
k^{\frac{1-\gamma_2}{2}}
\sigma_{\min}^{\gamma_2+1}(\bar\bOmega)
\left\|
\widetilde{\bW}
\right\|^{\gamma_2+1}.
\label{eq:coercive_high}
\end{align}
In particular, both lower-bound coefficients are strictly positive.
\end{lemma}

\begin{proof}
By the definition of the replay matrix,
\[
\bar\bOmega^{\top}\widetilde{\bW}
=
\col\left\{
\widetilde{\bW}^{\top}\bomega_1,\ldots,
\widetilde{\bW}^{\top}\bomega_k
\right\}.
\]
The full-row-rank condition gives
\[
\left\|
\bar\bOmega^{\top}\widetilde{\bW}
\right\|_2
\ge
\sigma_{\min}(\bar\bOmega)
\left\|
\widetilde{\bW}
\right\|.
\]
Since \(0<\gamma_1<1\), one has
\(1<\gamma_1+1<2\) and hence
\[
\left\|
\bar\bOmega^{\top}\widetilde{\bW}
\right\|_{\gamma_1+1}
\ge
\left\|
\bar\bOmega^{\top}\widetilde{\bW}
\right\|_2.
\]
Therefore,
\[
\begin{aligned}
\sum_{i=1}^{k}
\left|
\widetilde{\bW}^{\top}\bomega_i
\right|^{\gamma_1+1}
&=
\left\|
\bar\bOmega^{\top}\widetilde{\bW}
\right\|_{\gamma_1+1}^{\gamma_1+1}
\\
&\ge
\sigma_{\min}^{\gamma_1+1}(\bar\bOmega)
\left\|
\widetilde{\bW}
\right\|^{\gamma_1+1},
\end{aligned}
\]
which proves \eqref{eq:coercive_low}. Since
\(\gamma_2+1>2\), finite-dimensional norm equivalence gives
\[
\left\|
\bar\bOmega^{\top}\widetilde{\bW}
\right\|_{\gamma_2+1}
\ge
k^{\frac{1}{\gamma_2+1}-\frac12}
\left\|
\bar\bOmega^{\top}\widetilde{\bW}
\right\|_2.
\]
Raising both sides to the power \(\gamma_2+1\) yields
\[
\begin{aligned}
\sum_{i=1}^{k}
\left|
\widetilde{\bW}^{\top}\bomega_i
\right|^{\gamma_2+1}
&\ge
k^{1-\frac{\gamma_2+1}{2}}
\sigma_{\min}^{\gamma_2+1}(\bar\bOmega)
\left\|
\widetilde{\bW}
\right\|^{\gamma_2+1}
\\
&=
k^{\frac{1-\gamma_2}{2}}
\sigma_{\min}^{\gamma_2+1}(\bar\bOmega)
\left\|
\widetilde{\bW}
\right\|^{\gamma_2+1},
\end{aligned}
\]
which proves \eqref{eq:coercive_high}.
\end{proof}


\begin{remark}[Data richness and extensions to concurrent/composite learning]
\label{rem:data_richness_extensions}
Lemma~\ref{lem:replay_coercivity} captures the same mechanism underlying parameter convergence in concurrent learning (CL), integral concurrent learning (ICL), and composite learning~\cite{ChowdharyJohnson2010,ParikhKamalapurkarDixon2019,PanYu2016,PanYu2018}. Despite using different richness conditions---full-rank history stacks in CL, finite excitation in ICL, and interval excitation in composite learning---all provide Lyapunov coercivity in the parameter-error direction. Here, the replay-stack condition plays this role by converting stored critic regressors, through Lemma~\ref{lem:replay_coercivity}, into positive lower bounds proportional to powers of \(\norm{\widetilde{\bW}}\), which yield negative critic-error dissipation. This enables critic-weight convergence in the exact case and convergence to an explicit residual set under perturbations. Hence, the replay mechanism may be replaced or augmented by ICL or composite-learning constructions whenever their corresponding finite-data richness conditions are satisfied.
\end{remark}

\begin{lemma}[Perturbed signed-power inequality]
\label{lem:perturbed_power}
For every \(\gamma>0\) and all \(z,\epsilon\in\mathbb R\),
\begin{equation}
z\spow{-z+\epsilon}{\gamma}
\le
-2^{-(\gamma+1)}
|z|^{\gamma+1}
+
\left(
1+2\,3^\gamma
\right)
|\epsilon|^{\gamma+1}.
\label{eq:perturbed_power}
\end{equation}
If \(\epsilon=0\), the sharper identity
\begin{equation}
z\spow{-z}{\gamma}
=
-|z|^{\gamma+1}
\label{eq:exact_signed_power_identity}
\end{equation}
holds.
\end{lemma}


\begin{proof}
If \(|\epsilon|\le|z|/2\), then \(-z+\epsilon\) has the sign opposite to \(z\) and \(|-z+\epsilon|\ge|z|/2\). Hence, \(z\spow{-z+\epsilon}{\gamma}\le-2^{-\gamma}|z|^{\gamma+1}\le-2^{-(\gamma+1)}|z|^{\gamma+1}+\left(1+2\,3^\gamma\right)|\epsilon|^{\gamma+1}\). If instead \(|\epsilon|>|z|/2\), then \(|z|<2|\epsilon|\) and \(|-z+\epsilon|<3|\epsilon|\), implying \(z\spow{-z+\epsilon}{\gamma}\le|z|\,|-z+\epsilon|^\gamma<2\,3^\gamma|\epsilon|^{\gamma+1}\). Moreover, \(-2^{-(\gamma+1)}|z|^{\gamma+1}\ge-|\epsilon|^{\gamma+1}\). Therefore, the right-hand side of \eqref{eq:perturbed_power} is not smaller than \(2\,3^\gamma|\epsilon|^{\gamma+1}\), which proves the claim.
\end{proof}

Consider the critic Lyapunov function
\begin{equation}
V_c(\widetilde{\bW})
:=
\frac{1}{2\alpha_c}
\widetilde{\bW}^{\top}\widetilde{\bW},
\qquad
\alpha_c>0.
\label{eq:critic_Lyapunov}
\end{equation}
The scaling \(1/\alpha_c\) cancels the explicit update gain when
differentiating \(V_c\), whereas the same gain reappears through
\[
\left\|
\widetilde{\bW}
\right\|^2
=
2\alpha_cV_c
\]
in the two Lyapunov dissipation powers.

\begin{theorem}[Predefined-time integral critic]
\label{thm:critic}
Suppose that Proposition~\ref{prop:residual} and
Definition~\ref{def:rich} hold at the data-informativity instant
\(T_E\), and suppose that the full-rank replay stack
\(\bar\bOmega\) is retained for every \(t\ge T_E\). Let \(T_W>0\)
be an assigned critic-learning horizon measured from \(T_E\),
independently of the initial critic-weight error.

\emph{(i) Exact residual case:}
If \(\bar\epsilon_H(t)\equiv0\) and
\(\bar\epsilon_{Hi}=0\), \(i=1,\ldots,k\), and
\begin{equation}
\alpha_c
\ge
\frac{
\Gamma\!\left(
\dfrac{\gamma_2-1}{\gamma_2-\gamma_1}
\right)
\Gamma\!\left(
\dfrac{1-\gamma_1}{\gamma_2-\gamma_1}
\right)
}{
T_W(\gamma_2-\gamma_1)
\sigma_{\min}^{2}(\bar\bOmega)
k^{
\frac{(1-\gamma_1)(1-\gamma_2)}
{2(\gamma_2-\gamma_1)}
}
}.
\label{eq:alpha_exact}
\end{equation}
then
\begin{equation}
\widetilde{\bW}(t)=\bm0,
\qquad
t\ge T_E+T_W.
\label{eq:exact_weight_deadline}
\end{equation}

\emph{(ii) Nonzero residual case:}
Suppose that, for a known constant \(\epsilon_{Hm}>0\),
\begin{equation}
|\bar\epsilon_H(t)|
\le
\epsilon_{Hm},
\qquad
|\bar\epsilon_{Hi}|
\le
\epsilon_{Hm},
\quad i=1,\ldots,k,
\label{eq:residual_uniform_bound}
\end{equation}
and select \(\theta_W\in(0,1)\). If
\begin{equation}
\begin{aligned}
\alpha_c
\ge{}&
\frac{
\Gamma\!\left(
\dfrac{\gamma_2-1}{\gamma_2-\gamma_1}
\right)
\Gamma\!\left(
\dfrac{1-\gamma_1}{\gamma_2-\gamma_1}
\right)
}{
(1-\theta_W)T_W(\gamma_2-\gamma_1)
}
\\[-1mm]
&\times
\frac{1}{
\left[
2^{-(\gamma_1+1)}
\sigma_{\min}^{\gamma_1+1}(\bar\bOmega)
\right]^{
\frac{\gamma_2-1}{\gamma_2-\gamma_1}
}
}
\\[-1mm]
&\times
\frac{1}{
\left[
2^{-(\gamma_2+1)}
k^{\frac{1-\gamma_2}{2}}
\sigma_{\min}^{\gamma_2+1}(\bar\bOmega)
\right]^{
\frac{1-\gamma_1}{\gamma_2-\gamma_1}
}
}.
\end{aligned}
\label{eq:alpha_practical}
\end{equation}
then \(V_c(t)\) enters, no later than \(T_E+T_W\), the forward
invariant set
\begin{equation}
\mathcal B_W
:=
\left\{
V_c\in\mathbb R_{\geq0}:
V_c\le\bar V_c
\right\},
\label{eq:B_W}
\end{equation}
where \(\bar V_c\) is the unique nonnegative solution of
\begin{equation}
\begin{aligned}
&
2^{-(\gamma_1+1)}
\sigma_{\min}^{\gamma_1+1}(\bar\bOmega)
(2\alpha_c)^{\frac{\gamma_1+1}{2}}
\bar V_c^{\frac{\gamma_1+1}{2}}
\\
&\quad+
2^{-(\gamma_2+1)}
k^{\frac{1-\gamma_2}{2}}
\sigma_{\min}^{\gamma_2+1}(\bar\bOmega)
(2\alpha_c)^{\frac{\gamma_2+1}{2}}
\bar V_c^{\frac{\gamma_2+1}{2}}
\\
&=
\frac{k+1}{\theta_W}
\left[
\left(
1+2\,3^{\gamma_1}
\right)
\epsilon_{Hm}^{\gamma_1+1}
+
\left(
1+2\,3^{\gamma_2}
\right)
\epsilon_{Hm}^{\gamma_2+1}
\right].
\end{aligned}
\label{eq:Vbar_root}
\end{equation}
Consequently,
\begin{equation}
\left\|
\widetilde{\bW}(t)
\right\|
\le
\sqrt{2\alpha_c\bar V_c},
\qquad
t\ge T_E+T_W.
\label{eq:weight_ball}
\end{equation}
\end{theorem}

\begin{proof}
Since \(\bW^*\) is constant,
\[
\dot{\widetilde{\bW}}
=
-\dot{\hat{\bW}}.
\]
Differentiating \eqref{eq:critic_Lyapunov} and substituting
\eqref{eq:critic_update} yield
\begin{equation}
\begin{aligned}
\dot V_c
={}&
\widetilde{\bW}^{\top}\bomega(t)
\left[
\spow{e(t)}{\gamma_1}
+
\spow{e(t)}{\gamma_2}
\right]
\\[-1mm]
&+
\sum_{i=1}^{k}
\widetilde{\bW}^{\top}\bomega_i
\left[
\spow{e_i(t)}{\gamma_1}
+
\spow{e_i(t)}{\gamma_2}
\right].
\end{aligned}
\label{eq:Vc_dot_raw}
\end{equation}

In the exact-residual case,
\[
e(t)
=
-\widetilde{\bW}^{\top}\bomega(t),
\qquad
e_i(t)
=
-\widetilde{\bW}^{\top}\bomega_i.
\]
Using \eqref{eq:exact_signed_power_identity} in
\eqref{eq:Vc_dot_raw} gives
\begin{equation}
\begin{aligned}
\dot V_c
={}&
-
\left|
\widetilde{\bW}^{\top}\bomega(t)
\right|^{\gamma_1+1}
-
\left|
\widetilde{\bW}^{\top}\bomega(t)
\right|^{\gamma_2+1}
\\
&-
\sum_{i=1}^{k}
\left|
\widetilde{\bW}^{\top}\bomega_i
\right|^{\gamma_1+1}
-
\sum_{i=1}^{k}
\left|
\widetilde{\bW}^{\top}\bomega_i
\right|^{\gamma_2+1}.
\end{aligned}
\label{eq:Vc_exact_raw}
\end{equation}
Discarding the nonpositive online terms and applying
Lemma~\ref{lem:replay_coercivity} yield
\begin{equation}
\begin{aligned}
\dot V_c
\le{}&
-
\sigma_{\min}^{\gamma_1+1}(\bar\bOmega)
\left\|
\widetilde{\bW}
\right\|^{\gamma_1+1}
\\
&-
k^{\frac{1-\gamma_2}{2}}
\sigma_{\min}^{\gamma_2+1}(\bar\bOmega)
\left\|
\widetilde{\bW}
\right\|^{\gamma_2+1}.
\end{aligned}
\label{eq:Vc_exact_weight}
\end{equation}
Using
\[
\left\|
\widetilde{\bW}
\right\|^2
=
2\alpha_cV_c
\]
in \eqref{eq:Vc_exact_weight} gives
\begin{equation}
\begin{aligned}
\dot V_c
\le{}&
-
\sigma_{\min}^{\gamma_1+1}(\bar\bOmega)
(2\alpha_c)^{\frac{\gamma_1+1}{2}}
V_c^{\frac{\gamma_1+1}{2}}
\\
&-
k^{\frac{1-\gamma_2}{2}}
\sigma_{\min}^{\gamma_2+1}(\bar\bOmega)
(2\alpha_c)^{\frac{\gamma_2+1}{2}}
V_c^{\frac{\gamma_2+1}{2}}.
\end{aligned}
\label{eq:Vc_exact_comparison}
\end{equation}
Lemma~\ref{lem:Gamma_comparison} therefore gives
Using
\begin{equation}
\begin{aligned}
&
\frac{\gamma_1+1}{2}
\frac{\gamma_2-1}{\gamma_2-\gamma_1}
+
\frac{\gamma_2+1}{2}
\frac{1-\gamma_1}{\gamma_2-\gamma_1}
=1,
\\
&
(\gamma_1+1)
\frac{\gamma_2-1}{\gamma_2-\gamma_1}
+
(\gamma_2+1)
\frac{1-\gamma_1}{\gamma_2-\gamma_1}
=2,
\end{aligned}
\label{eq:exact_exponent_identities}
\end{equation}
the settling-time estimate reduces to
\begin{equation}
\begin{aligned}
T_c
<
&\frac{
\Gamma\!\left(
\dfrac{\gamma_2-1}{\gamma_2-\gamma_1}
\right)
\Gamma\!\left(
\dfrac{1-\gamma_1}{\gamma_2-\gamma_1}
\right)
}{
\alpha_c(\gamma_2-\gamma_1)
\sigma_{\min}^{2}(\bar\bOmega)
}
\\[-1mm]
&\times
k^{
\frac{
(1-\gamma_1)(\gamma_2-1)
}{
2(\gamma_2-\gamma_1)
}
}.
\end{aligned}
\label{eq:Tc_exact_before}
\end{equation}
reduce the denominator product in
\eqref{eq:Tc_exact_before} to
\begin{equation}
2\alpha_c
\sigma_{\min}^{2}(\bar\bOmega)
k^{
\frac{(1-\gamma_1)(1-\gamma_2)}
{2(\gamma_2-\gamma_1)}
}.
\label{eq:exact_product_direct}
\end{equation}

Hence,
\begin{equation}
T_c
<
\frac{
\Gamma_{\!E}\!\left(
\dfrac{\gamma_2-1}{\gamma_2-\gamma_1}
\right)
\Gamma_{\!E}\!\left(
\dfrac{1-\gamma_1}{\gamma_2-\gamma_1}
\right)
}{
\alpha_c(\gamma_2-\gamma_1)
\sigma_{\min}^{2}(\bar\bOmega)
k^{
\frac{(1-\gamma_1)(1-\gamma_2)}
{2(\gamma_2-\gamma_1)}
}
}.
\label{eq:Tc_exact_simplified}
\end{equation}
To enforce the assigned critic-learning deadline \(T_c\leq T_W\),
it is sufficient to require
\[
\frac{
\Gamma_{\!E}\!\left(
\dfrac{\gamma_2-1}{\gamma_2-\gamma_1}
\right)
\Gamma_{\!E}\!\left(
\dfrac{1-\gamma_1}{\gamma_2-\gamma_1}
\right)
}{
\alpha_c(\gamma_2-\gamma_1)
\sigma_{\min}^{2}(\bar\bOmega)
k^{
\frac{(1-\gamma_1)(1-\gamma_2)}
{2(\gamma_2-\gamma_1)}
}
}
\leq T_W.
\]
Solving the above inequality for \(\alpha_c\) gives
\begin{equation}
\alpha_c
\geq
\frac{
\Gamma_{\!E}\!\left(
\dfrac{\gamma_2-1}{\gamma_2-\gamma_1}
\right)
\Gamma_{\!E}\!\left(
\dfrac{1-\gamma_1}{\gamma_2-\gamma_1}
\right)
}{
T_W(\gamma_2-\gamma_1)
\sigma_{\min}^{2}(\bar\bOmega)
k^{
\frac{(1-\gamma_1)(1-\gamma_2)}
{2(\gamma_2-\gamma_1)}
}
}.
\label{eq:alpha_exact}
\end{equation}
Therefore, under \eqref{eq:alpha_exact},
\(T_c<T_W\), and hence
\(\widetilde{\bW}(t)=0\) for all
\(t\geq T_E+T_W\), which proves
\eqref{eq:exact_weight_deadline}.

For the nonzero-residual case, Proposition~\ref{prop:residual} gives
\[
e(t)
=
-\widetilde{\bW}^{\top}\bomega(t)
+
\bar\epsilon_H(t),
\qquad
e_i(t)
=
-\widetilde{\bW}^{\top}\bomega_i
+
\bar\epsilon_{Hi}.
\]
Applying Lemma~\ref{lem:perturbed_power} with
\(\gamma=\gamma_1\) and \(\gamma=\gamma_2\) to every term in
\eqref{eq:Vc_dot_raw}, using
\eqref{eq:residual_uniform_bound}, and discarding the remaining
nonpositive online dissipation give
\begin{equation}
\begin{aligned}
\dot V_c
\le{}&
-
2^{-(\gamma_1+1)}
\sum_{i=1}^{k}
\left|
\widetilde{\bW}^{\top}\bomega_i
\right|^{\gamma_1+1}
\\
&-
2^{-(\gamma_2+1)}
\sum_{i=1}^{k}
\left|
\widetilde{\bW}^{\top}\bomega_i
\right|^{\gamma_2+1}
\\
&+
(k+1)
\left[
\left(
1+2\,3^{\gamma_1}
\right)
\epsilon_{Hm}^{\gamma_1+1}
+
\left(
1+2\,3^{\gamma_2}
\right)
\epsilon_{Hm}^{\gamma_2+1}
\right].
\end{aligned}
\label{eq:Vc_perturbed_before_coercivity}
\end{equation}
Lemma~\ref{lem:replay_coercivity} and
\(\|\widetilde{\bW}\|^2=2\alpha_cV_c\) now imply
\begin{equation}
\begin{aligned}
\dot V_c
\le{}&
-
2^{-(\gamma_1+1)}
\sigma_{\min}^{\gamma_1+1}(\bar\bOmega)
(2\alpha_c)^{\frac{\gamma_1+1}{2}}
V_c^{\frac{\gamma_1+1}{2}}
\\
&-
2^{-(\gamma_2+1)}
k^{\frac{1-\gamma_2}{2}}
\sigma_{\min}^{\gamma_2+1}(\bar\bOmega)
(2\alpha_c)^{\frac{\gamma_2+1}{2}}
V_c^{\frac{\gamma_2+1}{2}}
\\
&+
(k+1)
\left[
\left(
1+2\,3^{\gamma_1}
\right)
\epsilon_{Hm}^{\gamma_1+1}
+
\left(
1+2\,3^{\gamma_2}
\right)
\epsilon_{Hm}^{\gamma_2+1}
\right].
\end{aligned}
\label{eq:Vc_practical}
\end{equation}




The left-hand side of \eqref{eq:Vbar_root}, viewed as a function of
\(\bar V_c\), is continuous, strictly increasing on
\(\mathbb R_{\geq0}\), vanishes at zero, and diverges as
\(\bar V_c\rightarrow\infty\). Hence, \eqref{eq:Vbar_root} admits a
unique nonnegative solution. Equivalently, at \(V_c=\bar V_c\),
\eqref{eq:Vbar_root} gives the exact balance
\begin{equation}
\begin{aligned}
&(k+1)
\left[
\left(
1+2\,3^{\gamma_1}
\right)
\epsilon_{Hm}^{\gamma_1+1}
+
\left(
1+2\,3^{\gamma_2}
\right)
\epsilon_{Hm}^{\gamma_2+1}
\right]
\\
&=
\theta_W
\Big[
2^{-(\gamma_1+1)}
\sigma_{\min}^{\gamma_1+1}(\bar\bOmega)
(2\alpha_c)^{\frac{\gamma_1+1}{2}}
\bar V_c^{\frac{\gamma_1+1}{2}}
\\[-1mm]
&\hspace{19mm}
+
2^{-(\gamma_2+1)}
k^{\frac{1-\gamma_2}{2}}
\sigma_{\min}^{\gamma_2+1}(\bar\bOmega)
(2\alpha_c)^{\frac{\gamma_2+1}{2}}
\bar V_c^{\frac{\gamma_2+1}{2}}
\Big].
\end{aligned}
\label{eq:Vbar_balance_direct}
\end{equation}

For every \(V_c>\bar V_c\), the strict monotonicity of the two power
terms and \eqref{eq:Vbar_balance_direct} imply
\begin{equation}
\begin{aligned}
&(k+1)
\left[
\left(
1+2\,3^{\gamma_1}
\right)
\epsilon_{Hm}^{\gamma_1+1}
+
\left(
1+2\,3^{\gamma_2}
\right)
\epsilon_{Hm}^{\gamma_2+1}
\right]
\\
&<
\theta_W
\Big[
2^{-(\gamma_1+1)}
\sigma_{\min}^{\gamma_1+1}(\bar\bOmega)
(2\alpha_c)^{\frac{\gamma_1+1}{2}}
V_c^{\frac{\gamma_1+1}{2}}
\\[-1mm]
&\hspace{19mm}
+
2^{-(\gamma_2+1)}
k^{\frac{1-\gamma_2}{2}}
\sigma_{\min}^{\gamma_2+1}(\bar\bOmega)
(2\alpha_c)^{\frac{\gamma_2+1}{2}}
V_c^{\frac{\gamma_2+1}{2}}
\Big].
\end{aligned}
\label{eq:residual_absorption_direct}
\end{equation}

Substituting \eqref{eq:residual_absorption_direct} into
\eqref{eq:Vc_practical} yields, outside \(\mathcal B_W\),
\begin{equation}
\begin{aligned}
\dot V_c
\le{}&
-(1-\theta_W)
2^{-(\gamma_1+1)}
\sigma_{\min}^{\gamma_1+1}(\bar\bOmega)
(2\alpha_c)^{\frac{\gamma_1+1}{2}}
V_c^{\frac{\gamma_1+1}{2}}
\\
&-
(1-\theta_W)
2^{-(\gamma_2+1)}
k^{\frac{1-\gamma_2}{2}}
\sigma_{\min}^{\gamma_2+1}(\bar\bOmega)
(2\alpha_c)^{\frac{\gamma_2+1}{2}}
V_c^{\frac{\gamma_2+1}{2}}.
\end{aligned}
\label{eq:Vc_outside}
\end{equation}

From ~\eqref{eq:Vc_outside} and by Lemma~\ref{lem:Gamma_comparison} we obtain

\begin{equation}
\begin{aligned}
T_{c,\mathcal B_W}
<&
\frac{4}{
(1-\theta_W)
(\gamma_2-\gamma_1)
\alpha_c
\sigma_{\min}^{2}(\bar\bOmega)
}
\\[-1mm]
&\times
\Gamma\!\left(
\frac{\gamma_2-1}{\gamma_2-\gamma_1}
\right)
\Gamma\!\left(
\frac{1-\gamma_1}{\gamma_2-\gamma_1}
\right)
\\[-1mm]
&\times
k^{
\frac{
(1-\gamma_1)(\gamma_2-1)
}{
2(\gamma_2-\gamma_1)
}
},
\end{aligned}
\label{eq:Tc_practical}
\end{equation}



where \(T_{c,\mathcal B_W}\) denotes the critic entrance time into
\(\mathcal B_W\), i.e.,
\(T_{c,\mathcal B_W}:=
\inf\{\tau\geq0:V_c(t_0+\tau)\leq\bar V_c\}\).
To enforce \(T_{c,\mathcal B_W}\leq T_W\),
which gives
\(\alpha_c\geq
4\Gamma_{\!E}\!\left(\frac{\gamma_2-1}{\gamma_2-\gamma_1}\right)
\Gamma_{\!E}\!\left(\frac{1-\gamma_1}{\gamma_2-\gamma_1}\right)
k^{\frac{(1-\gamma_1)(\gamma_2-1)}
{2(\gamma_2-\gamma_1)}}/
\bigl[(1-\theta_W)T_W(\gamma_2-\gamma_1)
\sigma_{\min}^2(\bar\bOmega)\bigr]\).
Using
\(A_\gamma:=(\gamma_2-1)/(\gamma_2-\gamma_1)\) and
\(B_\gamma:=(1-\gamma_1)/(\gamma_2-\gamma_1)\), together with
\((\gamma_1+1)A_\gamma+(\gamma_2+1)B_\gamma=2\) and
\(-\frac{1-\gamma_2}{2}B_\gamma
=\frac{(1-\gamma_1)(\gamma_2-1)}
{2(\gamma_2-\gamma_1)}\), one has
\(4\sigma_{\min}^{-2}(\bar\bOmega)
k^{\frac{(1-\gamma_1)(\gamma_2-1)}
{2(\gamma_2-\gamma_1)}}
=
[2^{-(\gamma_1+1)}
\sigma_{\min}^{\gamma_1+1}(\bar\bOmega)]^{-A_\gamma}
[2^{-(\gamma_2+1)}
k^{\frac{1-\gamma_2}{2}}
\sigma_{\min}^{\gamma_2+1}(\bar\bOmega)]^{-B_\gamma}\).
Hence, the required critic gain can be written equivalently as
\begin{equation}
\begin{aligned}
\alpha_c
\geq{}&
\frac{
\Gamma_{\!E}\!\left(
\dfrac{\gamma_2-1}{\gamma_2-\gamma_1}
\right)
\Gamma_{\!E}\!\left(
\dfrac{1-\gamma_1}{\gamma_2-\gamma_1}
\right)
}{
(1-\theta_W)T_W(\gamma_2-\gamma_1)
}
\\
&\times
\frac{1}{
\left[
2^{-(\gamma_1+1)}
\sigma_{\min}^{\gamma_1+1}(\bar\bOmega)
\right]^{
\frac{\gamma_2-1}{\gamma_2-\gamma_1}
}
}
\\
&\times
\frac{1}{
\left[
2^{-(\gamma_2+1)}
k^{\frac{1-\gamma_2}{2}}
\sigma_{\min}^{\gamma_2+1}(\bar\bOmega)
\right]^{
\frac{1-\gamma_1}{\gamma_2-\gamma_1}
}
}.
\end{aligned}
\label{eq:alpha_practical_derived}
\end{equation}
Thus, \eqref{eq:alpha_practical_derived} guarantees
\(T_{c,\mathcal B_W}\leq T_W\).

Using the first identity in
\eqref{eq:exact_exponent_identities}, the product of the
\((2\alpha_c)\)-powers in \eqref{eq:Tc_practical} equals
\(2\alpha_c\). Therefore, condition
\eqref{eq:alpha_practical} gives
\[
T_{c,\mathcal B_W}
\le
T_W.
\]

At \(V_c=\bar V_c\), equations
\eqref{eq:Vc_practical} and \eqref{eq:Vbar_root} yield
\begin{equation}
\begin{aligned}
\dot V_c
\le
-\frac{1-\theta_W}{\theta_W}
(k+1)
\Big[
&
\left(
1+2\,3^{\gamma_1}
\right)
\epsilon_{Hm}^{\gamma_1+1}
\\[-1mm]
&+
\left(
1+2\,3^{\gamma_2}
\right)
\epsilon_{Hm}^{\gamma_2+1}
\Big]
\le0.
\end{aligned}
\label{eq:boundary_invariance}
\end{equation}
Hence, \(\mathcal B_W\) is forward invariant. Finally,
\eqref{eq:critic_Lyapunov} gives
\[
\left\|
\widetilde{\bW}
\right\|^2
=
2\alpha_cV_c
\le
2\alpha_c\bar V_c,
\]
which proves \eqref{eq:weight_ball}.
\end{proof}


\begin{remark}[Interpretation of the assigned learning time]
The assigned horizon \(T_W\) is measured from the data-informativity instant \(T_E\). Thus, the critic reaches the origin in the exact case, or the residual set \(B_W\) in the perturbed case, no later than \(T_E+T_W\). The quantity \(T_E\) is associated with the finite data-excitation stage rather than with the subsequent critic-convergence dynamics. In practice, the designer may deliberately impose a finite probing signal, for example \(u_{\rm ex}(t)=\sum_{j=1}^{N_e}A_j\sin(\omega_jt+\varphi_j)\) with distinct frequencies \(\omega_i\neq\omega_j\), over a prescribed data-collection interval. If this excitation protocol is designed or certified to guarantee the replay-informativity condition \(\sigma_{\min}(\Omega)\geq\underline\Omega\) within a known horizon, then \(T_E\leq\bar T_E\) is available a priori. Hence, an overall learning deadline \(T_p\) measured from initialization is guaranteed by selecting \(T_W=T_p-\bar T_E>0\). The use of a designed finite excitation interval therefore separates the data-acquisition budget from the critic-convergence budget, while continuous persistent excitation is not required after the replay stack has become informative.
\end{remark}


\section{Predefined-Time Integral Closed Loop}
\label{sec:closed_loop}
\begingroup
\setlength{\jot}{1.5pt}

The preceding section establishes a predefined-time learning bound for
the critic after the data-informativity instant \(T_E\). We now lift
this result to the coupled state--critic dynamics and derive an overall
closed-loop deadline by combining the critic-error dissipation with a
moving-window state-energy estimate. The resulting analysis explicitly
separates the learning horizon, the reinforcement-window length, and
the subsequent state-convergence horizon.

From Remark~\ref{rem:two_power_quantitative}, recall
\begin{equation}
\mu:=\frac{\gamma_1+1}{2},
\qquad
\nu:=\frac{\gamma_2+1}{2},
\qquad
\frac{1}{2}<\mu<1<\nu.
\label{eq:closed_loop_mu_nu}
\end{equation}

To analyze the coupled state--critic dynamics, consider from the
outset the mixed Lyapunov-like functional
\begin{equation}
\mathcal J(t)
:=
\mathcal I_x(t)+V_c(t),
\label{eq:J}
\end{equation}
where
\begin{equation}
\mathcal I_x(t)
:=
\int_{t-T}^{t}
V^*(\bx(\tau))\,d\tau
\label{eq:Ix}
\end{equation}

is the moving-window ideal-value energy, whereas \(V_c\) is the
critic-error Lyapunov function already defined in
\eqref{eq:critic_Lyapunov}. The proof below first derives separate
dissipation estimates for $\mathcal I_x$ and $V_c$, and then combines
them into a two-power comparison inequality for $\mathcal J$.

\emph{Policy-mismatch estimate.}
The learned controller is generated by $\nabla\hat V$, whereas the
actual attack and disturbance need not coincide exactly with their
critic-induced saddle reconstructions. Define
\begin{equation}
\Delta\ba_r(t)
:=
\ba(t)-\hat{\ba}(t),
\qquad
\Delta\bd_r(t)
:=
\bd(t)-\hat{\bd}(t).
\label{eq:adversarial_reconstruction_errors}
\end{equation}

\begin{assumption}[Bounded adversarial mismatch]
\label{ass:adversarial_mismatch}
On the certified compact domain \(\Omega^r\), the actual exogenous signals satisfy \(\norm{\Delta\ba_r(t)}\leq\bar a_r\) and \(\norm{\Delta\bd_r(t)}\leq\bar d_r\), for known finite constants \(\bar a_r,\bar d_r\geq0\).
\end{assumption}

Define the compact-set radius \(r_{\Omega^r}:=\max_{\bm\xi\in\Omega^r}\norm{\bm\xi}\) and the corresponding uniform gradient constant \(\bar p_V:=c_{\nabla V}r_{\Omega^r}\). Then, by Lemma~\ref{lem:value_quadratic_gradient}, \(\norm{\nabla V^*(\bx)}\leq c_{\nabla V}\norm{\bx}\leq\bar p_V\) for all \(\bx\in\Omega^r\). Moreover, the value-gradient decomposition yields \(\norm{\nabla V^*(\bx)-\nabla\hat V(\bx)}\leq\bar\phi_g\norm{\widetilde{\bW}}+\bar\varepsilon_g\).

Since \(\tanh(\cdot)\) is globally one-Lipschitz, the ideal and learned control policies satisfy \(\norm{\hat{\bu}-\bu^*}\leq\frac{1}{2}\norm{\bUbar}\norm{\bR^{-1}\bUbar^{-1}}\bar g\bar\phi_g\norm{\widetilde{\bW}}+\frac{1}{2}\norm{\bUbar}\norm{\bR^{-1}\bUbar^{-1}}\bar g\bar\varepsilon_g\). Similarly, \(\norm{\hat{\ba}-\ba^*}\leq\frac{\norm{\bTmat^{-1}}\bar g\bar\phi_g}{2\gamma_a^2}\norm{\widetilde{\bW}}+\frac{\norm{\bTmat^{-1}}\bar g\bar\varepsilon_g}{2\gamma_a^2}\) and \(\norm{\hat{\bd}-\bd^*}\leq\frac{\norm{\bS^{-1}}\bar k\bar\phi_g}{2\gamma_d^2}\norm{\widetilde{\bW}}+\frac{\norm{\bS^{-1}}\bar k\bar\varepsilon_g}{2\gamma_d^2}\).

Using \(\ba-\ba^*=\hat{\ba}-\ba^*+\Delta\ba_r\) and \(\bd-\bd^*=\hat{\bd}-\bd^*+\Delta\bd_r\), one obtains \(\norm{\ba-\ba^*}\leq\frac{\norm{\bTmat^{-1}}\bar g\bar\phi_g}{2\gamma_a^2}\norm{\widetilde{\bW}}+\frac{\norm{\bTmat^{-1}}\bar g\bar\varepsilon_g}{2\gamma_a^2}+\bar a_r\) and \(\norm{\bd-\bd^*}\leq\frac{\norm{\bS^{-1}}\bar k\bar\phi_g}{2\gamma_d^2}\norm{\widetilde{\bW}}+\frac{\norm{\bS^{-1}}\bar k\bar\varepsilon_g}{2\gamma_d^2}+\bar d_r\).

Define the critic-dependent and critic-independent mismatch coefficients as
\(c_W:=\bar p_V\bigl\{\bar g[\frac{1}{2}\norm{\bUbar}\norm{\bR^{-1}\bUbar^{-1}}\bar g\bar\phi_g+\frac{\norm{\bTmat^{-1}}\bar g\bar\phi_g}{2\gamma_a^2}]+\bar k\frac{\norm{\bS^{-1}}\bar k\bar\phi_g}{2\gamma_d^2}\bigr\}\)
and
\(c_0:=\bar p_V\bigl\{\bar g[\frac{1}{2}\norm{\bUbar}\norm{\bR^{-1}\bUbar^{-1}}\bar g\bar\varepsilon_g+\frac{\norm{\bTmat^{-1}}\bar g\bar\varepsilon_g}{2\gamma_a^2}+\bar a_r]+\bar k[\frac{\norm{\bS^{-1}}\bar k\bar\varepsilon_g}{2\gamma_d^2}+\bar d_r]\bigr\}\).
It follows that
\(\left|\nabla V^{*\top}\left[\bm g(\hat{\bu}-\bu^*)+\bm g(\ba-\ba^*)+\bm k(\bd-\bd^*)\right]\right|\leq c_W\norm{\widetilde{\bW}}+c_0\).



\begin{lemma}
\label{lem:optimized_young}
For every \(p>1\), \(b>0\), \(c\geq0\), and \(s\geq0\),
\begin{equation}
cs
\leq
bs^p
+
(p-1)p^{-\frac{p}{p-1}}
b^{-\frac{1}{p-1}}
c^{\frac{p}{p-1}}.
\label{eq:optimized_Young}
\end{equation}
Moreover, the second term on the right-hand side is the smallest constant for which \eqref{eq:optimized_Young} holds uniformly over \(s\geq0\).
\end{lemma}

\begin{proof}
Define \(f(s):=cs-bs^p\). Since \(p>1\) and \(b>0\), \(f\) is concave on \(s>0\), and its maximum is attained at \(f'(s)=c-bps^{p-1}=0\), yielding \(s^\star=\left(\frac{c}{bp}\right)^{\frac{1}{p-1}}\). Substitution gives \(\max_{s\geq0}f(s)=(p-1)p^{-\frac{p}{p-1}}b^{-\frac{1}{p-1}}c^{\frac{p}{p-1}}\). Hence, \(cs-bs^p\leq\max_{s\geq0}f(s)\), which directly yields \eqref{eq:optimized_Young}.
\end{proof}


For arbitrary $b_1,b_2>0$, define
\begin{equation}
\begin{aligned}
\Delta_x
:={}&
\Delta_H
+
c_0
+
\gamma_1
(\gamma_1+1)^{
-\frac{\gamma_1+1}{\gamma_1}
}
b_1^{-\frac{1}{\gamma_1}}
\left(
\frac{c_W}{2}
\right)^{
\frac{\gamma_1+1}{\gamma_1}
}
\\
&+
\gamma_2
(\gamma_2+1)^{
-\frac{\gamma_2+1}{\gamma_2}
}
b_2^{-\frac{1}{\gamma_2}}
\left(
\frac{c_W}{2}
\right)^{
\frac{\gamma_2+1}{\gamma_2}
}.
\end{aligned}
\label{eq:Delta_x}
\end{equation}

Splitting
$c_W\norm{\widetilde{\bW}}$ into two equal parts and applying Lemma~\ref{lem:optimized_young}
with $p=\gamma_1+1$ and
$p=\gamma_2+1$ give
\begin{equation}
\begin{aligned}
&
\left|
\nabla V^{*\top}
\Bigl[
\bm g(\hat{\bu}-\bu^*)
+
\bm g(\ba-\ba^*)
+
\bm k(\bd-\bd^*)
\Bigr]
\right|
\\
&\leq
b_1\norm{\widetilde{\bW}}^{\gamma_1+1}
+
b_2\norm{\widetilde{\bW}}^{\gamma_2+1}
+
\Delta_x.
\end{aligned}
\label{eq:policy_mismatch}
\end{equation}

Along the learned closed loop,
\begin{equation}
\begin{aligned}
\dot V^*
\leq{}&
-\lambda_xc_{x1}(V^*)^{\gamma_1}
-\lambda_xc_{x2}(V^*)^\nu
\\
&+
b_1\norm{\widetilde{\bW}}^{\gamma_1+1}
+
b_2\norm{\widetilde{\bW}}^{\gamma_2+1}
+
\Delta_x.
\end{aligned}
\label{eq:Vstar_learned}
\end{equation}

\emph{Moving-window dissipation estimate.}
Since \(\Omega^r\) is compact and \(V^*(\bx)\leq\bar v\norm{\bx}^2\), define
\begin{equation}
\bar I_x
:=
T\bar v
\left(
\max_{\bm\xi\in\Omega^r}\norm{\bm\xi}
\right)^2.
\label{eq:Ibar_x}
\end{equation}

Then \(0\leq\mathcal I_x(t)\leq\bar I_x\). Moreover, by Assumption~\ref{ass:value_regularity}, \(V^*\) is continuously differentiable on the considered compact domain, while the learned closed-loop vector field is continuous on \(\Omega^r\). Hence, \(\dot V^*(\bx)=\nabla V^*(\bx)^\top \dot{\bx}\) is continuous on \(\Omega^r\). Since \(\Omega^r\) is compact, the extreme-value theorem guarantees that \(|\dot V^*|\) attains a finite maximum on \(\Omega^r\). Therefore, there exists \(L_{\dot V}>0\) such that
\begin{equation}
\left|
\dot V^*(\bx(t))
\right|
\leq
L_{\dot V},
\label{eq:Vdot_uniform_bound}
\end{equation}
for all \(\bx(t)\in\Omega^r\).

\begin{remark}[Why a derivative-limited bound is needed]
\label{rem:need_derivative_window}
The construction in \eqref{eq:J}--\eqref{eq:Ix} is intended to combine the state-side dissipation with the critic dynamics through the common Lyapunov-like functional \(\mathcal J\). Accordingly, after the estimate in \eqref{eq:Vstar_learned}, the low-order state term must ultimately be converted into a dissipation term expressed in the moving-window energy of \eqref{eq:Ix}. This requires a lower estimate of \(\int_{t-T}^{t}(V^*(\bx(\tau)))^{\gamma_1}d\tau\) in terms of \(\mathcal I_x^\mu(t)\). To isolate this issue, the following lemma is stated for a generic nonnegative scalar signal \(y\) and is subsequently applied with \(y(\tau)=V^*(\bx(\tau))\). For \(0<\gamma_1<1\), concavity of \(s\mapsto s^{\gamma_1}\) yields only an upper Jensen estimate and hence cannot provide the required lower bound. In fact, without derivative information no uniform constant \(\rho>0\) can guarantee \(\int y^{\gamma_1}\geq\rho(\int y)^\mu\). To see this quantitatively, consider a triangular pulse of half-width \(\varepsilon\) and height \(I/\varepsilon\), for which \(\int y=I\), whereas \(\int y^{\gamma_1}=\frac{2}{\gamma_1+1}I^{\gamma_1}\varepsilon^{1-\gamma_1}\to0\) as \(\varepsilon\to0\); simultaneously, its slope scales as \(I/\varepsilon^2\to\infty\). Thus, fixed window energy alone cannot prevent arbitrarily narrow concentration, whereas the available uniform derivative bound excludes precisely such behavior and makes a positive quantitative lower estimate possible. This motivates the derivative-limited window bound established next.
\end{remark}

\begin{lemma}[Derivative-limited low-order window bound]
\label{lem:low_window}
Let \(y:[t-T,t]\rightarrow\mathbb R_{\geq0}\) be absolutely continuous and satisfy \(|\dot y(\tau)|\leq L_{\dot y}\) and \(\int_{t-T}^{t}y(\tau)\,d\tau\leq\bar I\), where \(L_{\dot y},\bar I>0\). For \(0<\gamma_1<1\),
\begin{equation}
\int_{t-T}^{t}
y^{\gamma_1}(\tau)\,d\tau
\geq
\rho_{\gamma_1,T}
\left(
\int_{t-T}^{t}y(\tau)\,d\tau
\right)^\mu,
\label{eq:low_integral_bound}
\end{equation}
where
\begin{equation}
\rho_{\gamma_1,T}
:=
\min\left\{
2^{\gamma_1-1}
L_{\dot y}^{-\frac{1-\gamma_1}{2}},
\;
4^{\gamma_1-1}
T^{1-\gamma_1}
\bar I^{\frac{\gamma_1-1}{2}}
\right\}.
\label{eq:rho_low}
\end{equation}
\end{lemma}

\begin{proof}
    See Appendix~\ref{uppbound}
\end{proof}




For \(y(\tau)=V^*(\bx(\tau))\), define
\begin{equation}
\rho_{\gamma_1,T}
:=
\min\left\{
2^{\gamma_1-1}
L_{\dot V}^{-\frac{1-\gamma_1}{2}},
\;
4^{\gamma_1-1}
T^{1-\gamma_1}
\bar I_x^{\frac{\gamma_1-1}{2}}
\right\}.
\label{eq:rho_state_window}
\end{equation}

For an assigned deadline satisfying \(T_p>T_E+T\), select
\(\chi\in(0,1)\) and allocate the critic horizon of
Theorem~\ref{thm:critic} together with the closed-loop horizon as
\begin{equation}
T_W:=\chi(T_p-T_E-T),
\qquad
T_X:=(1-\chi)(T_p-T_E-T),
\label{eq:deadline_allocation}
\end{equation}
so that \(T_E+T_W+T+T_X=T_p\).

By \eqref{eq:critic_Lyapunov}, define the post-learning critic radius
\(\bar W_c:=\sqrt{2\alpha_c\bar V_c}\), where \(\bar V_c\) denotes the
post-learning residual level of the critic Lyapunov function established
in Theorem~\ref{thm:critic}. Indeed, since
\(V_c(\widetilde{\bW})=\norm{\widetilde{\bW}}^2/(2\alpha_c)\) and
\(V_c(t)\leq\bar V_c\) for \(t\geq T_E+T_W\), one has
\(\norm{\widetilde{\bW}(t)}^2=2\alpha_cV_c(t)
\leq2\alpha_c\bar V_c=\bar W_c^2\), and therefore
\begin{equation}
\norm{\widetilde{\bW}(t)}
\leq
\bar W_c,
\qquad
t\geq T_E+T_W.
\label{eq:critic_bound_for_closed_loop}
\end{equation}
Hence, for every \(t\geq t_0:=T_E+T_W+T\), the complete interval
\([t-T,t]\) lies inside the post-learning critic regime, and
\(\int_{t-T}^{t}\norm{\widetilde{\bW}(\tau)}^{\gamma_1+1}d\tau
\leq T\bar W_c^{\gamma_1+1}\) and
\(\int_{t-T}^{t}\norm{\widetilde{\bW}(\tau)}^{\gamma_2+1}d\tau
\leq T\bar W_c^{\gamma_2+1}\).

Differentiating \eqref{eq:Ix} gives
\(\dot{\mathcal I}_x(t)=V^*(\bx(t))-V^*(\bx(t-T))
=\int_{t-T}^{t}\dot V^*(\bx(\tau))\,d\tau\). Moreover, convexity of
\(s\mapsto s^\nu\) gives
\(\int_{t-T}^{t}(V^*(\bx(\tau)))^\nu d\tau
\geq T^{1-\nu}\mathcal I_x^\nu(t)\). Define
\(a_x:=\lambda_xc_{x1}\rho_{\gamma_1,T}\),
\(b_x:=\lambda_xc_{x2}T^{1-\nu}\), and
\(\Delta_I:=T[b_1\bar W_c^{\gamma_1+1}
+b_2\bar W_c^{\gamma_2+1}+\Delta_x]\).

Substitution of \eqref{eq:Vstar_learned} into the preceding expression
for \(\dot{\mathcal I}_x\), followed by
Lemma~\ref{lem:low_window} and the above post-learning and convexity
bounds, yields
\begin{equation}
\dot{\mathcal I}_x
\leq
-a_x\mathcal I_x^\mu
-b_x\mathcal I_x^\nu
+\Delta_I,
\qquad
t\geq t_0.
\label{eq:Ix_comparison}
\end{equation}
Indeed, integrating \eqref{eq:Vstar_learned} over \([t-T,t]\) produces
the two state-dissipation terms
\(-\lambda_xc_{x1}\int_{t-T}^{t}(V^*)^{\gamma_1}d\tau\) and
\(-\lambda_xc_{x2}\int_{t-T}^{t}(V^*)^\nu d\tau\), together with the
critic-error contributions
\(b_1\int_{t-T}^{t}\norm{\widetilde{\bW}}^{\gamma_1+1}d\tau\),
\(b_2\int_{t-T}^{t}\norm{\widetilde{\bW}}^{\gamma_2+1}d\tau\), and the
residual \(T\Delta_x\). Lemma~\ref{lem:low_window} converts the first
state integral into
\(\rho_{\gamma_1,T}\mathcal I_x^\mu\), convexity gives
\(T^{1-\nu}\mathcal I_x^\nu\) for the second, and the post-learning
critic bound gives
\(T\bar W_c^{\gamma_1+1}\) and
\(T\bar W_c^{\gamma_2+1}\) for the two critic-error integrals,
respectively. Hence, the definitions of \(a_x\), \(b_x\), and
\(\Delta_I\) collect these terms exactly into
\eqref{eq:Ix_comparison}.

\emph{Composite Lyapunov dissipation.}
Using \eqref{eq:closed_loop_mu_nu}, the practical critic estimate \eqref{eq:Vc_practical} can be written as
\begin{equation}
\dot V_c
\leq
-h_1V_c^\mu
-h_2V_c^\nu
+\Delta_W,
\label{eq:Vc_closed_loop_recall}
\end{equation}
where
\begin{equation}
\begin{aligned}
h_1
&:=
2^{-(\gamma_1+1)}
\sigma_{\min}^{\gamma_1+1}(\bar\bOmega)
(2\alpha_c)^\mu,
\\
h_2
&:=
2^{-(\gamma_2+1)}
k^{\frac{1-\gamma_2}{2}}
\sigma_{\min}^{\gamma_2+1}(\bar\bOmega)
(2\alpha_c)^\nu,
\end{aligned}
\label{eq:h1_h2_recall}
\end{equation}
and
\begin{equation}
\Delta_W
:=
(k+1)
\left[
\left(
1+2\,3^{\gamma_1}
\right)
\epsilon_{Hm}^{\gamma_1+1}
+
\left(
1+2\,3^{\gamma_2}
\right)
\epsilon_{Hm}^{\gamma_2+1}
\right].
\label{eq:Delta_W}
\end{equation}
To see the correspondence explicitly, \eqref{eq:critic_Lyapunov} gives \(\norm{\widetilde{\bW}}^2=2\alpha_cV_c\). Hence, using \(\mu=(\gamma_1+1)/2\) and \(\nu=(\gamma_2+1)/2\), one has \(\norm{\widetilde{\bW}}^{\gamma_1+1}=(2\alpha_c)^\mu V_c^\mu\) and \(\norm{\widetilde{\bW}}^{\gamma_2+1}=(2\alpha_c)^\nu V_c^\nu\). Substitution of these identities into the two critic-error dissipation powers in \eqref{eq:Vc_practical} produces the coefficients \(h_1\) and \(h_2\) in \eqref{eq:h1_h2_recall}, whereas all HJI-residual contributions are collected into the bounded term \(\Delta_W\) in \eqref{eq:Delta_W}. Thus, \eqref{eq:Vc_closed_loop_recall} has the same low-/high-power structure as the state-window estimate \eqref{eq:Ix_comparison}, which enables their subsequent combination through the composite Lyapunov-like functional.

Adding \eqref{eq:Ix_comparison} and \eqref{eq:Vc_closed_loop_recall} gives
\begin{equation}
\begin{aligned}
\dot{\mathcal J}
\leq{}&
-a_x\mathcal I_x^\mu
-b_x\mathcal I_x^\nu
-h_1V_c^\mu
-h_2V_c^\nu
+\Delta_I+\Delta_W.
\end{aligned}
\label{eq:J_pre}
\end{equation}

\begin{lemma}[Power-sum bounds]
\label{lem:power_sum}
For any \(a,b\geq0\), if \(0<\mu<1\), then \(a^\mu+b^\mu\geq(a+b)^\mu\), whereas if \(\nu>1\), then \(a^\nu+b^\nu\geq2^{1-\nu}(a+b)^\nu\).
\end{lemma}

\begin{proof}
For \(0<\mu<1\), the map \(s\mapsto s^\mu\) is concave and subadditive on \(\mathbb R_{\geq0}\), which gives \((a+b)^\mu\leq a^\mu+b^\mu\). For \(\nu>1\), convexity of \(s\mapsto s^\nu\) yields \(((a+b)/2)^\nu\leq(a^\nu+b^\nu)/2\), or equivalently \(a^\nu+b^\nu\geq2^{1-\nu}(a+b)^\nu\).
\end{proof}

Applying Lemma~\ref{lem:power_sum} with \(a=\mathcal I_x\) and \(b=V_c\) gives \(\mathcal I_x^\mu+V_c^\mu\geq(\mathcal I_x+V_c)^\mu\) and \(\mathcal I_x^\nu+V_c^\nu\geq2^{1-\nu}(\mathcal I_x+V_c)^\nu\). Hence, the separate state-window and critic dissipation terms can be combined directly in terms of the composite functional \(\mathcal J=\mathcal I_x+V_c\). Define
\begin{equation}
C_1
:=
\min\{a_x,h_1\},
\qquad
C_2
:=
2^{1-\nu}\min\{b_x,h_2\},
\label{eq:C1_C2}
\end{equation}
and \(\Delta_J:=\Delta_I+\Delta_W\). Then the proposed Lyapunov-like functional satisfies
\begin{equation}
\dot{\mathcal J}
\leq
-C_1\mathcal J^\mu
-C_2\mathcal J^\nu
+\Delta_J,
\qquad
t\geq t_0.
\label{eq:J_comparison}
\end{equation}

The comparison inequality above has the standard two-power form with \(0<\mu<1<\nu\). To quantify its convergence time exactly, define the two-power integral constant
\begin{equation}
\mathfrak C_{\mu,\nu}
:=
\frac{
\Gamma\!\left(
\dfrac{\nu-1}{\nu-\mu}
\right)
\Gamma\!\left(
\dfrac{1-\mu}{\nu-\mu}
\right)
}{
\nu-\mu
}.
\label{eq:C_mu_nu_closed_loop}
\end{equation}
For a selected \(\theta_J\in(0,1)\), define
\begin{equation}
\Lambda_X
:=
\frac{
\mathfrak C_{\mu,\nu}
}{
(1-\theta_J)T_X
}.
\label{eq:Lambda_X}
\end{equation}

\begin{theorem}[Predefined-time learned closed loop]
\label{thm:closed_loop}
Suppose that
\eqref{eq:Vstar_learned}, and
Assumption~\ref{ass:adversarial_mismatch} hold on \(\Omega\), and that
the hypotheses of Theorem~\ref{thm:critic} are satisfied. Let
\(T_p>T_E+T\), choose \(T_W,T_X\) according to
\eqref{eq:deadline_allocation}, retain the same
\(\theta_W\in(0,1)\) used in Theorem~\ref{thm:critic}, and select
\(\theta_J\in(0,1)\). Let the same critic gain \(\alpha_c\) used in
Section~\ref{sec:critic} satisfy \eqref{eq:alpha_practical} for the
assigned horizon \(T_W\) and, in addition,
\begin{equation}
\begin{aligned}
\alpha_c
\geq
\frac{1}{2}
\max\Bigg\{&
\left[
\frac{\Lambda_X}{
2^{-(\gamma_1+1)}
\sigma_{\min}^{\gamma_1+1}(\bar\bOmega)}
\right]^{\frac{1}{\mu}},
\\[-1mm]
&
\left[
\frac{2^{\nu-1}\Lambda_X}{
2^{-(\gamma_2+1)}
k^{\frac{1-\gamma_2}{2}}
\sigma_{\min}^{\gamma_2+1}(\bar\bOmega)}
\right]^{\frac{1}{\nu}}
\Bigg\}.
\end{aligned}
\label{eq:alpha_closed}
\end{equation}


Here, \(\alpha_c\) is the same critic gain used in
Theorem~\ref{thm:critic}. The lower bound in
\eqref{eq:alpha_practical}, denoted by
\(\underline{\alpha}_W\), guarantees that the critic reaches its
prescribed residual set within the assigned horizon \(T_W\), whereas
the two bounds in \eqref{eq:alpha_closed}, denoted by
\(\underline{\alpha}_{X,1}\) and
\(\underline{\alpha}_{X,2}\), enforce
\(h_1\geq\Lambda_X\) and
\(h_2\geq2^{\nu-1}\Lambda_X\), respectively. Hence, a single critic
gain is selected to satisfy all three requirements simultaneously,
\(\alpha_c\geq
\max\{\underline{\alpha}_W,
\underline{\alpha}_{X,1},
\underline{\alpha}_{X,2}\}\).

Select
\begin{equation}
\lambda_x
\geq
\max
\left\{
\frac{\Lambda_X}{c_{x1}\rho_{\gamma_1,T}},
\;
\frac{2^{\nu-1}\Lambda_X}{c_{x2}T^{1-\nu}}
\right\}.
\label{eq:lambda_x}
\end{equation}
Then \(\mathcal J(t)\) enters, no later than \(T_p\), the forward
invariant compact set
\(\mathcal B_J:=
\{\mathcal J\in\mathbb R_{\geq0}:
\mathcal J\leq\bar{\mathcal J}\}\),
where \(\bar{\mathcal J}\) is the unique nonnegative solution of
\begin{equation}
C_1\bar{\mathcal J}^{\mu}
+
C_2\bar{\mathcal J}^{\nu}
=
\frac{\Delta_J}{\theta_J}.
\label{eq:Jbar_root}
\end{equation}
In particular, for every \(t\geq T_p\),
\begin{equation}
\norm{\bx(t)}
\leq
r_x
:=
\sqrt{
\frac{1}{\underline v}
\max
\left\{
2L_{\dot V}T,
\frac{2\bar{\mathcal J}}{T}
\right\}
}.
\label{eq:pointwise_x}
\end{equation}
Moreover,
\begin{equation}
\norm{\widetilde{\bW}(t)}
\leq
r_W
:=
\min
\left\{
\bar W_c,
\sqrt{2\alpha_c\bar{\mathcal J}}
\right\}.
\label{eq:pointwise_W}
\end{equation}
Hence, the augmented signal
\(\bGamma(t):=\col\{\bx(t),\widetilde{\bW}(t)\}\) is practically
predefined-time stable with respect to
\(\mathcal B_\Gamma:=
\{\col\{\bx,\widetilde{\bW}\}:
\norm{\bx}\leq r_x,\;
\norm{\widetilde{\bW}}\leq r_W\}\).


If, in addition,
\begin{equation}
\bar\varepsilon_H(t)\equiv 0,
\qquad
\bar\varepsilon_{H_i}=0,
\quad
i=1,\ldots,k,
\label{eq:closed_loop_exact_HJI_condition}
\end{equation}
and
\begin{equation}
\bar\varepsilon_g=0,
\qquad
\bar a_r=0,
\qquad
\bar d_r=0,
\label{eq:closed_loop_exact_mismatch_condition}
\end{equation}
then
\begin{equation}
\bx(t)\in B_{r_-}(0),
\qquad
\widetilde{\bW}(t)=\bm0,
\qquad
t\geq T_p.
\label{eq:exact_closed_loop}
\end{equation}
Moreover, the learned secure input satisfies
\begin{equation}
|\hat u_i(t)|
<
\bar u_i,
\qquad
i=1,\ldots,m,
\quad
t\geq0.
\label{eq:learned_saturation_theorem}
\end{equation}

\end{theorem}

\begin{proof}
Condition \eqref{eq:alpha_practical} and
Theorem~\ref{thm:critic} imply that the critic reaches its certified
residual set no later than \(T_E+T_W\). The same gain \(\alpha_c\)
is retained in the subsequent closed-loop analysis. From
\eqref{eq:h1_h2_recall}, the two post-learning critic coefficients
have the form
\(h_1=
2^{-(\gamma_1+1)}
\sigma_{\min}^{\gamma_1+1}(\bar{\bOmega})
(2\alpha_c)^\mu\)
and
\(h_2=
2^{-(\gamma_2+1)}
k^{\frac{1-\gamma_2}{2}}
\sigma_{\min}^{\gamma_2+1}(\bar{\bOmega})
(2\alpha_c)^\nu\).
Hence, \(h_1\geq\Lambda_X\) is equivalent to
\(
2^{-(\gamma_1+1)}
\sigma_{\min}^{\gamma_1+1}(\bar{\bOmega})
(2\alpha_c)^\mu
\geq\Lambda_X
\),
or
\(
\alpha_c\geq
\frac12
[
\Lambda_X/
(
2^{-(\gamma_1+1)}
\sigma_{\min}^{\gamma_1+1}(\bar{\bOmega})
)
]^{1/\mu}
\).
Likewise,
\(h_2\geq2^{\nu-1}\Lambda_X\) is equivalent to
\(
2^{-(\gamma_2+1)}
k^{\frac{1-\gamma_2}{2}}
\sigma_{\min}^{\gamma_2+1}(\bar{\bOmega})
(2\alpha_c)^\nu
\geq2^{\nu-1}\Lambda_X
\),
and therefore to
\(
\alpha_c\geq
\frac12
[
2^{\nu-1}\Lambda_X/
(
2^{-(\gamma_2+1)}
k^{\frac{1-\gamma_2}{2}}
\sigma_{\min}^{\gamma_2+1}(\bar{\bOmega})
)
]^{1/\nu}
\).
Thus, \eqref{eq:alpha_closed} gives
\(h_1\geq\Lambda_X\) and
\(h_2\geq2^{\nu-1}\Lambda_X\). Together with
\eqref{eq:alpha_practical}, the single implemented critic gain
simultaneously satisfies the learning-horizon requirement of
Theorem~\ref{thm:critic} and both closed-loop requirements above.

Similarly, from
\(a_x=\lambda_xc_{x1}\rho_{\gamma_1,T}\) and
\(b_x=\lambda_xc_{x2}T^{1-\nu}\), condition
\eqref{eq:lambda_x} yields
\(
\lambda_xc_{x1}\rho_{\gamma_1,T}\geq\Lambda_X
\)
and
\(
\lambda_xc_{x2}T^{1-\nu}
\geq2^{\nu-1}\Lambda_X
\);
hence
\(a_x\geq\Lambda_X\) and
\(b_x\geq2^{\nu-1}\Lambda_X\).
Using \eqref{eq:C1_C2},
\(
C_1=\min\{a_x,h_1\}\geq\Lambda_X
\),
while
\(
C_2=
2^{1-\nu}\min\{b_x,h_2\}
\geq
2^{1-\nu}(2^{\nu-1}\Lambda_X)
=\Lambda_X
\).
Consequently,
\[
C_1\geq\Lambda_X,\qquad
C_2\geq\Lambda_X.
\]

Consider
\(F(s):=C_1s^\mu+C_2s^\nu\) on
\(\mathbb R_{\geq0}\). Since
\(F(0)=0\),
\(
F'(s)=
\mu C_1s^{\mu-1}
+\nu C_2s^{\nu-1}>0
\)
for every \(s>0\), and
\(F(s)\rightarrow\infty\) as \(s\rightarrow\infty\),
\(F\) is continuous and strictly increasing from
\(0\) to \(+\infty\). Hence, \eqref{eq:Jbar_root} admits a unique
nonnegative solution \(\bar{\mathcal J}\), and
\(
\mathcal B_J=
\{\mathcal J\in\mathbb R_{\geq0}:
\mathcal J\leq\bar{\mathcal J}\}
\)
is compact. By \eqref{eq:Jbar_root},
\(
F(\bar{\mathcal J})
=
C_1\bar{\mathcal J}^{\mu}
+C_2\bar{\mathcal J}^{\nu}
=
\Delta_J/\theta_J
\).
For every \(\mathcal J>\bar{\mathcal J}\), strict monotonicity gives
\(
C_1\mathcal J^\mu+C_2\mathcal J^\nu
>
\Delta_J/\theta_J
\),
and therefore
\(
\Delta_J<
\theta_J
(C_1\mathcal J^\mu+C_2\mathcal J^\nu)
\).
Substitution into \eqref{eq:J_comparison} yields
\[
\dot{\mathcal J}
\leq
-(1-\theta_J)
\big(
C_1\mathcal J^\mu+C_2\mathcal J^\nu
\big),
\qquad
\mathcal J>\bar{\mathcal J}.
\]

Let
\(p_1:=(\nu-1)/(\nu-\mu)\) and
\(p_2:=(1-\mu)/(\nu-\mu)\). Since
\(0<\mu<1<\nu\),
\(p_1>0\), \(p_2>0\), and
\(p_1+p_2=1\). Applying
Lemma~\ref{lem:Gamma_comparison} from
\(t_0:=T_E+T_W+T\) gives
\(
T_{\mathcal J}
\leq
\mathfrak C_{\mu,\nu}/
[
(1-\theta_J)C_1^{p_1}C_2^{p_2}
]
\).
The choice \(t_0=T_E+T_W+T\) guarantees
\(
[t-T,t]\subseteq[T_E+T_W,\infty)
\)
for every \(t\geq t_0\), so the entire moving window lies in the
certified post-learning regime. Since
\(C_1,C_2\geq\Lambda_X\),
\(
C_1^{p_1}C_2^{p_2}
\geq
\Lambda_X^{p_1+p_2}
=\Lambda_X
\),
and therefore
\[
T_{\mathcal J}
\leq
\frac{\mathfrak C_{\mu,\nu}}
     {(1-\theta_J)\Lambda_X}
=
\frac{\mathfrak C_{\mu,\nu}}
{(1-\theta_J)
\frac{\mathfrak C_{\mu,\nu}}
     {(1-\theta_J)T_X}}
=
T_X.
\]
Thus,
\(
t_0+T_{\mathcal J}
\leq
T_E+T_W+T+T_X
=T_p
\),
so
\(
\mathcal J(t)\leq\bar{\mathcal J}
\)
for every \(t\geq T_p\).

It remains to verify invariance. At
\(\mathcal J=\bar{\mathcal J}\),
\eqref{eq:J_comparison} and
\eqref{eq:Jbar_root} give
\(
\dot{\mathcal J}
\leq
-C_1\bar{\mathcal J}^{\mu}
-C_2\bar{\mathcal J}^{\nu}
+\Delta_J
=
-\Delta_J/\theta_J+\Delta_J
=
-[(1-\theta_J)/\theta_J]\Delta_J
\leq0
\).
Hence, the vector field cannot point outward through the boundary
of \(\mathcal B_J\), and \(\mathcal B_J\) is forward invariant.
Consequently,
\[
\mathcal J(t)\leq\bar{\mathcal J},
\qquad t\geq T_p .
\]

We next recover the pointwise state bound. Since
\(\mathcal J=\mathcal I_x+V_c\) with
\(\mathcal I_x,V_c\geq0\),
\(
\mathcal I_x(t)\leq\mathcal J(t)
\leq\bar{\mathcal J}
\)
for every \(t\geq T_p\). Recall that
\(
\mathcal I_x(t)
=
\int_{t-T}^{t}V^*(\bx(\tau))\,d\tau
\).
If
\(V^*(\bx(t))\leq2L_{\dot V}T\), then directly
\(
V^*(\bx(t))
\leq
\max\{
2L_{\dot V}T,
2\bar{\mathcal J}/T
\}
\).
Consider instead
\(V^*(\bx(t))>2L_{\dot V}T\). By
\eqref{eq:Vdot_uniform_bound}, for every
\(\tau\in[t-T,t]\),
\(
|V^*(\bx(t))-V^*(\bx(\tau))|
\leq
L_{\dot V}|t-\tau|
\leq
L_{\dot V}T
\),
and hence
\(
V^*(\bx(\tau))
\geq
V^*(\bx(t))-L_{\dot V}T
>
\frac12V^*(\bx(t))
\).
Therefore,
\(
\mathcal I_x(t)
=
\int_{t-T}^{t}V^*(\bx(\tau))\,d\tau
>
\frac{T}{2}V^*(\bx(t))
\).
Since
\(\mathcal I_x(t)\leq\bar{\mathcal J}\),
\(
V^*(\bx(t))
<
2\bar{\mathcal J}/T
\).
Combining the two alternatives gives
\[
V^*(\bx(t))
\leq
\max
\left\{
2L_{\dot V}T,\,
\frac{2\bar{\mathcal J}}{T}
\right\},
\qquad t\geq T_p .
\]
Using
\(V^*(\bx)\geq\underline v\norm{\bx}^2\),
\(
\underline v\norm{\bx(t)}^2
\leq
\max\{
2L_{\dot V}T,
2\bar{\mathcal J}/T
\}
\),
and hence
\(
\norm{\bx(t)}
\leq
\sqrt{
\underline v^{-1}
\max\{
2L_{\dot V}T,
2\bar{\mathcal J}/T
\}
}
=r_x
\),
which proves \eqref{eq:pointwise_x}.

For the critic error,
\(V_c(t)\leq\mathcal J(t)\leq\bar{\mathcal J}\).
From \eqref{eq:critic_Lyapunov},
\(
V_c(t)
=
\norm{\widetilde{\bW}(t)}^2/(2\alpha_c)
\),
so
\(
\norm{\widetilde{\bW}(t)}^2
=
2\alpha_cV_c(t)
\leq
2\alpha_c\bar{\mathcal J}
\),
and consequently
\(
\norm{\widetilde{\bW}(t)}
\leq
\sqrt{2\alpha_c\bar{\mathcal J}}
\).
On the other hand,
\eqref{eq:critic_bound_for_closed_loop} gives the independent
post-learning estimate
\(
\norm{\widetilde{\bW}(t)}\leq\bar W_c
\).
Therefore,
\[
\norm{\widetilde{\bW}(t)}
\leq
\min
\left\{
\bar W_c,\,
\sqrt{2\alpha_c\bar{\mathcal J}}
\right\}
=r_W,
\qquad t\geq T_p ,
\]
which is \eqref{eq:pointwise_W}. Hence, for
\(
\bm\Gamma(t):=
\operatorname{col}\{\bx(t),\widetilde{\bW}(t)\}
\),
one has
\(
\bm\Gamma(t)\in\mathcal B_r
\)
for every \(t\geq T_p\), where
\(
\mathcal B_r=
\{
\operatorname{col}\{\bx,\widetilde{\bW}\}:
\norm{\bx}\leq r_x,\,
\norm{\widetilde{\bW}}\leq r_W
\}
\).
This proves the claimed practical predefined-time stability.

Suppose now that
\eqref{eq:closed_loop_exact_HJI_condition} and
\eqref{eq:closed_loop_exact_mismatch_condition} hold. By
Theorem~\ref{thm:critic},
\(
\widetilde{\bW}(t)=\bm0
\)
for every
\(t\geq T_E+T_W\). Hence the learned value-gradient coincides with
its ideal value-gradient after the critic deadline, and therefore
\(
\hat{\bu}(t)=\bu^*(t)
\),
\(
\ba(t)=\ba^*(t)
\),
and
\(
\bd(t)=\bd^*(t)
\)
under \eqref{eq:closed_loop_exact_mismatch_condition}. Thus all
critic-dependent and critic-independent contributions entering
\(\Delta_W\) and \(\Delta_I\) vanish, so
\(
\Delta_W=0
\),
\(
\Delta_I=0
\),
and
\(
\Delta_J=\Delta_I+\Delta_W=0
\).
In particular, \eqref{eq:Jbar_root} becomes
\(
C_1\bar{\mathcal J}^{\mu}
+C_2\bar{\mathcal J}^{\nu}=0
\),
which, since \(C_1,C_2>0\), implies
\(\bar{\mathcal J}=0\).

More directly, for every
\(t\geq t_0=T_E+T_W+T\), the complete interval
\([t-T,t]\) lies after the exact critic-learning deadline and the
state moving-window functional satisfies
\(
\dot{\mathcal I}_x
\leq
-a_x\mathcal I_x^\mu
-b_x\mathcal I_x^\nu
\).
Lemma~\ref{lem:Gamma_comparison} therefore gives the exact settling
estimate
\(
T_x
\leq
\mathfrak C_{\mu,\nu}/
[
a_x^{p_1}b_x^{p_2}
]
\).
From \eqref{eq:lambda_x},
\(a_x\geq\Lambda_X\) and
\(b_x\geq2^{\nu-1}\Lambda_X\geq\Lambda_X\), so
\(
a_x^{p_1}b_x^{p_2}
\geq\Lambda_X^{p_1+p_2}
=\Lambda_X
\).
Consequently,
\[
T_x
\leq
\frac{\mathfrak C_{\mu,\nu}}{\Lambda_X}
=
(1-\theta_J)T_X
<
T_X .
\]
Hence
\(
\mathcal I_x(t)=0
\)

no later than
\(T_E+T_W+T+T_X=T_p\), and thus
\(
\mathcal I_x(t)=0
\)
for every \(t\geq T_p\). Since
\(
\mathcal I_x(t)
=
\int_{t-T}^{t}V^*(\bx(\tau))\,d\tau
\),
\(V^*(\bx(\tau))\geq0\), and
\(V^*(\bx(\tau))\) is continuous,
\(
\mathcal I_x(t)=0
\)
implies that the trajectory lies in the zero-level set of \(V^*\)
throughout \([t-T,t]\). By the positive definiteness of \(V^*\), this
zero-level set is \(\{\bm0\}\); hence, in particular,
\(
\bx(t)\in B_{r_-}(0)
\)
for every \(t\geq T_p\). Together with
\(
\widetilde{\bW}(t)=\bm0
\)
for \(t\geq T_E+T_W\), this proves
\eqref{eq:exact_closed_loop}.

Finally, the learned saturation-aware policy is defined
componentwise through the factor \(\tanh(\cdot)\). Since
\(
|\tanh(s)|<1
\)
for every finite \(s\in\mathbb R\), each learned input component
satisfies
\(
|\hat u_i(t)|
=
\bar u_i|\tanh(\cdot)|
<
\bar u_i
\)
for
\(i=1,\ldots,m\) and every \(t\geq0\).
Therefore, \eqref{eq:learned_saturation_theorem} holds independently
of the convergence argument above.
\end{proof}

\begin{remark}[Deadline accounting and convergence interpretation]
\label{rem:deadline_convergence_interpretation}
The decomposition \(T_p=T_E+T_W+T+T_X\) is an intrinsic but conservative certificate of integral learning: \(T_E\) builds an informative replay stack, \(T_W\) is the assigned critic-adaptation horizon, \(T\) flushes pre-learning data from the moving window, and \(T_X\) is the assigned convergence horizon of \(\mathcal J\). It does not prescribe sequential evolution, since \(x(t)\) and \(\widetilde W(t)\) evolve simultaneously and the state may substantially decrease, or even enter its ultimate neighborhood, during critic adaptation. The proof intentionally ignores this transient improvement: with \(t_0:=T_E+T_W+T\), it invokes the certified post-learning critic regime and guarantees \(T_J\leq T_X\), so that \(T_{\Gamma}\leq T_{\mathrm{cert}}:=t_0+T_J\leq T_p\), where \(T_{\Gamma}:=\inf\{t\geq0:\Gamma(\tau)\in\mathcal B_{\Gamma},\ \forall\,\tau\geq t\}\) is the actual permanent-entry time. Hence, even if simultaneous state--critic convergence yields \(T_{\Gamma}\leq T_E+T_W\), the certificate remains valid. The radius of \(\mathcal B_\Gamma\) is explicitly determined by the HJI residual, value-gradient approximation error, actual attack/disturbance reconstruction mismatch, and post-learning critic radius \(\bar W_c\). If these perturbations persist, the augmented closed loop reaches \(\mathcal B_\Gamma\) no later than \(T_p\); under \eqref{eq:closed_loop_exact_HJI_condition} and \eqref{eq:closed_loop_exact_mismatch_condition}, the residual set reduces to its exact limiting case and exact predefined-time convergence follows. Thus, the additive deadline is a conservative worst-case time-budget certificate, not a prescribed temporal ordering of state and critic convergence.
\end{remark}

\endgroup

\begin{remark}
\label{rem:inner_region_interpretation}
The origin remains the nominal regulation target and is not excluded from the
HJI domain $\Omega_i$. The restricted set
$\Omega_i^r:=\{x_i\in\Omega_i:\|x_i\|\ge r_{i,-}\}$ is introduced only
for the predefined-time comparison analysis. In particular, whenever
$0<\|x_i\|<r_{i,-}$, the local coordination error has already entered the
prescribed terminal neighborhood $\mathcal B_{r_{i,-}}(0)$, and no further
outer-region comparison is required. Hence, the exclusion of the origin from
$\Omega_i^r$ is purely analytical and does not modify the underlying HJI
problem or the nominal coordination objective. If disturbances or attacks
subsequently drive $x_i$ outside $\mathcal B_{r_{i,-}}(0)$, the same
predefined-time comparison becomes applicable again.
\end{remark}

\section{Numerical Illustration}
\label{sec:simulation}

\label{subsec:simulation_setup}

\subsection{Simulation Setup}
\label{subsec:simulation_setup}

A two-link planar manipulator with
\(\bx=\col\{q_1,q_2,\dot q_1,\dot q_2\}\) is considered, whose dynamics are
\(\mathbf M(\bm q)\ddot{\bm q}
+\mathbf C(\bm q,\dot{\bm q})\dot{\bm q}
+\mathbf D\dot{\bm q}
+\mathbf G(\bm q)
=\bu+\ba+\boldsymbol{\tau}_d\).
The nonzero model entries are
\(M_{11}=2.70+0.70\cos q_2\),
\(M_{12}=M_{21}=0.80+0.35\cos q_2\),
\(M_{22}=0.80\),
\(C_{11}=-0.35\dot q_2\sin q_2\),
\(C_{12}=-0.35(\dot q_1+\dot q_2)\sin q_2\),
\(C_{21}=0.35\dot q_1\sin q_2\),
\(\mathbf D=\diag(0.12,0.08)\),
\(G_1=8.5\sin q_1+2.6\sin(q_1+q_2)\), and
\(G_2=2.6\sin(q_1+q_2)\).
The game parameters are
\(\mathbf Q_x=\diag(7,6,1.5,1.2)\),
\(\mathbf R=0.06\mathbf I_2\),
\(\mathbf T=0.8\mathbf I_2\),
\(\mathbf S=\mathbf I_4\),
\((\gamma_a,\gamma_d)=(3.5,4.5)\), and
\(\bar u=8\).
The predefined-time parameters are
\((\gamma_1,\gamma_2)=(0.5,2)\),
\((\mu,\nu)=(0.75,1.5)\),
\((\kappa_1,\kappa_2)=(0.60,0.04)\),
\(\theta_W=\theta_J=0.10\),
\(\bar V=1\),
\(L_{\dot V}=10\), and
\(r_\omega=2\).

The critic uses \(L=20\) Gaussian radial basis functions
\(\phi_j(x)=\exp(-\|x-\mu_j\|^2/(2\sigma_j^2))\),
\(j=1,\ldots,L\), and a replay stack of \(k=200\) samples.
The finite-data informativity requirement is numerically assessed using
the threshold
\(\sigma_{\mathrm{req}}=3\times10^{-2}\),
together with the conditioning requirement
\(\operatorname{cond}(\Omega)<5\times10^3\).
The data-acquisition interval is
\(T_E=0.30~\mathrm{s}\), the integral window is
\(T=0.05~\mathrm{s}\), and the deadline-allocation factor is
\(\chi=0.30\).
Two prescribed deadlines are considered,
\(T_p=2.2~\mathrm{s}\) and \(T_p=10~\mathrm{s}\), with
\(T_W=\chi(T_p-T_E-T)\) and
\(T_X=(1-\chi)(T_p-T_E-T)\).
For \(T_p=2.2~\mathrm{s}\),
\((T_W,T_X)=(0.555,1.295)~\mathrm{s}\),
\(T_E+T_W+T=0.905~\mathrm{s}\),
\(\lambda_x=4.137109\times10^{1}\), and
\(\alpha_c=1.064897\times10^{4}\);
for \(T_p=10~\mathrm{s}\),
\((T_W,T_X)=(2.895,6.755)~\mathrm{s}\),
\(T_E+T_W+T=3.245~\mathrm{s}\),
\(\lambda_x=7.931245\), and
\(\alpha_c=2.876914\times10^{3}\).
The terminal radius is \(r_x=0.04\), the actuator bound remains
\(|u_j|<8\), and all responses are displayed over \(12~\mathrm{s}\).

The matched FDI attack is
\(\ba(t)=e^{-0.26t}\col\{\psi_{a1}(t),\psi_{a2}(t)\}\), where
\(\psi_{a1}=1.45\sin(5.6t+0.20)+0.95\sin(11.8t-0.10)
+0.48\cos(17.4t+0.35)\) and
\(\psi_{a2}=-1.30\cos(6.2t-0.25)+0.82\sin(12.5t+0.40)
-0.42\cos(18.6t-0.15)\).
The external torque disturbance is
\(\boldsymbol{\tau}_d(t)=
e^{-0.31t}\col\{\psi_{d1}(t),\psi_{d2}(t)\}\), with
\(\psi_{d1}=0.40\sin(6.8t+0.15)+0.22\cos(13.2t+0.25)
+0.11\sin(19.5t-0.30)\) and
\(\psi_{d2}=-0.35\cos(6.1t-0.50)+0.20\sin(12.6t+0.05)
-0.10\cos(18.8t+0.38)\).

The nominal initial state is
\(\bx(0)=\col\{1.470,-1.155,0.420,-0.315\}\).
The multi-initial-condition test uses five levels with
\(\|x_0\|\approx0.55,1.02,1.48,1.94,2.40\).
For comparison, the proposed Predefined-time IRL method is evaluated
against Fixed-time IRL and conventional ADP-IRL, implemented according
to the corresponding formulations in~\cite{Vu2026FixedTimeIRL}.
For a fair comparison, both baselines use the same replay data and the
same learning interval \(T_W\) as the corresponding Predefined-time IRL
case.

Since the initial critic weights determine the initial approximate value
gradient and therefore the induced control policy, their choice is not
merely numerical but directly affects whether the initial closed loop
satisfies the admissibility and stability requirements of
Definition~\ref{ass:bounded_exogenous} under
Assumption~\ref{ass:plant_regularity}.
Hence, the initialization should be selected within a region that
preserves an admissible stabilizing policy before learning improves the
critic; related constructions are discussed in~\cite{Vu2026FixedTimeIRL}.

\subsection{Simulation Results}
\label{sec:simulation_results}


The finite-data condition is satisfied before the end of the acquisition
phase. In particular, the required level
\(\sigma_{\mathrm{req}}=3\times10^{-2}\), where
\(\sigma_{\mathrm{req}}\) is the prescribed finite-data threshold,
is reached at approximately
\(0.13~\mathrm{s}<T_E=0.30~\mathrm{s}\), while the smallest-singular-value
indicator subsequently increases to approximately
\(1.24\times10^{-1}\). Hence, the replay data become informative with
a clear margin before critic learning is completed.

The critic-weight responses are shown in Fig.~\ref{fig:critic_weights}. Most weight variation occurs within the first \(5\times10^{-2}~\mathrm{s}\), after which both groups become practically stationary well before \(T_W\). This rapid transient results from the informative replay stack available before learning, the simultaneous use of all stored directions, regressor normalization, and the two-power correction, which ensures strong descent both far from and near the residual set. The deadline-dependent critic gain further accelerates convergence, particularly for the shorter prescribed time. Hence, the fast critic transient is mainly due to informative finite replay data, normalization, and two-power learning. Since both deadline tests use the same replay information and critic problem, their resulting weight solutions remain essentially identical.


\begin{figure}[H]
    \centering
    \makebox[\columnwidth][c]{%
    \subfloat[Quadratic weights]{
        \includegraphics[
            width=0.465\columnwidth,
            height=0.4\columnwidth
        ]{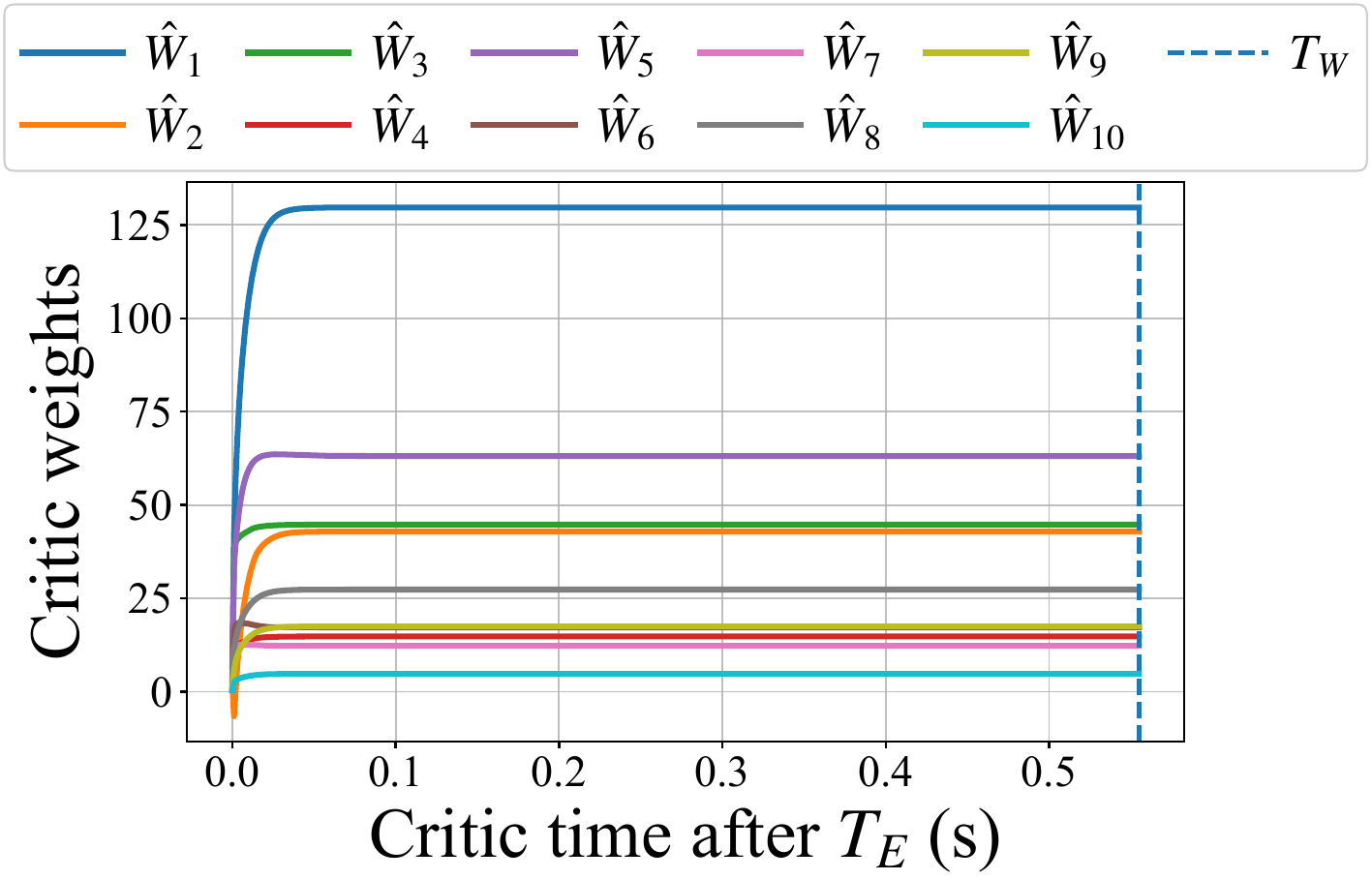}
        \label{fig:critic_weights_1}
    }%
    \hspace{0.02\columnwidth}%
    \subfloat[Higher-order weights]{
        \includegraphics[
            width=0.465\columnwidth,
            height=0.4\columnwidth
        ]{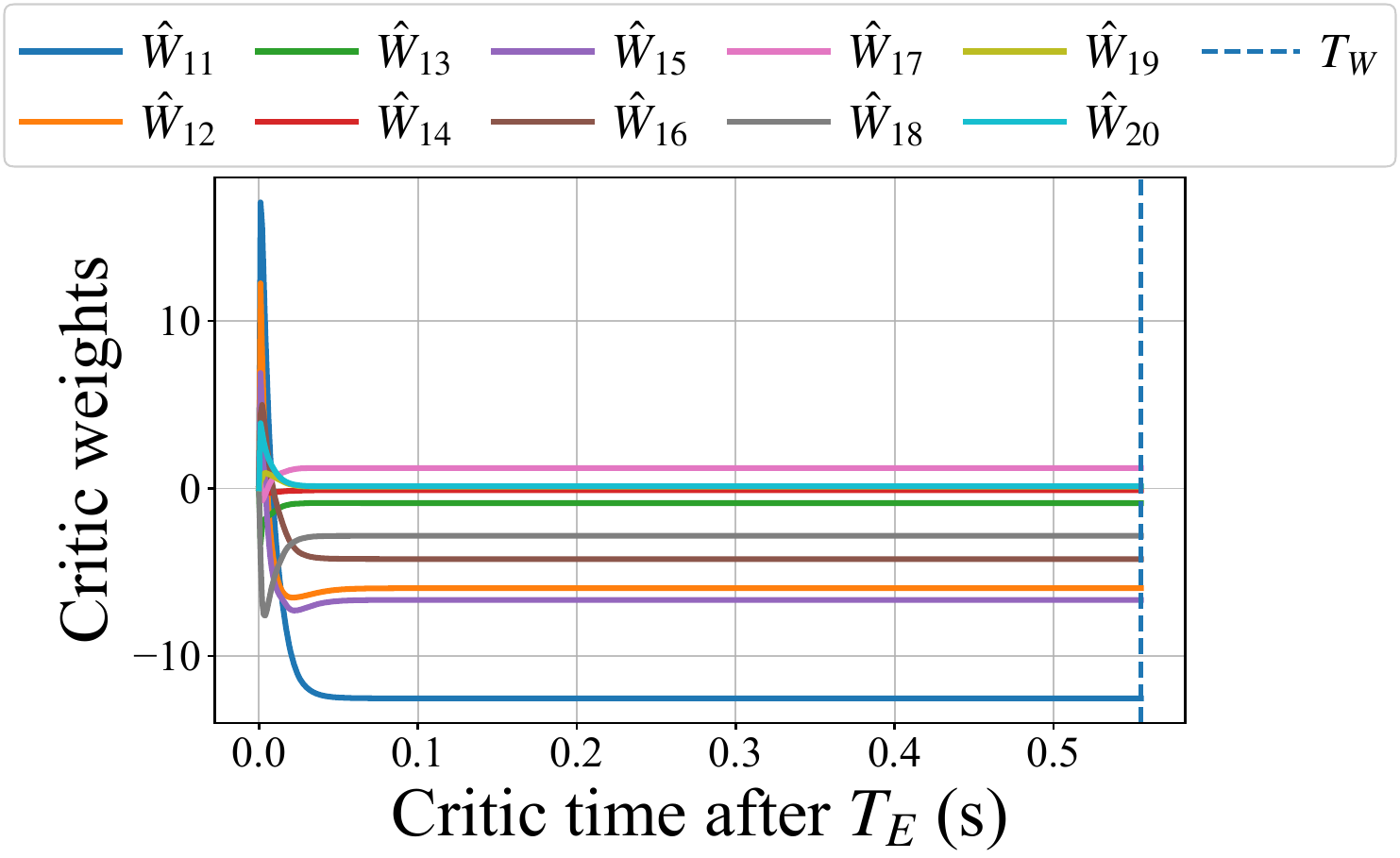}
        \label{fig:critic_weights_2}
    }}
    \caption{Critic-weight evolution during the finite learning interval.}
    \label{fig:critic_weights}
\end{figure}

The influence of the assigned deadline is shown in
Fig.~\ref{fig:state_deadline_comparison}. With
\(T_p=10~\mathrm{s}\), the available time budget is large and the
states approach the origin several seconds before the assigned
deadline. When the deadline is reduced to \(T_p=2.2~\mathrm{s}\), the
transient is compressed considerably and the states are driven toward
the origin much faster. The regulation objective is unchanged; what
changes is the amount of time available to achieve it.

\begin{figure}[H]
    \centering
    \makebox[\columnwidth][c]{%
    \subfloat[\(T_p=2.2~\mathrm{s}\)]{
        \includegraphics[width=0.465\columnwidth]
        {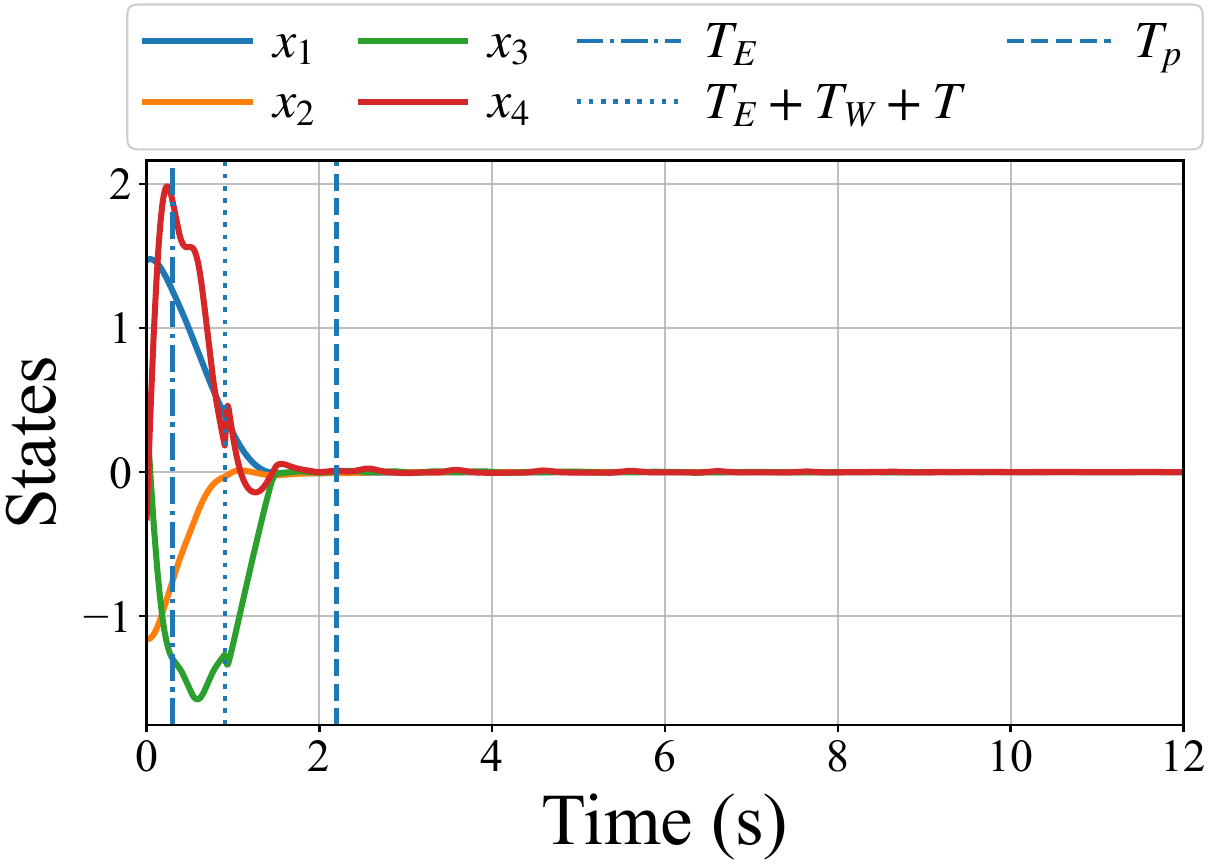}
        \label{fig:states_T22}
    }%
    \hspace{0.02\columnwidth}%
    \subfloat[\(T_p=10~\mathrm{s}\)]{
        \includegraphics[width=0.465\columnwidth]
        {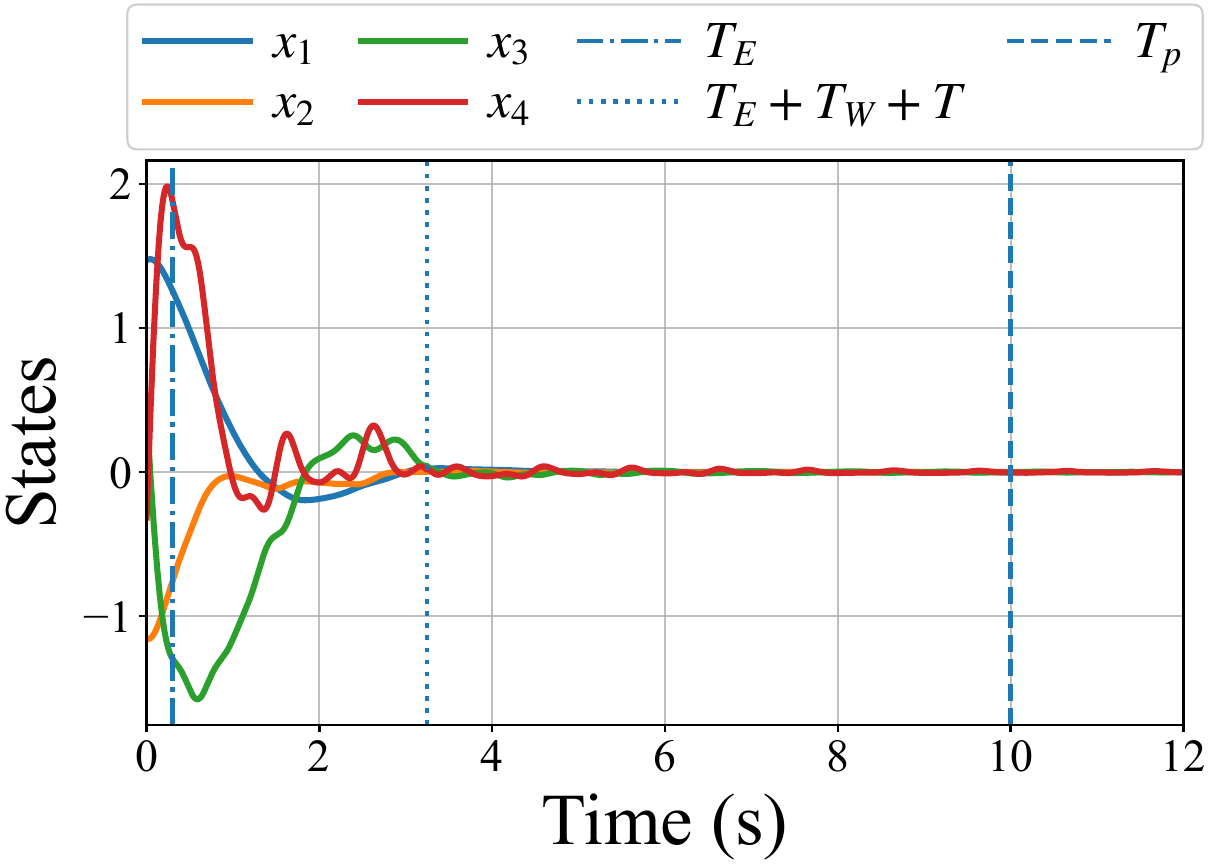}
        \label{fig:states_T10}
    }}
    \caption{Closed-loop state responses under two assigned deadlines.}
    \label{fig:state_deadline_comparison}
\end{figure}

The corresponding control signals in
Fig.~\ref{fig:control_deadline_comparison} show the direct cost of
reducing \(T_p\). For \(T_p=10~\mathrm{s}\), a large actuator margin
remains during most of the transient. For \(T_p=2.2~\mathrm{s}\), one
control channel reaches approximately the hard limit
\(\bar u=8\) over a finite interval. Therefore, the shorter deadline is
obtained by using more available control authority. This mechanism is
not unlimited because the saturation-aware input penalty enforces
\(|u_j|<\bar u\).


\begin{figure}[H]
    \centering
    \makebox[\columnwidth][c]{%
    \subfloat[\(T_p=2.2~\mathrm{s}\)]{
        \includegraphics[
            width=0.465\columnwidth,
            height=0.5\columnwidth
        ]{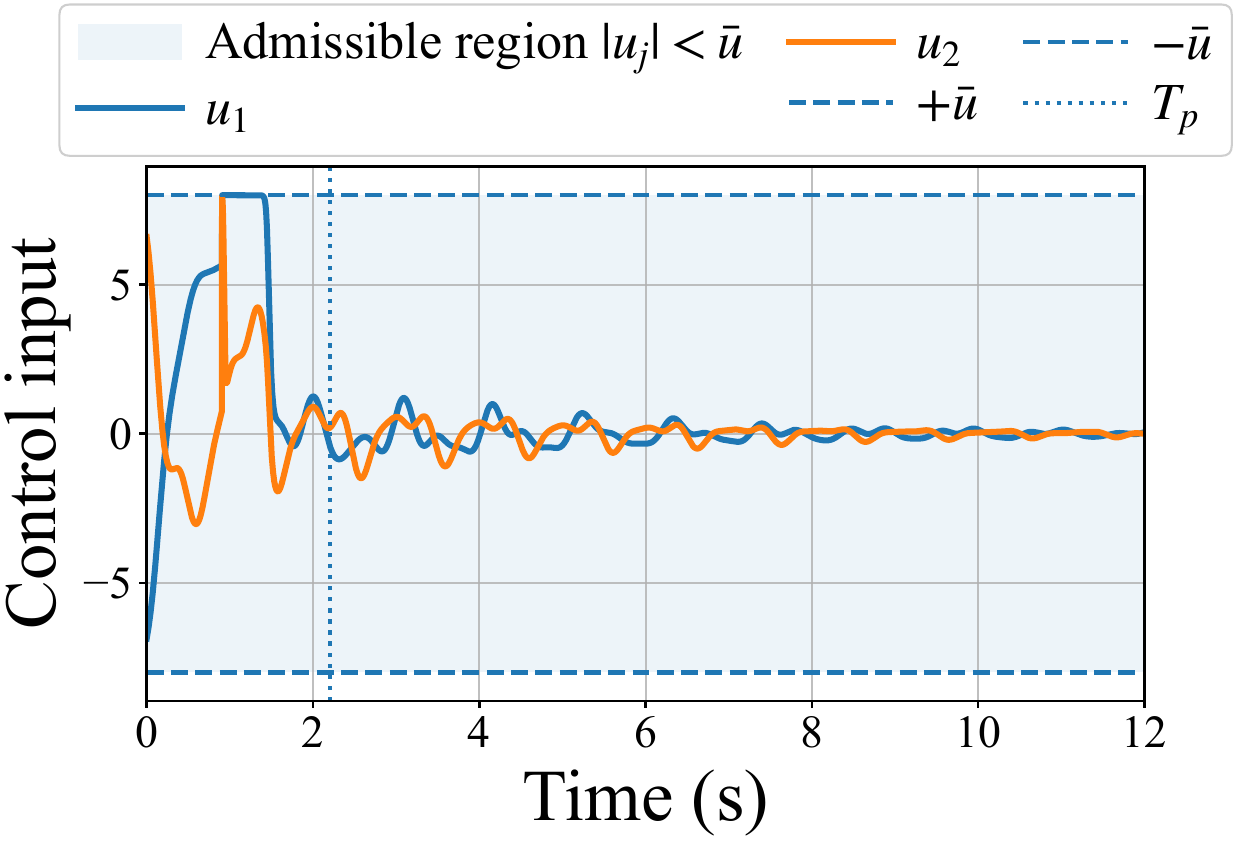}
        \label{fig:control_T22}
    }%
    \hspace{0.02\columnwidth}%
    \subfloat[\(T_p=10~\mathrm{s}\)]{
        \includegraphics[
            width=0.465\columnwidth,
            height=0.5\columnwidth
        ]{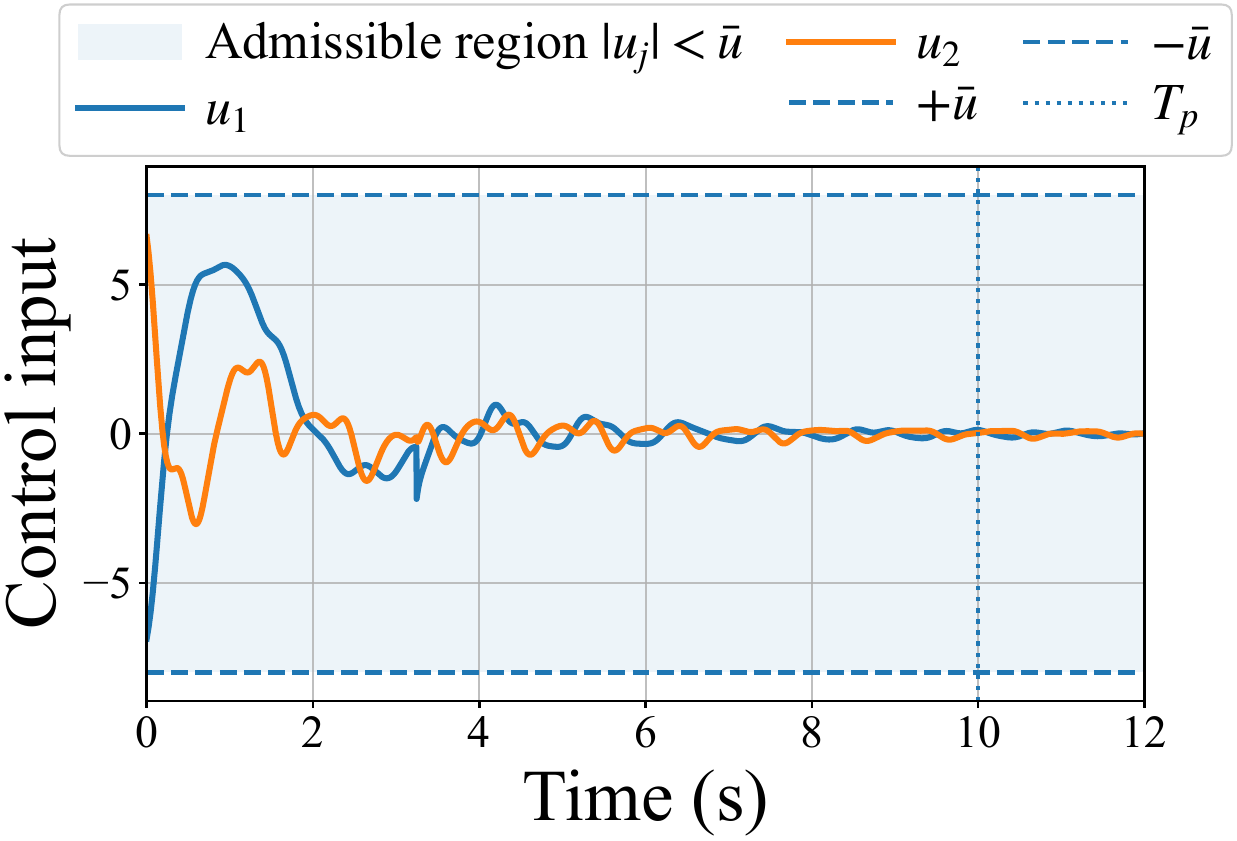}
        \label{fig:control_T10}
    }}
    \caption{Control signals under the same actuator bound
    \(\bar u=8\).}
    \label{fig:control_deadline_comparison}
\end{figure}

The adverse signals used in both tests are shown in
Fig.~\ref{fig:attack_disturbance}. Both the matched FDI attack and the
external torque disturbance remain nonzero around the short deadline.
Hence, the fast response obtained for \(T_p=2.2~\mathrm{s}\) is not
produced by removing the adverse signals before the prescribed time.


\begin{figure}[H]
    \centering
    \makebox[\columnwidth][c]{%
    \subfloat[FDI attack]{
        \includegraphics[
            width=0.465\columnwidth,
            height=0.5\columnwidth
        ]{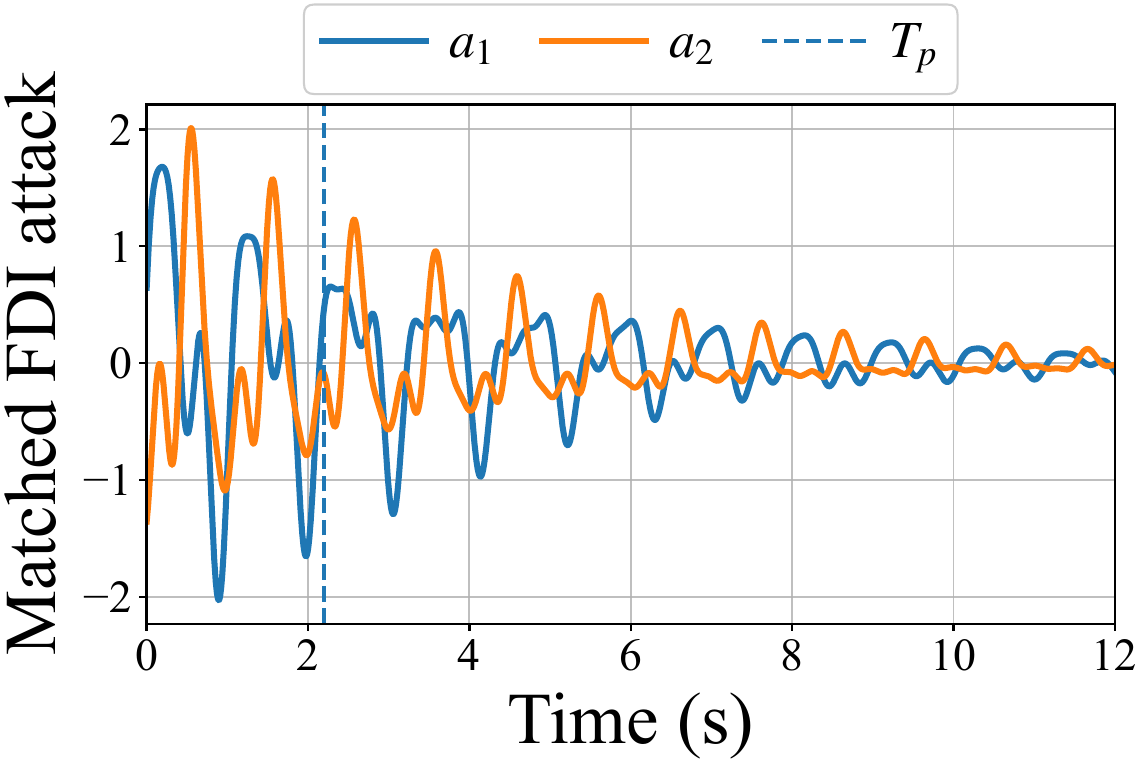}
        \label{fig:attack}
    }%
    \hspace{0.02\columnwidth}%
    \subfloat[Torque disturbance]{
        \includegraphics[
            width=0.465\columnwidth,
            height=0.5\columnwidth
        ]{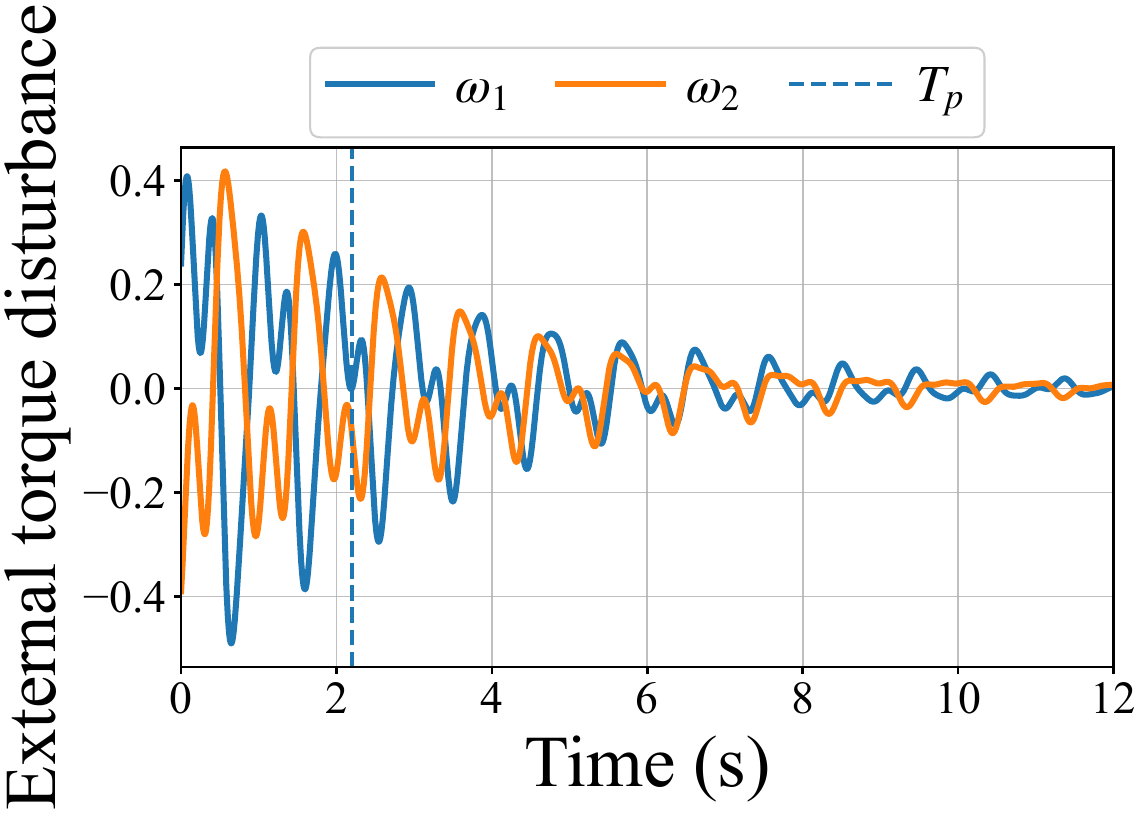}
        \label{fig:disturbance}
    }}
    \caption{Matched FDI attacks and external torque disturbances.}
    \label{fig:attack_disturbance}
\end{figure}

The influence of the initial displacement is shown in
Fig.~\ref{fig:multi_ic_comparison}. For \(T_p=10~\mathrm{s}\), all
five initial conditions, with norms ranging from \(0.55\) to \(2.40\),
enter the practical region \(r_x=0.04\) with a large time margin. For
\(T_p=2.2~\mathrm{s}\), the closer initial conditions enter the
practical region within or near the prescribed time, whereas the
farthest condition, \(\|x_0\|=2.40\), remains slightly above
\(r_x\) at \(T_p\) and enters the region shortly afterwards.


\begin{figure}[H]
    \centering

    \subfloat[\(T_p=2.2~\mathrm{s}\)]{
        \includegraphics[width=0.92\columnwidth]
        {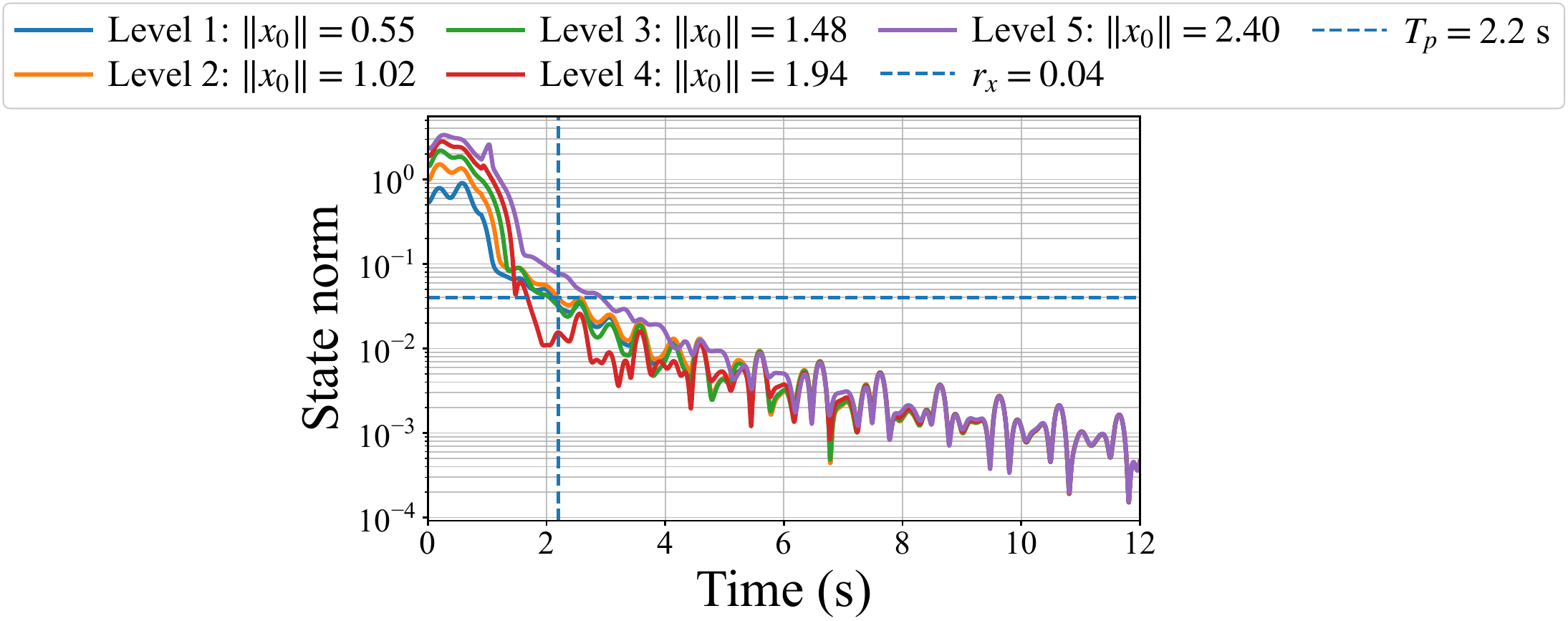}
        \label{fig:multi_T22}
    }

    \vspace{0.5em}

    \subfloat[\(T_p=10~\mathrm{s}\)]{
        \includegraphics[width=0.92\columnwidth]
        {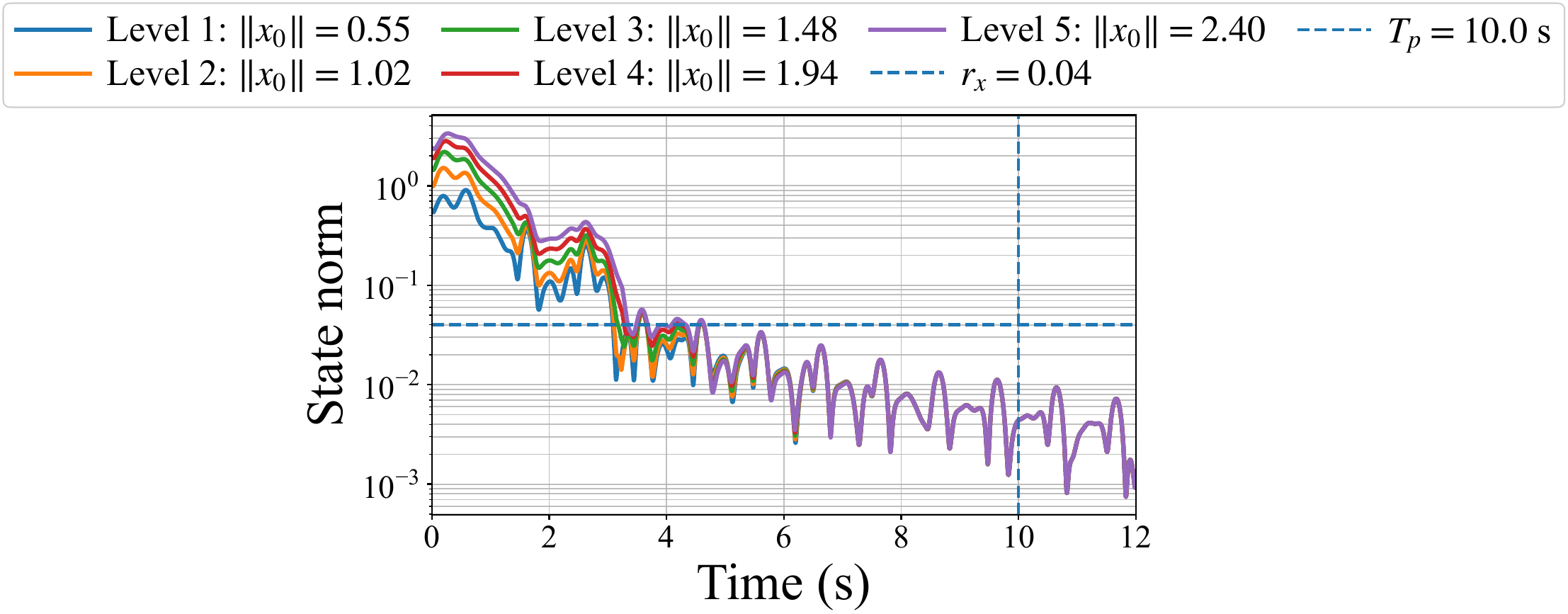}
        \label{fig:multi_T10}
    }

    \caption{State-norm responses for five initial-condition levels.}
    \label{fig:multi_ic_comparison}
\end{figure}

This result shows the feasibility tradeoff directly. A farther initial
condition requires a larger state displacement within the same finite
interval, while the available physical input remains bounded by
\(\bar u=8\). Consequently, reducing \(T_p\) does not create unlimited
convergence capability. Once the actuator is close to saturation,
further reduction of the prescribed time cannot be compensated by an
arbitrarily large control input.

A direct comparison among Predefined-time IRL, Fixed-time IRL, and
ADP\_IRL is given in Fig.~\ref{fig:controller_state_comparison}. The
three methods use the same initial condition, actuator bound,
attack/disturbance signals, replay information, and simulation horizon.
At \(T_p=2.2~\mathrm{s}\), the Predefined-time IRL state norm has
already decreased to the order of \(10^{-2}\) and lies below
\(r_x=0.04\), whereas the Fixed-time IRL and ADP\_IRL responses remain
approximately one order larger. At \(T_p=10~\mathrm{s}\), all three
methods have enough time to approach a small neighborhood of the
origin, so the difference between them becomes less significant.


\begin{figure}[H]
    \centering

    \subfloat[\(T_p=2.2~\mathrm{s}\)]{
        \includegraphics[width=0.6\columnwidth]
        {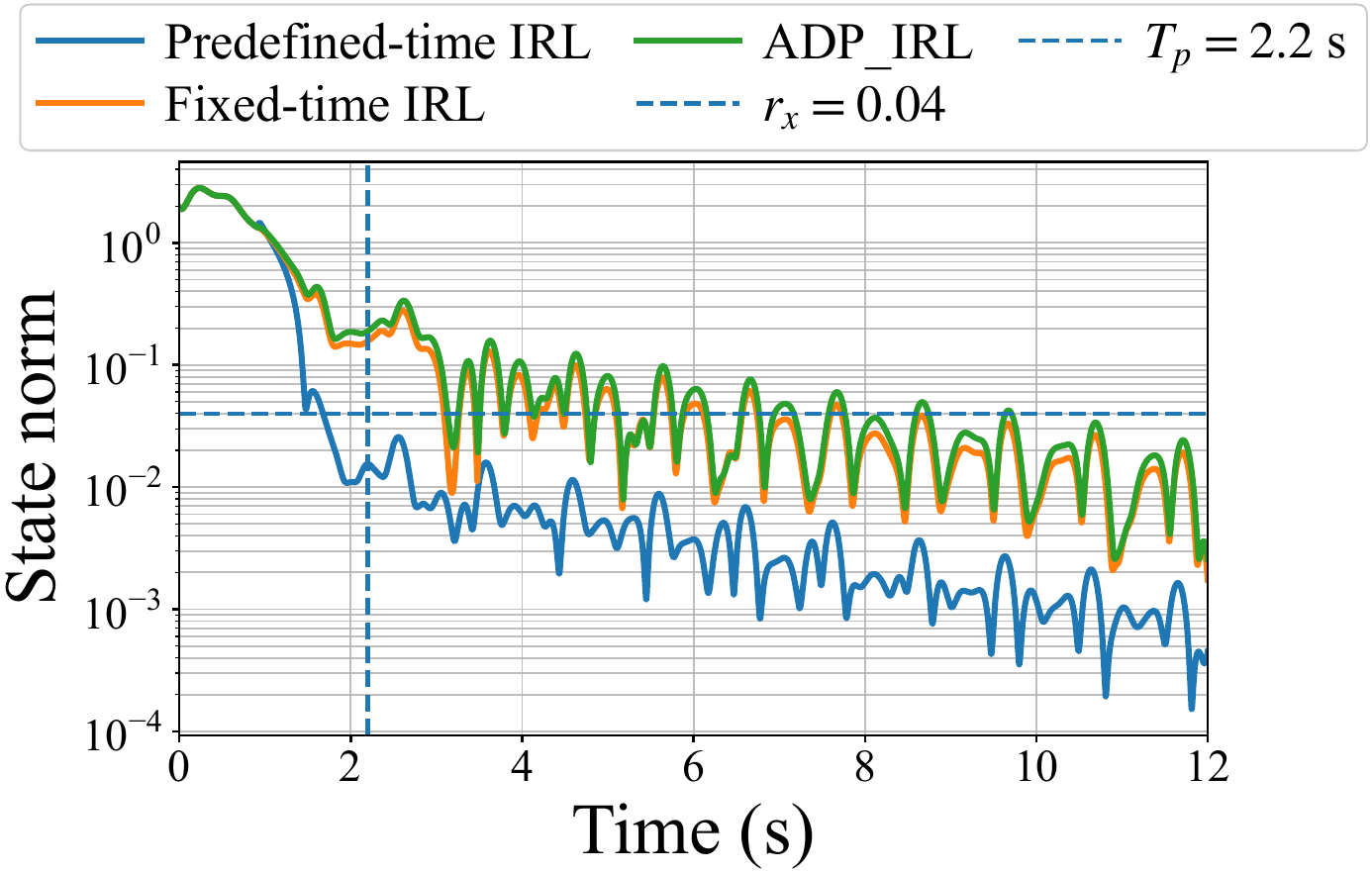}
        \label{fig:comparison_state_T22}
    }

    \vspace{0.5em}

    \subfloat[\(T_p=10~\mathrm{s}\)]{
        \includegraphics[width=0.6\columnwidth]
        {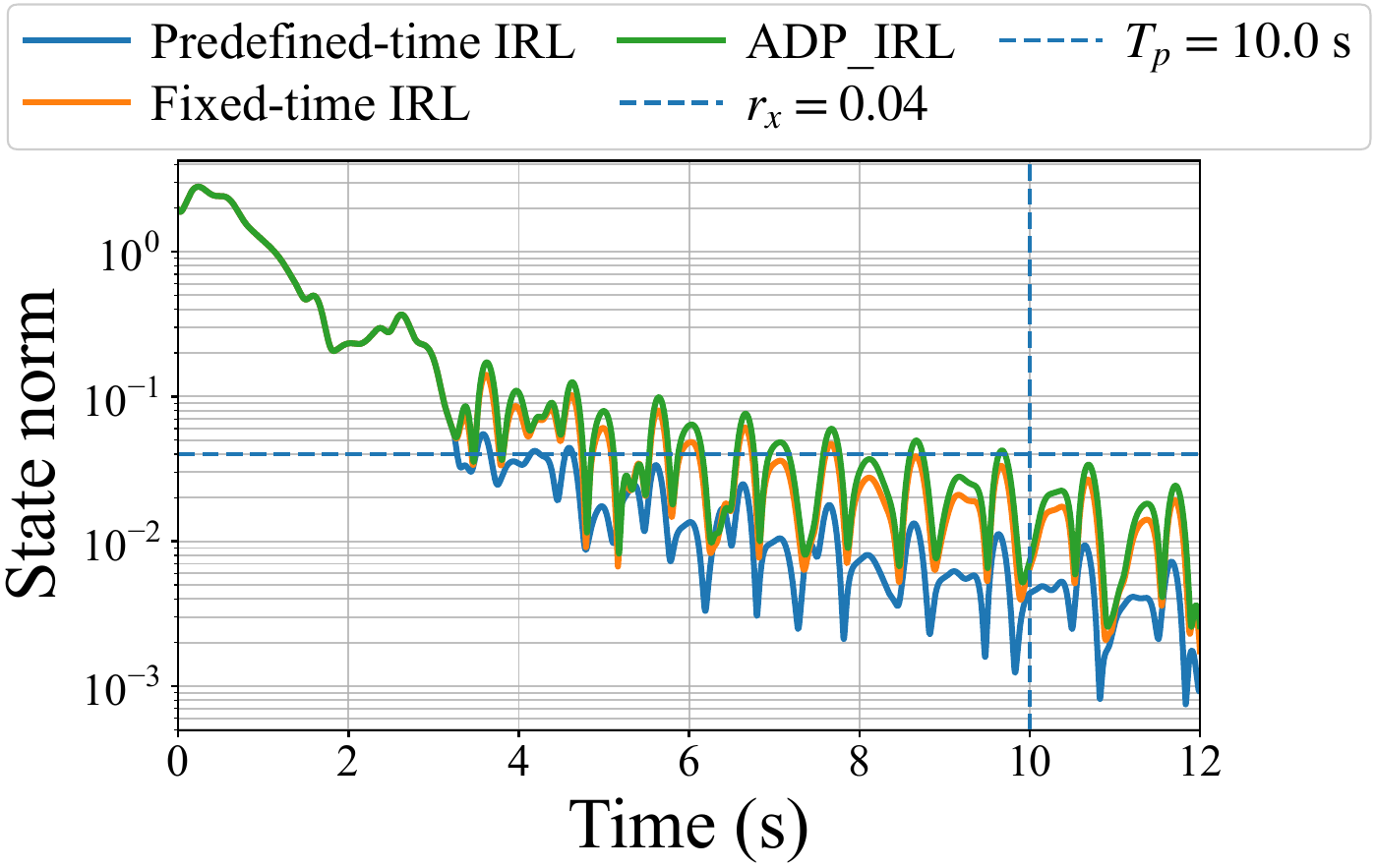}
        \label{fig:comparison_state_T10}
    }

    \caption{State-norm comparison among Predefined-time IRL,
    Fixed-time IRL, and ADP\_IRL.}
    \label{fig:controller_state_comparison}
\end{figure}

The control-effort comparison in
Fig.~\ref{fig:controller_control_comparison} explains the difference
in the state responses. For \(T_p=2.2~\mathrm{s}\),
Predefined-time IRL uses nearly the full admissible magnitude
\(\bar u=8\), whereas the peak values of Fixed-time IRL and
ADP\_IRL are approximately \(6\)--\(6.5\). For
\(T_p=10~\mathrm{s}\), such strong actuation is no longer necessary.
The shorter settling behavior is therefore obtained through a larger,
but still admissible, control effort.







\begin{figure}[H]
    \centering

    \subfloat[\(T_p=2.2~\mathrm{s}\)]{
        \includegraphics[width=0.7\columnwidth]
        {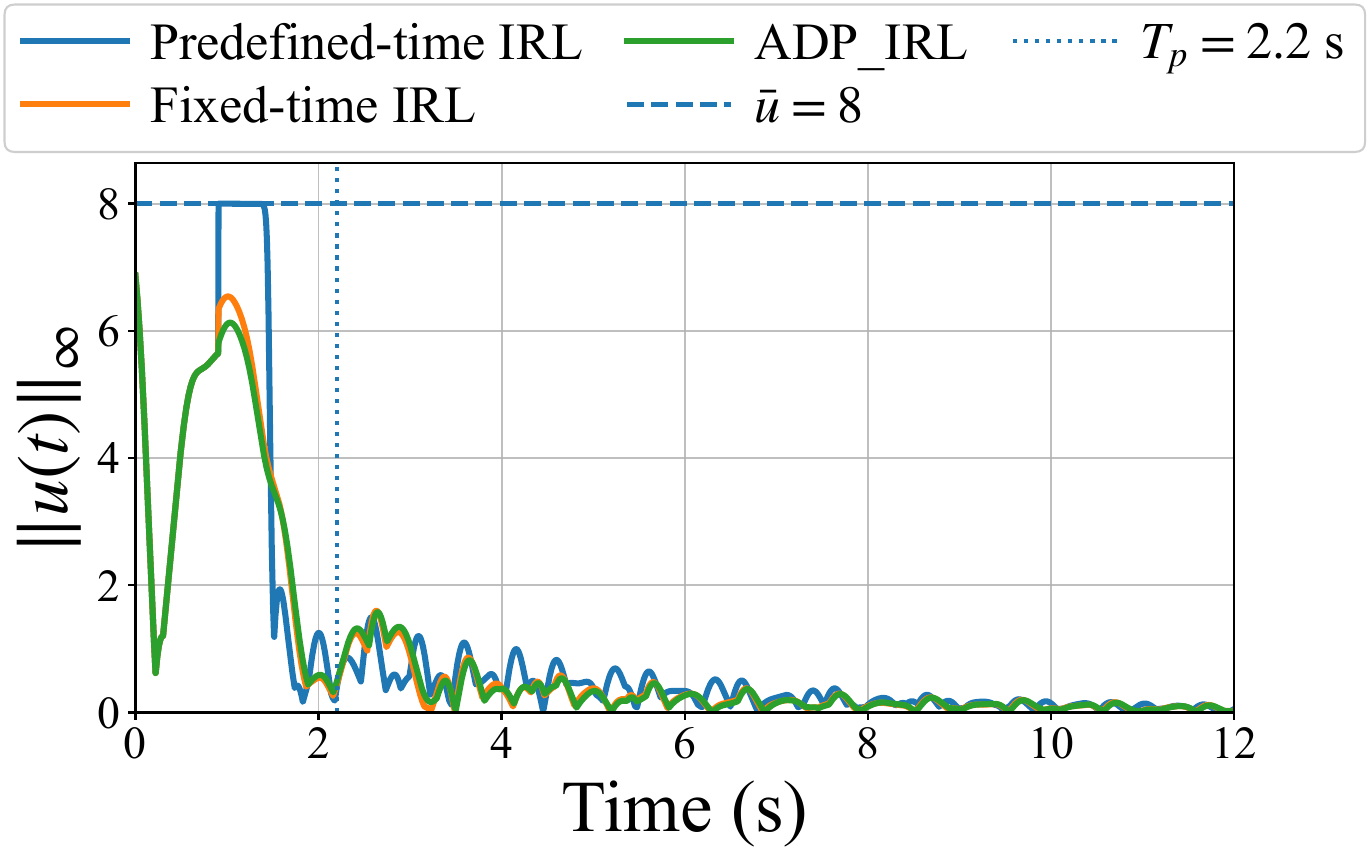}
        \label{fig:comparison_control_T22}
    }

    \vspace{-0.4em}

    \subfloat[\(T_p=10~\mathrm{s}\)]{
        \includegraphics[width=0.7\columnwidth]
        {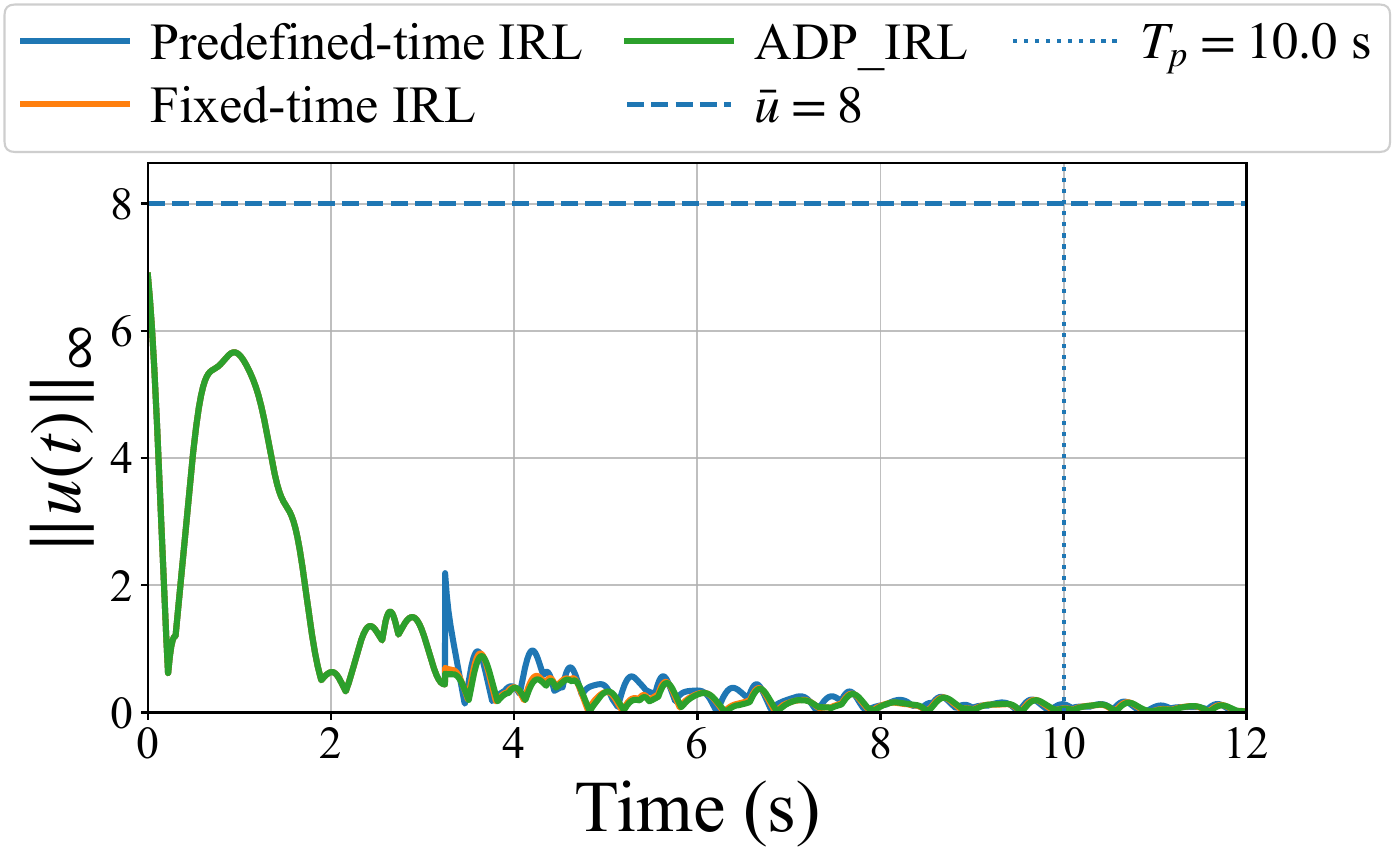}
        \label{fig:comparison_control_T10}
    }

    \vspace{-0.3em}
    \caption{Control-effort comparison among Predefined-time IRL,
    Fixed-time IRL, and ADP\_IRL.}
    \label{fig:controller_control_comparison}
    \vspace{-0.5em}
\end{figure}

Overall, the two deadline tests show a clear
settling-time/control-effort tradeoff. Reducing \(T_p\) from
\(10~\mathrm{s}\) to \(2.2~\mathrm{s}\) produces a much faster nominal
response, but at the same time forces the controller toward the
saturation boundary. The multiple-initial-condition test shows the same
effect from another direction: when the initial state is farther from
the origin, the same bounded actuator requires more time to bring the
trajectory into the practical terminal region. Therefore,
Predefined-time IRL does not mean that \(T_p\) can be reduced toward
zero independently of the plant and actuator. The assigned deadline
must remain compatible with the available control authority, the
considered initial-condition region, and the learning and approximation
conditions. The \(T_p=2.2~\mathrm{s}\) experiment therefore also shows
the practical boundary of feasible predefined-time operation.

\begin{remark}
The finite control authority and system limitations imply that the
prescribed convergence time cannot, in general, be reduced arbitrarily.
This observation suggests a promising research direction: establishing
an explicit lower bound on the admissible predefined time as a function
of actuator limits, system parameters, the initial-condition region,
and disturbance/attack bounds. Such a bound would characterize the
smallest physically feasible prescribed convergence time for a given
system.
\end{remark}

\section{Discussion}

\label{subsec:scalability}

The present framework suggests several extensions beyond predefined-time stabilization. Our recent work~\cite{VuEtAl2026FixedTimeMASIRL} extended fixed-time IRL to saturated nonlinear multi-agent systems under FDI attacks, showing that finite-window learning and resilient optimal control can be distributed, although its settling-time bound remains gain-induced rather than designer-assigned. This motivates distributed predefined-time IRL with a common prescribed deadline preserved under graph coupling, heterogeneity, partial measurements, delays, switching information structures, and network growth, together with explicit network-dependent constants.

A major challenge is networked and cyber--physical control under intermittent communication, DoS attacks, packet losses, finite-rate channels, quantization, delays, bandwidth limits, and decentralized communication decisions. Existing works separately address stochastic sensor/actuator dropouts, event-triggered discrete-time ADP under DoS, resilient model-free ADP, predefined-time self/event-triggered control under DoS, quantized predefined-time coordination, and bit-rate-constrained secure control~\cite{Jiang2024DropoutRL,Lu2024DoSADP,Zhong2025DoSADP,Yang2025PredefinedDoS,Hua2026PredefinedCPS,Shi2026QuantizedPredefined,Wang2024BitRate}, but a unified treatment remains largely open. A key direction is predefined-time IRL for distributed nonlinear discrete-time systems jointly addressing learning, control, communication scheduling, security, and channel capacity. The prescribed horizon should be a true wall-clock guarantee, rather than one accumulated only over successful-transmission or attack-free intervals, thereby exposing explicit tradeoffs among deadline, required communication/bit rate, packet loss and DoS severity, quantization, and asynchronous updates. Although fixed-time RL has recently been developed for uncertain discrete-time nonlinear systems~\cite{Ali2026FixedTimeDiscreteRL}, the distributed networked IRL counterpart with unknown dynamics and combined communication impairments remains largely unresolved, despite the natural relevance of discrete-time models to digital CPS implementation.

Beyond stabilization, similar issues arise in predefined-time optimal tracking and safety-critical learning. Reachability methods could jointly guarantee forward safety and an assigned convergence deadline, including multi-pursuer/multi-evader games with bounded rationality under uncertainty, while Gaussian-process and physics-informed learning may address model uncertainty and stabilize HJB/HJI approximation for autonomous navigation. Another fundamental issue is the need for an initially admissible stabilizing policy in many ADP/IRL schemes, which creates the circularity that stabilization is required before an optimal stabilizer can be learned; in predefined-time ADP, this initialization also consumes part of the prescribed horizon. A critic-weight initialization methodology developed in our related work~\cite{Vu2026FixedTimeIRL} provides a systematic way to select initial weights consistent with admissibility and closed-loop stability, and its integration with predefined-time learning is therefore promising. A complementary remedy is a Koopman warm start that reuses replay-stack data to construct a lifted predictor, constrained stabilizing policy, and initial value-function approximation before nonlinear IRL is activated. Recent connections between Koopman spectral theory and Hamilton--Jacobi equations~\cite{Vaidya2025KoopmanHJ}, together with certified data-driven Koopman control~\cite{Strasser2026SafEDMD}, support this direction. The key issue is not computational speed, but certification of regional admissibility, explicit propagation of lifting and finite-data errors, and inclusion of Koopman construction time within the same predefined-time budget.

Finally, the present theory assumes sufficient value-function smoothness. Replacing classical differentiability with viscosity sub/supersolution arguments may establish prescribed-time properties directly for generalized HJB/HJI solutions while accommodating nonsmooth state constraints, safety boundaries, switching communication modes, and hybrid dynamics within a broader predefined-time ADP framework.

\section{Conclusion}

This paper developed a predefined-time resilient critic-only IRL
framework for nonlinear systems with unknown drift dynamics, while
simultaneously accounting for input constraints, FDI attacks, and
external disturbances. By combining finite-window Bellman--Isaacs
identities, finite replay data, and a two-power critic update, both the
critic and the closed-loop state are driven into practical regions
within designer-assigned times, with the associated gains determined
explicitly from the prescribed deadline. The Lyapunov analysis accounts
directly for the data-acquisition time and the learning-window delay,
while the saturated policy preserves the actuator constraints by
construction. The simulations further show that the convergence time
can be actively shaped, but only at the cost of increased control effort
and within the finite physical capability of the actuators.

\section*{Acknowledgment}
The authors would like to thank Ho Chi Minh City University of Technology (HCMUT) and Vietnam National University Ho Chi Minh City (VNU-HCM) for supporting this research.

\appendices

\section{Proof of Theorem~\ref{thm:normalized_predefined_time}}
\label{App1}
\begin{proof}
Let \(x(\cdot;x_0)\) denote the forward-complete trajectory initialized at \(x_0\), and define \(T(x_0):=\inf\{t\geq0:V(x(t;x_0))=0\}\). Since \(V\) is positive definite, \(V(x_0)=0\) implies \(x_0=0\), in which case \(T(x_0)=0\). Hence, consider \(V_0:=V(x_0)>0\). Introduce \(\phi(s):=(\alpha s^p+\beta s^q)^r\) and \(\Psi(v):=\int_0^v\phi^{-1}(s)\,ds\), \(v\geq0\).

We first evaluate the improper integral that normalizes the prescribed settling-time bound. Factoring \(\alpha s^p\) from the denominator gives \(\int_0^\infty(\alpha s^p+\beta s^q)^{-r}ds=\alpha^{-r}\int_0^\infty s^{-pr}\left(1+\frac{\beta}{\alpha}s^{q-p}\right)^{-r}ds\). Set \(z:=\frac{\beta}{\alpha}s^{q-p}\), \(s=\left(\frac{\alpha}{\beta}\right)^{\frac{1}{q-p}}z^{\frac{1}{q-p}}\), so that \(ds=\frac{1}{q-p}\left(\frac{\alpha}{\beta}\right)^{\frac{1}{q-p}}z^{\frac{1}{q-p}-1}dz\) and \(s^{-pr}=\left(\frac{\alpha}{\beta}\right)^{-\frac{pr}{q-p}}z^{-\frac{pr}{q-p}}\). Consequently, \(\int_0^\infty(\alpha s^p+\beta s^q)^{-r}ds=\frac{1}{\alpha^r(q-p)}\left(\frac{\alpha}{\beta}\right)^{\frac{1-pr}{q-p}}\int_0^\infty z^{\frac{1-pr}{q-p}-1}(1+z)^{-r}dz\).

Define \(a:=\frac{1-pr}{q-p}\) and \(b:=\frac{qr-1}{q-p}\). The conditions \(pr<1<qr\) and \(q>p\) imply \(a>0\) and \(b>0\), while \(a+b=\frac{1-pr+qr-1}{q-p}=r\). Using the Euler beta-function identity \(\mathrm B(a,b):=\int_0^\infty z^{a-1}(1+z)^{-(a+b)}dz=\frac{\Gamma_{\!E}(a)\Gamma_{\!E}(b)}{\Gamma_{\!E}(a+b)}\), we therefore obtain \(\int_0^\infty(\alpha s^p+\beta s^q)^{-r}ds=\frac{\Gamma_{\!E}\!\left(\frac{1-pr}{q-p}\right)\Gamma_{\!E}\!\left(\frac{qr-1}{q-p}\right)}{\alpha^r\Gamma_{\!E}(r)(q-p)}\left(\frac{\alpha}{\beta}\right)^{\frac{1-pr}{q-p}}=\gamma_{p,q,r}\). This also shows explicitly why \(pr<1\) guarantees integrability at \(s=0\), whereas \(qr>1\) guarantees integrability as \(s\to\infty\).

For every \(t<T(x_0)\), one has \(V(x(t))>0\). Since \(V(x(t))\) is locally absolutely continuous and \(\Psi\) is continuously differentiable on \((0,\infty)\), the chain rule for absolutely continuous functions gives, for almost all \(t<T(x_0)\), \(\frac{d}{dt}\Psi(V(x(t)))=\frac{\dot V(x(t))}{[\alpha V^p(x(t))+\beta V^q(x(t))]^r}\leq-\frac{\gamma_{p,q,r}}{T_s}\). Integration over \([0,t]\subset[0,T(x_0))\) yields \(\Psi(V(x(t)))-\Psi(V_0)\leq-\frac{\gamma_{p,q,r}}{T_s}t\). Because \(\Psi(V(x(t)))\geq0\), it follows that \(t\leq\frac{T_s}{\gamma_{p,q,r}}\Psi(V_0)=\frac{T_s}{\gamma_{p,q,r}}\int_0^{V_0}\frac{ds}{(\alpha s^p+\beta s^q)^r}\).

If \(T(x_0)=\infty\), the preceding inequality would have to hold for arbitrarily large \(t\), which is impossible because its right-hand side is finite. Hence \(T(x_0)<\infty\), and letting \(t\uparrow T(x_0)\) gives \(T(x_0)\leq\frac{T_s}{\gamma_{p,q,r}}\int_0^{V_0}\frac{ds}{(\alpha s^p+\beta s^q)^r}\leq\frac{T_s}{\gamma_{p,q,r}}\int_0^\infty\frac{ds}{(\alpha s^p+\beta s^q)^r}=T_s\). At \(t=T(x_0)\), continuity gives \(V(x(T(x_0)))=0\); positive definiteness of \(V\) therefore implies \(x(T(x_0))=0\). Since the origin is an equilibrium, the trajectory remains there thereafter.
\end{proof}

\section{Data-Driven Construction of an Admissible Initial Policy for
Definition~\ref{def:admissible_secure_policy} and
Assumption~\ref{ass:plant_regularity}}
\label{app:koopman_warm_start}

\begingroup
\setstretch{0.80}

The replay matrix \(\bar{\Omega}\) in \eqref{eq:Omega_bar} presupposes a finite
bank of measured trajectory data, which may instead be deliberately collected
from the nominal plant before online learning without knowing the nonlinear
maps \(f\), \(g\), or \(k\). Here, ``unknown'' means that their analytical
expressions and parameters are unavailable, while the state/input dimensions
and measured state--input trajectories are known. From these data, a lifting
\(z=\psi(x)\), with \(\psi(0)=0\), can be chosen and a finite-dimensional
Koopman predictor \(\dot z=A_Kz+B_Ku+\varepsilon_K\) identified by
EDMD/least-squares techniques
\cite{Williams2015EDMD,Korda2018KoopmanMPC}. This linear controlled model is
only a finite-data approximation; for control-affine nonlinear systems, the
natural Koopman representation is generally bilinear and an exact linear
realization need not exist \cite{Bruder2021BilinearKoopman}.

The construction may also be refreshed online. Data collected over successive
windows \([t_j,t_{j+1})\) can update a local predictor
\(\dot z=A_{K,j}z+B_{K,j}u+\varepsilon_{K,j}\) and feedback
\(u_{c,0}^{(j)}(x)=K_{K,j}\psi(x)\), provided the admissibility conditions below
hold on each window. This receding-window feedback only supplies an admissible
policy that keeps the state in the compact operating set \(\Omega\); it does not
solve the optimal-control problem. Once this condition and the required finite
replay information are available, the proposed ADP--IRL mechanism improves the
policy toward the optimum within the prescribed predefined-time horizon.

For the warm start, assume that the lifting contains the original state,
\(x=C_xz\), and is locally norm-equivalent on \(\Omega\), i.e.,
\(c_1\|x\|\leq\|\psi(x)\|\leq c_2\|x\|\) for some \(c_1,c_2>0\). Assume further
that there exist \(K_K\), \(P=P^\top>0\), and \(Q=Q^\top>0\) such that
\((A_K+B_KK_K)^\top P+P(A_K+B_KK_K)=-Q\),
\(K_K\psi(x)\in\mathcal U\) for all \(x\in\Omega\), and
\(\|\varepsilon_K\|\leq c_K\|z\|\) with
\(2\|P\|c_K<\lambda_{\min}(Q)\). Then set
\(u_{c,0}(x)=K_K\psi(x)\). In the receding-window case, the same conditions are
imposed on each
\((A_{K,j},B_{K,j},K_{K,j},P_j,Q_j,\varepsilon_{K,j})\), ensuring that
\(u_{c,0}^{(j)}\) remains admissible on its window.

Indeed, with \(V_K=z^\top Pz\),
\(\dot V_K\leq-[\lambda_{\min}(Q)-2\|P\|c_K]\|z\|^2<0\) for \(z\neq0\).
Thus, \(z=0\) is asymptotically stable; since \(x=C_xz\) and the lifting is
norm-equivalent on \(\Omega\), stability of \(x=0\) follows directly. Together
with \(u_{c,0}(x)\in\mathcal U\), this gives
\(u_{c,0}\in\Psi(\Omega)\). The same argument holds window-wise, so the refreshed
data-driven feedback acts only as an admissible stabilizing safeguard in
\(\Omega\), while optimality and the prescribed convergence deadline are
provided by the predefined-time ADP--IRL mechanism.

\endgroup

\section{Proof lemma~\ref{lem:low_window}}
\label{uppbound}
\begin{proof}
Let \(M_y:=\max_{\tau\in[t-T,t]}y(\tau)\), and let \(\tau_m\) be a maximizing point. At least one of \([t-T,\tau_m]\) and \([\tau_m,t]\) has length no smaller than \(T/2\). Since \(y\) is \(L_{\dot y}\)-Lipschitz, \(y(\tau)\geq M_y/2\) on a one-sided interval of length \(\min\{T/2,M_y/(2L_{\dot y})\}\). Therefore, \(\int_{t-T}^{t}y(\tau)\,d\tau\geq (M_y/2)\min\{T/2,M_y/(2L_{\dot y})\}\).

If \(M_y\leq L_{\dot y}T\), the preceding inequality gives \(M_y\leq2\sqrt{L_{\dot y}\int_{t-T}^{t}y(\tau)\,d\tau}\). Since \(\gamma_1-1<0\) and \(0\leq y(\tau)\leq M_y\), one has \(y^{\gamma_1}(\tau)=y(\tau)y^{\gamma_1-1}(\tau)\geq y(\tau)M_y^{\gamma_1-1}\). Hence,
\[
\int_{t-T}^{t}y^{\gamma_1}(\tau)\,d\tau
\geq
2^{\gamma_1-1}
L_{\dot y}^{-\frac{1-\gamma_1}{2}}
\left(
\int_{t-T}^{t}y(\tau)\,d\tau
\right)^\mu.
\]

If \(M_y>L_{\dot y}T\), the same window estimate gives \(M_y\leq(4/T)\int_{t-T}^{t}y(\tau)\,d\tau\). Using \(y^{\gamma_1}(\tau)\geq y(\tau)M_y^{\gamma_1-1}\) then yields
\[
\int_{t-T}^{t}y^{\gamma_1}(\tau)\,d\tau
\geq
4^{\gamma_1-1}T^{1-\gamma_1}
\left(
\int_{t-T}^{t}y(\tau)\,d\tau
\right)^{\gamma_1}.
\]
Since \(\gamma_1-\mu=(\gamma_1-1)/2<0\) and \(\int_{t-T}^{t}y(\tau)\,d\tau\leq\bar I\), it follows that
\[
\left(
\int_{t-T}^{t}y(\tau)\,d\tau
\right)^{\gamma_1}
\geq
\bar I^{\frac{\gamma_1-1}{2}}
\left(
\int_{t-T}^{t}y(\tau)\,d\tau
\right)^\mu.
\]
Combining the two cases proves \eqref{eq:low_integral_bound}.
\end{proof}



\raggedbottom

\vspace{0.5em}
\section*{Author Biographies}
\vspace{-0.3em}


\noindent
\begin{minipage}[t]{0.23\columnwidth}
    \vspace{0pt}
    \centering
    \includegraphics[
        width=\linewidth,
        height=1.15in,
        keepaspectratio
    ]{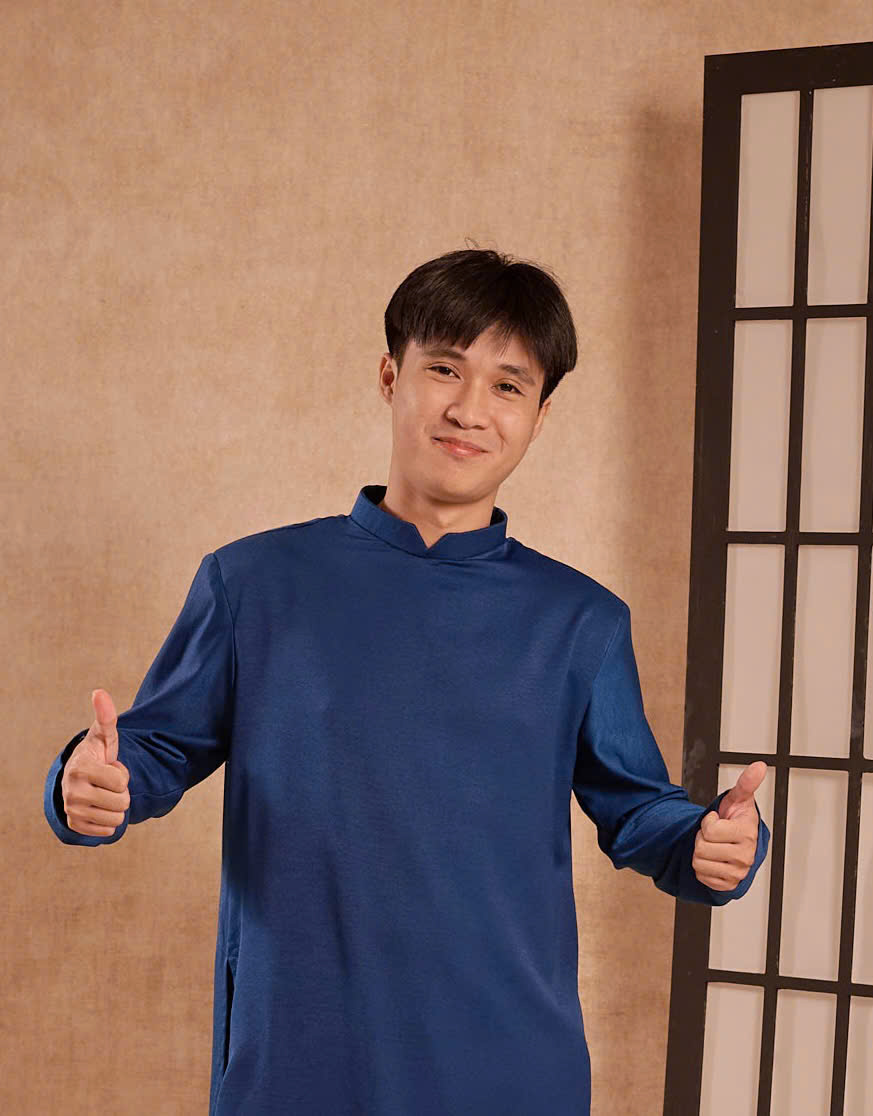}
\end{minipage}
\hfill
\begin{minipage}[t]{0.73\columnwidth}
    \vspace{0pt}
    \small
    \textbf{Tien Dat Vu} is currently a senior undergraduate student in the Vietnamese–French Program in Mechatronics Engineering at Ho Chi Minh City University of Technology (HCMUT), Vietnam National University Ho Chi Minh City (VNU-HCM), Ho Chi Minh City, Vietnam. His research interests include control theory, learning-based and nonlinear control, adaptive and robust control (Convex optimization in control), reinforcement learning, data-driven control, secure and resilient control, optimal control of uncertain nonlinear systems, multi-agent systems, and artificial intelligence for dynamical systems. His broader research interests also include Riemannian geometry, contraction theory, formal methods, and their applications to nonlinear dynamical systems and robotics.
\end{minipage}

\vspace{1.0em}


\noindent
\begin{minipage}[t]{0.23\columnwidth}
    \vspace{0pt}
    \centering
    \includegraphics[
        width=\linewidth,
        height=1.15in,
        keepaspectratio
    ]{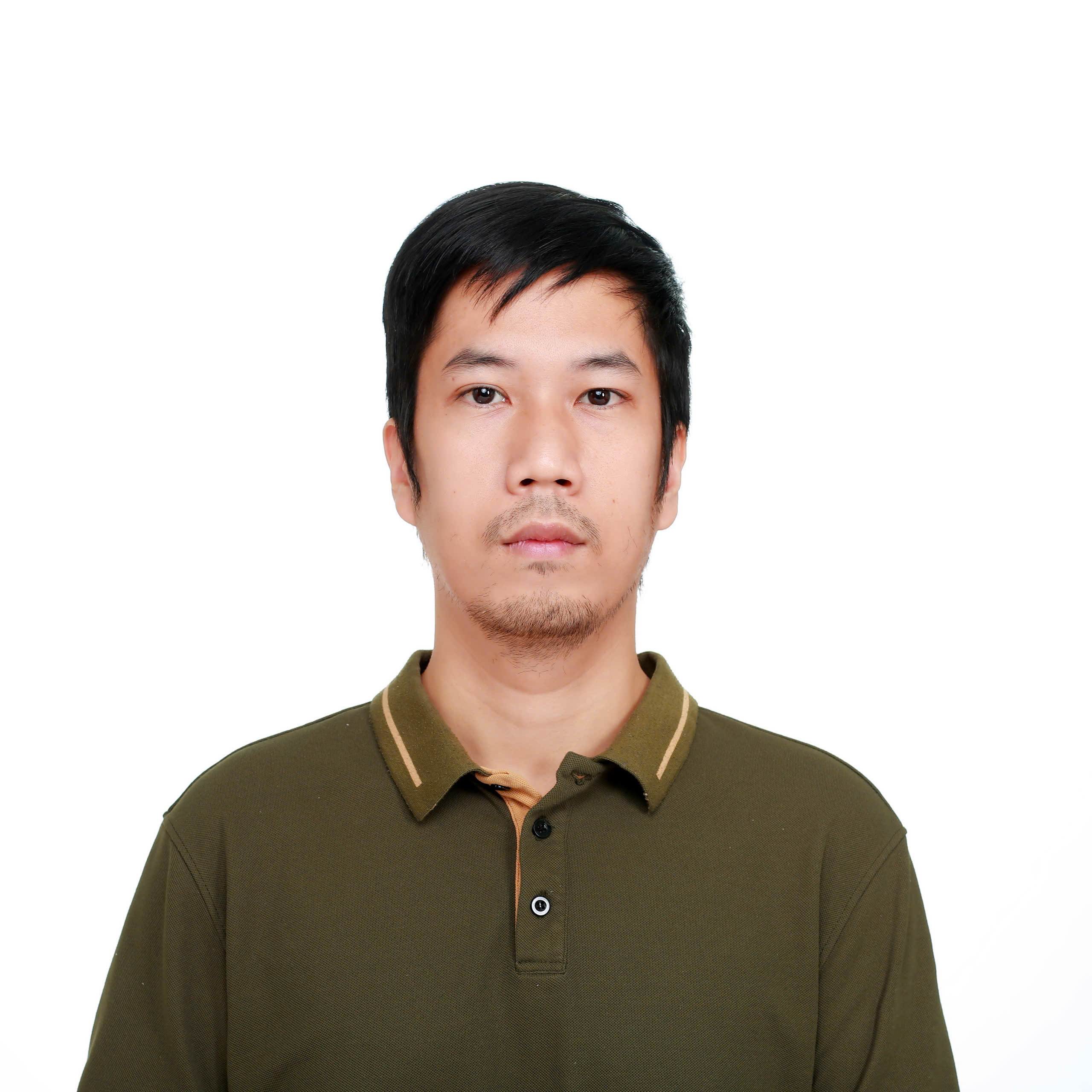}
\end{minipage}
\hfill
\begin{minipage}[t]{0.73\columnwidth}
    \vspace{0pt}
    \small
    \textbf{Nhat Minh Doan} received the B.S. degree in Mechanical
    Engineering from Bucknell University, Lewisburg, PA, USA,
    in 2015, and the Ph.D. degree from Keio University, Tokyo,
    Japan, in 2021. He is currently a Lecturer with the Faculty
    of Mechanical Engineering, Ho Chi Minh City University of
    Technology (HCMUT), Vietnam National University Ho Chi Minh
    City, Vietnam. His research interests include unmanned aerial
    vehicle design and control, wind turbine systems, multi-agent
    systems, networked control systems, and distributed control.
    His research also covers the modeling, analysis, and control
    of complex mechanical and networked dynamical systems.
\end{minipage}

\vspace{0.5em}

\end{document}